%% file: main-SSSP-and-APSP.tex
\newif\ifabstract
\abstracttrue
\newif\iffull
\ifabstract \fullfalse \else \fulltrue \fi

\documentclass[11pt]{article}

\input{header.tex}

\begin{document}

\title{Deterministic Algorithms for Decremental Shortest Paths \\via Layered	Core Decomposition}

\author{Julia Chuzhoy\thanks{Toyota Technological Institute at Chicago. Email: {\tt cjulia@ttic.edu}. Part of the work was done while the author was a Weston visiting professor at the Department of Computer Science and Applied Mathematics, Weizmann Institute. Supported in part by NSF grant CCF-1616584.}  \and Thatchaphol Saranurak\thanks{Toyota Technological Institute at Chicago. Email: {\tt saranurak@ttic.edu}.}
}
\date{}
\setcounter{tocdepth}{2}
\begin{titlepage}

\maketitle

\pagenumbering{gobble}
\input{abstract.tex}

\newpage
\tableofcontents{}
\end{titlepage}
\pagenumbering{arabic}

\input{input_det.tex}
\input{prelims.tex}

\input{LCD.tex}

\input{SSSP.tex}

\input{APSP_new.tex}

\appendix

\input{Appendix.tex}

\input{mincost.tex}

\input{mincost_MWU.tex}

\input{mincost_dynamic.tex}

\input{more_app.tex}

\input{table.tex}

\bibliographystyle{alpha}
\bibliography{SSSP-edge}

\end{document}

%% file: abstract.tex
\begin{abstract}
        In the decremental single-source shortest paths (SSSP) problem, the input is an undirected graph $G=(V,E)$ with $n$ vertices and $m$ edges undergoing
        edge deletions, together with a fixed source vertex $s\in V$. The goal is to maintain
        a data structure that supports \emph{shortest-path queries}: given a vertex $v\in V$,
        quickly return an (approximate) shortest path from $s$ to $v$.
        The decremental all-pairs shortest paths (APSP) problem is defined similarly,
        but now the shortest-path queries are allowed between  any pair  of vertices of $V$. 
        
        Both problems have been studied extensively since the 80's, and algorithms with near-optimal total update time and query time have been discovered for them. 
        Unfortunately, all these algorithms are randomized and, more importantly, they need to assume an \emph{oblivious adversary} -- a drawback that prevents them from being used as subroutines in several known algorithms for classical static problems.
        In this paper, we provide new \emph{deterministic} algorithms for both problems, which, by definition, can handle an adaptive adversary.
        
        Our first result is a deterministic algorithm for the decremental SSSP problem on \emph{weighted} graphs with $O(n^{2+o(1)})$ total update time, that supports $(1+\eps)$-approximate shortest-path queries, with query time $O(|P|\cdot n^{o(1)})$, where $P$ is the returned path. 
        This is the first $(1+\eps)$-approximation adaptive-update algorithm supporting shortest-path queries in time below $O(n)$, that breaks the $O(mn)$ total update time bound of the classical algorithm of Even and Shiloah from 1981.
        Previously, Bernstein and Chechik {[}STOC'16, ICALP'17{]} 
        provided a $\tilde{O}(n^{2})$-time deterministic algorithm that supports approximate \emph{distance} queries, but unfortunately the algorithm cannot return the approximate shortest paths. Chuzhoy and Khanna {[}STOC'19{]} showed an $O(n^{2+o(1)})$-time randomized algorithm for SSSP that supports approximate shortest-path queries in the adaptive adversary regime, but their algorithm only works in the restricted setting where only vertex deletions, and not edge deletions are allowed, and it requires $\Omega(n)$ time to respond to shortest-path queries. 
        
        Our second result is a deterministic algorithm for the decremental
        APSP problem on unweighted graphs that achieves total update time $O(n^{2.5+\delta})$, for any constant $\delta>0$, supports approximate distance queries in $O(\log\log n)$ time, and supports approximate shortest-path queries in time $O(|E(P)|\cdot n^{o(1)})$, where $P$ is the returned path; the algorithm achieves an $O(1)$-multiplicative and $n^{o(1)}$-additive approximation on the path length. All previous
        algorithms for APSP either assume an oblivious adversary or
        have an $\Omega(n^{3})$ total update time when $m=\Omega(n^{2})$, even if an $o(n)$-multiplicative approximation is allowed. 
        
        To obtain both our results, we improve and generalize the \emph{layered core decomposition}
        data structure introduced by Chuzhoy and Khanna to be nearly optimal in terms of various parameters, and introduce a new generic approach 
        of rooting Even-Shiloach trees at expander sub-graphs of the given graph. We believe both these technical tools to be interesting in their own right and anticipate them to be useful for designing future dynamic algorithms that work against an adaptive adversary. 
\end{abstract}

%% file: input_det.tex
\section{Introduction}

In the decremental single-source shortest path ($\SSSP$) problem,
the input is an undirected graph $G=(V,E)$ with $n$ vertices
and $m$ edges undergoing edge deletions, together with a fixed source vertex
$s\in V$. The goal is to maintain a data structure that supports \emph{shortest-path} queries: given a vertex $v\in V$, quickly return an (approximate) shortest
 path from $s$ to $v$. We also consider \emph{distance} queries: given a vertex $v\in V$, return an approximate distance
from $s$ to $v$. The decremental all-pairs shortest path ($\APSP$)
problem is defined similarly, but now the shortest-path and distance queries are allowed between any pair
$u,v\in V$ of vertices. A trivial algorithm for both problems is to simply maintain the current graph $G$, and, given a query between a pair $u,v$ of vertices, run a BFS from one of these vertices, to report the shortest path between $v$ and $u$ in time $O(m)$. Our goal therefore is to design an algorithm whose  \emph{query time} -- the time required to respond to a query --  is significantly lower than this trivial $O(m)$ bound, while keeping the \emph{total update time} -- the time 
needed for maintaining the data structure over the entire sequence of updates,
including the initialization --- as small as possible. Observe that the best query time for shortest-path queries one can hope for is $O(|E(P)|)$, where $P$ is the returned path\footnote{Even though in extreme cases, where the graph is very sparse and the path $P$ is very long, $O(|E(P)|)$ query time may be comparable to $O(m)$, for brevity, we will say that $O(|E(P)|)$ query time is below the $O(m)$ barrier, as is typically the case. For similar reasons, we will say that $O(|E(P)|)$ query time is below $O(n)$ query time.}.

Both $\SSSP$ and $\APSP$  are
among the most well-studied dynamic graph problems. 
While almost optimal algorithms are known for both of them, all such 
algorithms are randomized and, more importantly, they assume an \emph{oblivious
	adversary}. In other words, the sequence of edge deletions must be fixed in advance
and cannot depend on the algorithm's responses to queries. Much of
the recent work in the area of dynamic graphs has focused
on developing so-called \emph{adaptive-update algorithms}, that do not assume
an oblivious adversary (see e.g.~\cite{Saranurak,dynamic-spanning-forest,NanongkaiSW17,ChuzhoyGLNPS19}
for dynamic connectivity, \cite{BhattacharyaHI15,BhattacharyaHN16,BhattacharyaK19,Wajc20}
for dynamic matching, and \cite{BernsteinChechik,BernsteinChechikSparse,henzinger16,Bernstein,fast-vertex-sparsest,GutenbergW20,BernsteinBGNSS20}
for dynamic shortest paths); we also say that such algorithms work against an \emph{adaptive adversary}. One of the motivating reasons to consider adaptive-update algorithms is that
several algorithms for classical \emph{static} problems need to use,  as subroutines, dynamic graph algorithms that can handle adaptive adversaries
(see e.g.~\cite{SleatorT83,Madry10_stoc,fast-vertex-sparsest,ChekuriQ17}).
In this paper, we provide new \emph{deterministic} algorithms for
both $\SSSP$ and $\APSP$ which, by definition, can handle adaptive adversary. 

Throughout this paper, we use the $\tilde O$ notation to hide $\poly\log n$ factors, and $\Ohat$ notation to hide $n^{o(1)}$ factors, where $n$ is the number of vertices in the input graph. We also assume that $\epsilon>0$ is a small constant in the discussion below.

\paragraph{SSSP.}
The current understanding of decremental $\SSSP$ in the oblivious-adversary setting is almost complete, even for weighted graphs.
Forster, Henzinger, and Nanongkai \cite{HenzingerKN14_focs}, 
improving upon the previous work of  Bernstein and Roditty \cite{BernsteinR11} and Forster et al. \cite{HenzingerKN14_soda}, provided a $(1+\eps)$-approximation algorithm, with close to the best possible total update time of $\Ohat(m\log L)$, where $L$ is the ratio of largest to smallest edge length.
The query time of the algorithm is also near optimal: approximate distance queries can be processed 
 in $\tilde{O}(1)$ time, and approximate shortest-path queries in
  $\tilde{O}(|E(P)|)$ time, where $P$ is the returned path. Due to known conditional lower bounds
of $\Omegahat(mn)$ on the total update time for the exact version of $\SSSP$\footnote{The bounds assume the Boolean Matrix Multiplication (BMM) conjecture \cite{DorHZ00,RodittyZ11} or the Online Matrix-vector Multiplication 
(OMv) conjecture \cite{HenzingerKNS15}, and show that in order to achieve $O(n^{2-\eps})$ query time,  for any constant $\epsilon>0$, the total update time of $\Omega(n^{3-o(1)})$ is required in graphs with $m = \Theta(n^2)$.}
the guarantees provided by this algorithm are close to the best possible. Unfortunately, all these algorithms are randomized and need to assume
an oblivious adversary.

For adaptive algorithms, the progress has been slower. It is well known
that the classical algorithm of Even and Shiloach \cite{EvenS}, that we refer to as $\EST$ throughout this paper,
combined with the standard weight rounding technique (e.g.~\cite{Zwick98,bernstein16}) gives a
$(1+\eps)$-approximate deterministic algorithm for $\SSSP$ with $\Otil(mn\log L)$
total update time and near-optimal query time. This bound was first
improved by Bernstein \cite{Bernstein}, generalizing a similar result of \cite{BernsteinChechik} for unweighted graphs,
to $\tilde{O}(n^{2}\log L)$ total update time. 
For the setting of sparse unweighted graphs, Bernstein and Chechik \cite{BernsteinChechikSparse} designed an algorithm with total update time $\tilde{O}(n^{5/4}\sqrt{m}) \leq \tilde O(mn^{3/4})$, and Gutenberg and Wulff-Nielsen \cite{GutenbergW20} showed an algorithm with $\Ohat(m\sqrt{n})$ total update time.

Unfortunately, all of the above mentioned algorithms
only support distance queries, but they cannot handle shortest-path queries. Recently, Chuzhoy and Khanna \cite{fast-vertex-sparsest} attempted to
fix this drawback, and obtained a randomized $(1+\epsilon)$-approximation
\emph{adaptive-update} algorithm with total expected update time $\Ohat(n^{2}\log L)$, that supports shortest-path queries.  Unfortunately,
this algorithm has several other drawbacks. First, it is randomized.
Second, the expected query time of $\tilde{O}(n\log L)$ may be much higher than the desired time that is proportional to the number of edges on the returned path.
Lastly, and most importantly, the algorithm only works in the more restricted setting where only \emph{vertex deletions} are allowed, as opposed to the more standard and general model with edge deletions\footnote{We emphasize that the
vertex-decremental version is known to be strictly easier
than the edge-decremental version for some problems. For example, there is a vertex-decremental algorithm for maintaining
	the exact distance between a fixed pair $(s,t)$ of vertices in unweighted undirected graphs using $O(n^{2.932})$
	total update time \cite{APSPfully4} (later improved to $O(n^{2.724})$
	in \cite{BrandNS19}), but the edge-decremental version
	requires $\Omegahat(n^{3})$ time when $m=\Omega(n^{2})$ assuming the OMv conjecture \cite{HenzingerKNS15}.
	A similar separation holds for decremental exact APSP.}.
Finally, a very recent work by Bernstein
et al.~\cite{BernsteinBGNSS20}, that is concurrent to this paper, shows a $(1+\eps)$-approximate
algorithm with $\Ohat(m\sqrt{n})$ total update time that can return an approximate
shortest path $P$ in $\Otil(n)$ time (but not in time proportional
to $|E(P)|$). The algorithm is randomized but works against an adaptive adversary.

As mentioned already, algorithms for approximate decremental SSSP are often used as subroutines in algorithms for static graph problems, including various flow and cut problems that we discuss below. Typically, in these applications, the following properties are desired from the algorithm for decremental $\SSSP$:

\begin{itemize}
	\item it should work against an adaptive adversary, and ideally it should be deterministic;
	\item it should be able to handle edge deletions (as opposed to only vertex deletions);
	\item it should support shortest-path queries, and not just distance queries; and
	\item it should have query time for shortest-path queries that is close to $O(|E(P)|)$, where $P$ is the returned path.
\end{itemize}

In this paper we provide the {first} algorithm for decremental $\SSSP$ that satisfies all of the above requirements and improves upon the classical $\Omega(mn)$
bound of Even and Shiloach \cite{EvenS}.
The total update time of the algorithm is
$\Ohat(n^{2}\log L)$, which is {almost optimal
	for dense graphs}.
\begin{thm}
	[Weighted SSSP]\label{thm: main for SSSP}There is a deterministic
	algorithm, that, given a simple undirected edge-weighted $n$-vertex graph
	$G$ undergoing edge deletions, a source vertex $s$, and a parameter
	$\eps\in(1/n,1)$, maintains a data structure in total update time
	 $\Ohat(n^{2}(\frac{\log L}{\eps^{2}}))$, where $L$ is the ratio
	of largest to smallest edge weights, and supports the following queries: 
	\begin{itemize}
		\item $\distquery(s,v):$ in $O(\log\log(nL))$ time return an estimate
		$\apxdist(u,v)$, with $\dist_{G}(s,v)\le\apxdist(s,v)\le(1+\epsilon)\dist_{G}(s,v)$; and 
		\item $\pquery(s,v):$ either declare that $s$ and $v$ are not connected in $G$
		in $O(1)$ time, or return a $s$-$v$ path $P$ of length at most
		$(1+\epsilon)\dist_{G}(s,v)$, in time $\Ohat(|E(P)|\log\log L)$. 
	\end{itemize}
\end{thm}

Compared to the algorithm of \cite{Bernstein}, our deterministic algorithm supports shortest-path, and not just distance queries, while having the same total update time up to a subpolynomial factor. %
Compared to the algorithm of \cite{fast-vertex-sparsest}, our algorithm handles the more general setting of edge deletions, is deterministic, and has faster query time. 
Compared to the work of \cite{BernsteinBGNSS20} that is concurrent with this paper, our algorithm is deterministic and has a faster query time, though its total update time is somewhat slower for sparse graphs.

These improvements over previous works allow us to obtain faster algorithms for a number of classical static flow and cut problems; %
see \Cref{sec:mincost,sec:more_app} 
for more details.  
Most of the resulting algorithms are deterministic. 
For example, we obtain a deterministic algorithm for $(1+\epsilon)$-approximate
minimum cost flow in unit edge-capacity graphs in $\Ohat(n^{2})$
time. The previous algorithms by \cite{LeeS14,Axiotis2020circulation}
take time $\Otil(\min\{m\sqrt{n},m^{4/3}\})$, that is slower in dense graphs.

\paragraph{APSP.}

Our understanding of decremental
$\APSP$ is also almost complete in the oblivious-adversary setting, even in weighted graphs. Bernstein \cite{bernstein16}, 
improving upon the works of Baswana et al. \cite{BaswanaHS07} and Roditty and Zwick \cite{rodittyZ2}, 
obtained a $(1+\eps)$-approximation algorithm with $\tilde{O}(mn\log L)$ total update time,
$O(1)$ query time for distance queries, and $\Otil(|E(P)|)$ query time for shortest-path queries.\footnote{Bernstein's algorithm works even in directed graphs.}
These bounds are conditionally optimal for small approximation factors\footnote{Assuming the BMM conjecture \cite{DorHZ00,RodittyZ11} or the OMv conjecture \cite{HenzingerKNS15}, $1.99$-approximation
	algorithms for decremental APSP require $\Omegahat(n^{3})$ total update time or $\Omegahat(n^{2})$ query time in undirected
	unweighted graphs when $m=\Omega(n^2)$.}. Another line of work \cite{BernsteinR11,henzinger16,abraham2014fully,HenzingerKN14_focs}, focusing on larger approximation factors, recently culminated with a near-optimal result by Chechik~\cite{chechik}: for any integer
$k\ge1$, the algorithm of~\cite{chechik} provides a $((2+\epsilon)k-1)$-approximation,
with $\Ohat(mn^{1/k}\log L)$ total update time and $O(\log\log(nL))$
query time for distance queries and $\Otil(|E(P)|)$ query time for shortest-path queries. This result is near-optimal because all parameters almost
match the best static algorithm of Thorup and Zwick \cite{TZ}. Unfortunately, both algorithms
of Bernstein \cite{bernstein16} and of Chechik~\cite{chechik} need
to assume an oblivious adversary.

In contrast, our current understanding of adaptive-update algorithms is very poor even
for unweighted graphs. The classical $\EST$ algorithm of Even and Shiloach \cite{EvenS} implies a
deterministic algorithm for decremental exact $\APSP$ in unweighted
graphs with $O(mn^{2})$ total update time and optimal query time of $O(|E(P)|)$ where $P$ is the returned path.
This running time was first improved by Forster, Henzinger, and Nanongkai
\cite{henzinger16}, who showed a deterministic $(1+\epsilon)$-approximation
algorithm with $\tilde{O}(mn)$ total update time and $O(\log\log n)$
query time for distance queries. Recently, Gutenberg and Wulff-Nilsen \cite{GutenbergW20}
significantly simplified this algorithm. Despite a long line
of research, the state-of-the-art in terms of total update time remains $\tilde{O}(mn)$, which
can be as large as $\tilde \Theta(n^{3})$ in dense graphs, in any algorithm whose query time is below the  $O(n)$ bound. 
To highlight our lack of understanding of the problem, no adaptive algorithms
that attain an $o(n^{3})$ total update time and query time below $O(n)$ for shortest-path queries are currently known for any density regime, even
if we allow huge approximation factors, such as, for example, any
$o(n)$-approximation\footnote{When we allow a factor-$n$ approximation, one can use deterministic decremental
	connectivity algorithms (e.g.~\cite{dynamic-connectivity}) with
	$\tilde{O}(m)$ total update time and $O(\log n)$ query time for distance queries.}.

In this work, we break this barrier by providing the first \textbf{deterministic}
algorithm with \textbf{sub-cubic} total update time, that achieves
a \textbf{constant} multiplicative and a \textbf{subpolynomial} additive
approximation: 
\begin{thm}
	[Unweighted APSP]\label{thm: main for APSP}There is a deterministic
	algorithm, that, given a simple unweighted undirected $n$-vertex graph
	$G$ undergoing edge deletions and a parameter $1\le k\le o(\log^{1/8}n)$,
	maintains a data structure using total update time of $\Ohat(n^{2.5+2/k})$
	and supports the following queries: 
	\begin{itemize}
		\item $\distquery(u,v):$ in $O(\log n\log \log n)$ time return an estimate $\apxdist(u,v)$,
		where $\dist_{G}(u,v)\le\apxdist(u,v)\le 3\cdot 2^{k}\cdot\dist_{G}(u,v)+\Ohat(1)$; and
		\item $\pquery(u,v):$ either declare that $u$ and $v$ are not connected
		in $O(\log n)$ time, or return a $u$-$v$ path $P$ of length at most
		$3\cdot 2^{k}\cdot\dist_{G}(u,v)+\Ohat(1)$, in $\Ohat(|E(P)|)$ time. 
	\end{itemize}
	The additive approximation term in $\distquery$ and $\pquery$ is $\exp(O(k\log^{3/4}n))=\Ohat(1).$ 
\end{thm}

For example, by letting $k$ be a large enough constant, we can obtain a total update time of $\Ohat(n^{2.501})$, constant multiplicative approximation, and $\exp(O(\log^{3/4}n))$ additive approximation.

We note that the concurrent work of~\cite{BernsteinBGNSS20} on dynamic spanners that was mentioned above implies a randomized $\Otil(1)$-multiplicative  approximation
adaptive-update algorithm for APSP with $\Otil(m)$ total update time  but it requires a large $\Otil(n)$ query time even for distance queries; in contrast, our algorithm is deterministic and has faster query times: $\Ohat(|E(P)|)$ for shortest-path and $O(\log n\log\log n)$ for distance queries.

\paragraph{Technical Highlights.}
Both our algorithms for $\SSSP$ and $\APSP$ are based on the \emph{Layered Core Decomposition} ($\lcd$) data structure introduced by
Chuzhoy and Khanna \cite{fast-vertex-sparsest}. Informally, one may think of the data structure as maintaining a ``compressed'' version of the graph. Specifically, it maintains a decomposition of the current graph $G$ into a relatively small number of expanders (called cores), such that every vertex of $G$ either lies in one of the cores, or has a short path connecting it to one of the cores. The data structure supports approximate shortest-path queries within the cores, and queries that return, for every vertex of $G$, a short path connecting it to one of the cores. Chuzhoy and Khanna \cite{fast-vertex-sparsest} presented a randomized algorithm for maintaining the $\lcd$ data structure, as the graph $G$ undergoes {\bf vertex} deletions, with total update time $\Ohat(n^2)$.
As our first main technical contribution, we improve and generalize their algorithm in a number of ways: first, our algorithm is deterministic; second, it can handle the more general setting of edge deletions and not just vertex deletions; we improve the total update time to the near optimal bound of $\Ohat(m)$; and we improve the query times of this algorithm.
We further motivate this data structure and discuss the technical barriers that we needed to overcome in order to obtain these improvements
in \Cref{sec: LCD}. We believe that the $\lcd$ data structure is
of independent interest and will be useful in future adaptive-update dynamic algorithms. Indeed, a near-optimal \emph{short-path
	oracle on decremental expanders} (from \Cref{subsec:shortkpath}), which is one of the technical ingredients of our $\lcd$ data structure, has already found further applications in other algorithms for dynamic problems \cite{BernsteinGS20scc}. 

Our second main contribution is a new generic method to exploit the
Even-Shiloach tree ($\EST$) data structure\footnote{Here, we generally include variants such as the monotone $\EST$.}.
Many previous algorithms for $\SSSP$ and $\APSP$ \cite{BernsteinR11,HenzingerKN14_focs,henzinger16,chechik}
need to maintain a collection $\T$ of several $\ESTs$. One drawback of this approach, is that, whenever the root
of an $\EST$ is disconnected due to a sequence of edge deletions, we need to reinitialize a new $\EST$,
leading to high total update time. To overcome this difficulty, most such algorithms
choose the locations of the roots of the trees \emph{at random}, so that they are ``hidden'' from an oblivious adversary, and hence cannot be disconnected
too often. Clearly, this approach fails completely against an adaptive
adversary, that can repeatedly delete edges incident to the roots of the trees. 

In order to withstand an adaptive adversary, we introduce the idea of ``rooting an $\EST$
at an expander'' instead. As an expander is known to be robust against edge deletions even from an adaptive adversary \cite{Saranurak,NanongkaiSW17,expander-pruning}, the adversary cannot disconnect the ``expander root'' of the tree too often, leading to smaller total update time. 
The $\lcd$ data structure  naturally allows us to apply this high level idea, as it maintains a relatively small number of expander subgraphs (cores). This leads to our algorithm for $\APSP$  in the small distance regime. We also use this idea
to implement the short-path oracle on expanders. 
We believe that our general approach of ``rooting a tree at an expander'' instead
of ``rooting a tree at a random location'' will be a key technique
for future adaptive-update algorithms. This idea was already exploited in a different way in a recent subsequent work
\cite{BernsteinGS20scc}. %

\paragraph{Organization.}
We provide preliminaries in \Cref{sec: prelims}.  \Cref{sec: LCD} focuses on our main technical contribution: the new $\lcd$ data structure. We exploit this data structure to obtain our algorithms for $\SSSP$ and $\APSP$ in \Cref{sec: SSSP} and \Cref{sec: APSP}, respectively. %
The new cut/flow applications of our $\SSSP$ algorithm (that exploit known reductions) appear in \Cref{sec:mincost,sec:more_app}.

%% file: prelims.tex
\section{Preliminaries}\label{sec: prelims}
All graphs considered in this paper are undirected and simple, so they do not have parallel edges or self loops.
Given a graph $G$ and a vertex $v\in V(G)$, we denote by $\deg_G(v)$ the degree of $v$ in $G$. Given a length function $\ell:E(G)\rightarrow \reals$ on the edges of $G$, for a pair $u,v$ of vertices in $G$, we denote by $\dist_G(u,v)$ the length of the shortest path connecting $u$ to $v$ in $G$, with respect to the edge lengths $\ell(e)$. As the graph $G$ undergoes edge deletions, the notation $\deg_G(v)$ and $\dist_G(u,v)$ always refer to the current graph $G$. For a path $P$ in $G$, we denote $|P|=|E(P)|$.

Given a graph $G$ and a subset $S$ of its vertices, let $G[S]$ be the subgraph of $G$ induced by $S$. We denote by $\delta_G(S)$ the total number of edges of $G$ with exactly one endpoint in set $S$, and we let $E_G(S)$ be the set of all edges of $G$ with both endpoints in $S$. 
Given two subsets $A,B$ of vertices of $G$, we let $E_G(A,B)$ denote the set of all edges with one endpoint in $A$ and another in $B$.
The \emph{volume} of a vertex set $S$ is $\vol_G(S)=\sum_{v\in S}\deg_G(v)$. 
If $S$ is a set of vertices with $1\leq |S|<|V(G)|$, then we may refer to $S$ as a \emph{cut}, and we denote $\overline S=V(G)\setminus S$. We let the \emph{conductance} of the cut $S$ be $\Phi_G(S)=\frac{\delta_G(S)}{\min\set{\vol_G(S),\vol_G(\overline S)}}$.
We may omit the subscript $G$ when clear from context. We denote $\vol(G)=\sum_{v\in V(G)}\deg_G(v)=2|E(G)|$.
Given a graph $G$, we let its conductance $\Phi(G)$ be the minimum, over all cuts $S$, of $\Phi_G(S)$. Notice that $0\leq \Phi(G)\leq 1$ always holds. We say that graph $G$ is a \emph{$\phi$-expander} iff $\Phi(G)\geq \phi$.

Suppose we are given a graph $G$ and a sub-graph $G'\subseteq G$. We say that $G'$ is a \emph{strong $\phi$-expander with respect to $G$} iff for every partition $(S,\nots)$ of $V(G')$ into non-empty subsets, $\frac{\delta_{G'}(S)}{\min\set{\vol_G(S),\vol_G(\nots)}}\geq \phi$ (note that in the denominator, the volumes of the sets $S,\nots$ of vertices are taken in graph $G$, not in $G'$ as in the definition of $\phi$-expansion of $G'$). It is easy to verify that, if $G'$ is a strong $\phi$-expander with respect to $G$, then it is also a $\phi$-expander.
The following two simple observations follow from the definition of a strong $\phi$-expander.

\begin{observation}\label{obs: high degrees in strong expander}
	Let $G$ be a graph such that for all $v\in V(G)$, $\deg_G(v)\geq h$ for some $h>0$, and let $G'\subseteq G$ be a strong $\phi$-expander with respect to $G$, for some $0<\phi<1$, such that $|V(G')|\geq 2$. Then, for every vertex $v\in V(G')$, $\deg_{G'}(v)\geq \phi h$.
\end{observation}
\begin{proof}
	Assume otherwise, and let $v\in V(G')$ be any vertex with $\deg_{G'}(v)< \phi h$. Consider the cut $(S,\nots)$ of $V(G')$, where $S=\set{v}$, and $\nots=V(G')\setminus \set{v}$. Then $\vol_G(S)\geq h$, $\vol_G(\nots)\geq h$, but $\delta(S)=\deg_{G'}(v)<\phi h$, contradicting the fact that $G'$ is a strong $\phi$-expander with respect to $G$.
\end{proof}

\begin{observation}\label{obs: high degrees in strong expander v2}
	Let $G$ be a graph and let $G'$ be a sub-graph of $G$ containing at least two vertices, such that $G'$ is a strong $\phi$-expander with respect to $G$, for some $0<\phi<1$. Then, for every vertex $v\in V(G')$ with $\deg_G(v)\le \vol_G(V(G'))/2$, $\deg_{G'}(v)\geq \phi \deg_{G}(v)$ must hold.
\end{observation}

\begin{proof}
Consider the cut $(\{v\},V(G')\setminus\{v\})$ in $G'$. Then $\frac{\deg_{G'}(v)}{\deg_{G}(v)} = \frac{\delta_{G'}(\{v\})}{\min\{\vol_{G}(\{v\}),\vol_{G}(V(G')\setminus\{v\})\}} \ge \phi$ must hold, as $G'$ is a strong $\phi$-expander with respect to $G$. Therefore,  $\deg_{G'}(v) \ge \phi \deg_{G}(v)$.
\end{proof}

Given a graph $G$, its {\em $k$-orientation} is an assignment of a direction to each undirected edge of $G$, so that each vertex of $G$ has  out-degree at most $k$. For a given orientation of the edges, for each vertex $u\in V(G)$, we denote by $\indeg_G(u)$ and $\outdeg_G(u)$ the in-degree and out-degree of $u$, respectively. Note that, if $G$ has a $k$-orientation, then for  every subset $S\subseteq V$ of its vertices, $|E_G(S)| \le k\cdot|S|$ must hold, and, in particular, $|E(G)|\le k\cdot |V(G)|$. We say that a set  $F\subseteq E(G)$ of edges has a $k$-orientation if the graph induced by $F$ has a $k$-orientation.

\paragraph{Decremental Connectivity/Spanning Forest.}
We use the results of~\cite{dynamic-connectivity}, who provide a deterministic data structure, that we denote by $\CONNSF(G)$, that, given an $n$-vertex unweighted undirected graph $G$, that is subject to edge deletions, maintains a spanning forest of $G$, with total update time $O((m+n)\log^2n)$, where $m$ is the number of edges in the initial graph $G$. Moreover, the data structure supports connectivity queries $\conn(G,u,v)$: given a pair  $u,v$ of vertices of $G$, return ``yes'' if $u$ and $v$ are connected in $G$, and ``no'' otherwise. The running time to respond to each such query is  $O(\log n/\log\log n)$. 

\paragraph{Even-Shiloach Trees.}
Suppose we are given a graph $G=(V,E)$ with integral lengths $\ell(e)\geq 1$ on its edges $e\in E$, a source $s$, and a distance bound $D\geq 1$. Even-Shiloach Tree (\EST) algorithm maintains a shortest-path tree from vertex $s$, that includes all vertices $v$ with $\dist(s,v)\leq D$, and, for every vertex $v$ with $\dist(s,v)\leq D$, the distance $\dist(s,v)$. Typically, \EST only supports edge deletions (see, e.g. \cite{EvenS,Dinitz,HenzingerKing}). However, as shown in \cite[Lemma 2.4]{BernsteinChechik}, it is easy to extend the data structure to also handle edge insertions in the following two cases: either (i) at least one of the endpoints of the inserted edge is a singleton vertex, or (ii) the distances from the source $s$ to other vertices do not decrease due to the insertion. We denote the corresponding data structure from \cite{BernsteinChechik} by $\EST(G,s,D)$. It was shown in \cite{BernsteinChechik}  that the total update time of $\EST(G,s,D)$, including the initialization and all edge deletions, is $O(m D + U)$, where $U$ is the total number of updates (edge insertions or deletions), and $m$ is the total number of edges that ever appear in $G$.

\paragraph{Greedy Degree Pruning.}

We consider a simple degree pruning procedure defined in~\cite{fast-vertex-sparsest}.
Given a graph $H$ and a degree
bound $d$, the procedure computes a vertex set $A\subseteq V(H)$, as follows. Start with $A=V(H)$. While there is a
vertex $v\in A$, such that fewer than $d$ neighbors of $v$ lie
in $A$, remove $v$ from $A$. We denote
this procedure by $\DSP(H,d)$ and denote by $A$ the output of the procedure. 
The following observation was implicitly shown in \cite{fast-vertex-sparsest}; for completeness, we provide its proof in Appendix.

\begin{observation}\label{obs:DSP}
	Let $A$ be the outcome of procedure $\DSP(H,d)$, for any graph $H$ and integer $d$. Then $A$ is the \emph{unique
		maximal} vertex set such that every vertex in $H[A]$ has degree at least
	$d$. That is, for any subset $A'$ of $V(H)$ where $H[A']$ has
	minimum degree at least $d$, $A'\subseteq A$ must hold. 
	\end{observation}

Consider now a graph $H$ that undergoes edge deletions, and let $A$ denote the outcome of procedure $\DSP(H,d)$ when applied to the current graph. Notice that, from the above observation, set $A$ is a \emph{decremental vertex set}, that is, vertices can only leave the set, as edges are deleted from $H$. We use the following algorithm, that we call $\algPS(H,d)$, that allows us to maintain the set $A$ as the graph $H$ undergoes edge deletions; the algorithm is implicit in \cite{fast-vertex-sparsest}.

The algorithm $\algPS(H,d)$ starts by running $\DSP(H,d)$ on the original graph $H$. Recall that the procedure initializes $A=V(H)$, and then iteratively deletes from $A$ vertices $v$ that have fewer than $d$ neighbors in $A$. In the remainder of the algorithm, we simply maintain the degree of every
vertex in $H[A]$ as $H$ undergoes edge deletions. Whenever, for any vertex $v$, $\deg_{H[A]}(v)$ falls below $d$, we remove $v$ from $A$.  Observe that vertex degrees in $H[A]$ are monotonically decreasing.
Moreover, each degree decrement at a vertex $v$ can be charged to an edge that is incident to $v$ and was deleted from $H[A]$. As each edge is charged at most twice,
the total update time is $O(|E(H)|+|V(H)|)$.
Therefore, we obtain the following immediate observation.

\begin{observation}
	\label{obs: maintain DSP}The total update time of $\algPS$ is $O(m+|V(H)|)$, where $m$ is the number of edges that belonged to graph $H$ at the beginning. Moreover, whenever the algorithm removes some vertex $v$ from set $A$, vertex $v$ has fewer than $d$ neighbors in $A$ in the current graph $H$.
\end{observation}

%% file: LCD.tex
\section{Layered Core Decomposition}

\label{sec: LCD}

Our main technical contribution is a data structure called \emph{Layered
	Core Decomposition} ($\lcd$), that improves and generalizes the data structure introduced in  \cite{fast-vertex-sparsest}. In order to define the data structure, we need to introduce the notions of \emph{virtual vertex degrees}, and a partition of vertices into \emph{layers}, which we do next.

Suppose we are given an $n$-vertex $m$-edge graph $G=(V,E)$ and a
parameter $\Delta > 1$. We emphasize that throughout this section, the input graph $G$ is unweighted, and the length of a path $P$ in $G$ is the number of its edges.
Let $\dmax$ be the largest vertex degree
in $G$. Let $r$ be the smallest integer, such that $\Delta^{r-1}>\dmax$.
Note that $r\leq O(\log_{\Delta}n)$. Next, we define degree thresholds
$h_{1},h_{2},\ldots,h_{r}$, as follows: $h_{j}=\Delta^{r-j}$. Therefore,
$h_{1}>\dmax$, $h_{r}=1$, and for all $1<j\leq r$, $h_{j}=h_{j-1}/\Delta$.
For convenience, we also denote $h_{r+1}=0$. 
\begin{defn}
	[Virtual Vertex Degrees and Layers]For all $1\le j\le r$, let $A_{j}$ be the outcome of  $\DP(G,h_{j})$, when applied to the current graph $G$.
	The \emph{virtual degree} $\td(v)$ of $v$ in $G$ is the largest
	value $h_{j}$ such that $v\in A_{j}$. If no such value exists, then
	$\td(v)=h_{r+1}=0$. For all $1\leq j\leq r+1$, let $\Lambda_{j}=\{v\mid\td(v)=h_{j}\}$ denote
	the set of vertices whose virtual degree is $h_{j}$. We call $\Lambda_{j}$
	the \emph{$j$th layer}. 
\end{defn}

Note that for every vertex $v\in V(G)$, $\td(v)\in\{h_{1},\ldots,h_{r+1}\}$. Also, $\Lambda_{1}=\emptyset$
since all vertex degrees are below $h_{1}$, and $\Lambda_{r+1}$,
the set of vertices with virtual degree $0$, contains all isolated vertices.
For all $1\leq j'<j\leq r+1$, we say that layer $\Lambda_{j'}$ is \emph{above} layer $\Lambda_{j}$. For convenience,
we write $\Lambda_{\le j}=\bigcup_{j'=1}^{j}\Lambda_{j'}$ and $\Lambda_{<j},\Lambda_{\ge j},\Lambda_{>j}$
are defined similarly. Notice that $\Lambda_{\le j}=A_{j}$. For any
vertex $u$, let $\deg_{\le j}(u)=|E_{G}(u,\Lambda_{\le j})|$ denote
the number of neighbors of $u$ that lie in layer $j$ or above.

Intuitively, the partition of $V(G)$ into layers is useful because, in a sense, we can tightly control the degrees of vertices in each layer. This is summarized more formally in the following three observations. %
The first observation, that follows immediately from \Cref{obs:DSP}, shows that every vertex in layer $\Lambda_j$ has many neighbors in layer $j$ and above:

\begin{observation}
	\label{obs: virt vs reg degrees connections} 
	
	Throughout the algorithm, for each  $1\leq j\leq r+1$, for each vertex $u\in\Lambda_{j}$,
	$\deg_{\le j}(u)\ge h_{j}$. Therefore, the minimum vertex degree in $G[\Lambda_{\le j}]$
	is always at least $h_{j}$. 
\end{observation}

As observed already, from \Cref{obs:DSP}, over the course of the algorithm, vertices may only be deleted from $\Lambda_{\le j} = A_j$. This immediately implies the following observation:
\begin{observation}
	\label{obs: virtual degrees monotone} As edges are deleted from $G$,
	for every vertex $v$, $\td(v)$ may only decrease. 
\end{observation}

Throughout, we denote by $n_{\le j}$ the number of vertices
that belonged to $\Lambda_{\le j}$ at the beginning of the algorithm, before any edges were deleted from the input graph. Observe that $n_{\le j} h_j \le m$ by \Cref{obs: virt vs reg degrees connections}.
The proof of the following observation appears in Appendix.

\begin{observation}
	\label{claim:bound edge} %
	For all $1\leq j\leq r$, let $E_{\ge j}$ be the set of all edges, such that at any point of time
	at least one endpoint of $e$ lied in $\Lambda_{\ge j}$. Then
	$E_{\ge j}$ has a $(\Delta h_{j})$-orientation, and so $|E_{\ge j}|\le\Delta h_{j}n$.
	Moreover, the total number of edges $e$, such that, at any point of the algorithm's execution, both endpoints of $e$ lied in 
	 $\Lambda_{j}$, is bounded by $n_{\le j}h_{j}\Delta$. 
	
\end{observation}

From \Cref{obs: virt vs reg degrees connections}, all vertex degrees in $G[\Lambda_{\leq j}]$ are at least $h_j$, so, in a sense, graph $G[\Lambda_{\leq j}]$ is a high-degree graph.
One advantage of high-degree graphs is that every pair of vertices lying in the same connected component of such a graph must have a short path connecting them; specifically, it is not hard to show that, if $u,v$ are two vertices lying in the same connected component $C$ of graph $G[\Lambda_{\leq j}]$, then there is a path connecting them in $C$, of length at most $O(|V(C)|/h_j)$. This property of graphs $G[\Lambda_{\leq j}]$ is crucial to our algorithms for \SSSP and \APSP, and one of the goals of the \lcd data structure is to support \emph{short-path} queries: given a pair of vertices $u,v\in \Lambda_{\leq j}$, either report that they lie in different connected components of $G[\Lambda_{\leq j}]$, or return a path of length at most roughly  $O(|V(C)|/h_j)$ connecting them, where $C$ is the connected component of $G[\Lambda_{\leq j]}$ containing $u$ and $v$.
 Additionally, one can show that a high-degree graph must contain a \emph{core decomposition}. Specifically, suppose we are given a simple $n$-vertex graph $H$, with minimum vertex degree at least $h$. Intuitively, a \emph{core} of $H$ is a vertex-induced sub-graph $K\subseteq H$, such that, for $\phi=\Omega(1/\log n)$, graph $K$ is a $\phi$-expander, and all vertex degrees in $K$ are at least $\phi h/3$. One can show that, if $K$ is a core, then its diameter is $O(\log n/\phi)$, and it is $(\phi h/3)$-edge-connected. A \emph{core decomposition} of $H$ is a collection $\fset=\{K_{1},\dots,K_{t}\}$ of vertex-disjoint cores, such that, for each vertex $u\notin\bigcup_{K\in\fset}V(K)$, there are at least
$2h/3$ edge-disjoint paths of length $O(\log n)$ from $u$ to vertices
in $\bigcup_{K\in\fset}V(K)$. The results of \cite{fast-vertex-sparsest} implicitly show the existence of a core decomposition in a high-degree graph, albeit with a much more complicated  definition of the cores and of the decomposition. For completeness,
in Section \ref{sec: expanding core decomposition existence} of the Appendix, we formally state and prove a theorem about the existence of a core decomposition in a high-degree graph. Though we do not need this theorem for the results of this paper, we feel that it is an interesting graph theoretic statement in its own right, that in a way motivates the \lcd data structure, whose intuitive goal is to maintain a layered analogue of the core decomposition of the input graph $G$, as it undergoes edge deletions.

Formally, the \lcd data structure receives as input an (unweighted) graph $G$ undergoing edge deletions, and two parameters $\Delta \geq 2$ and $1 \le q \le o(\log^{1/4}n)$. It maintains the partition  of $V(G)$ into layers $\Lambda_{1},\dots,\Lambda_{r+1}$, as described above, and additionally, for each layer $\Lambda_{j}$, the data structure maintains a
collection $\fset_{j}$ of vertex-disjoint subgraphs of the graph
$H_{j}=G[\Lambda_{j}]$, called \emph{cores} (while we do not formally have any requirements from the cores, e.g. we do not formally require that a core is an expander, our algorithm will in fact still ensure that this is the case, so the intuitive description of the cores given above matches what our algorithm actually does). Throughout, we use an additional parameter $\gamma(n)=\exp(O(\log^{3/4}n))=\Ohat(1)$. The data structure is required to support the following three types of queries:

\begin{itemize}
	\item $\shortpath(j,u,v)$: Given any pair of vertices $u$ and $v$ from
	$\Lambda_{\le j}$, either report that $u$ and $v$ lie in different
	connected component of $G[\Lambda_{\le j}]$,
	or return a simple path $P$ connecting $u$ to $v$ in $G[\Lambda_{\leq j}]$ of length $O(|V(C)|(\gamma(n))^{O(q)}/h_{j}) = \Ohat(|V(C)|/h_j)$,  where $C$ is the connected component
	of $G[\Lambda_{\le j}]$ containing $u$ and $v$. 

	\item $\tocore(u)$: Given any vertex $u$, return a simple path $P=(u=u_{1},\dots,u_{z}=v)$
	of length $O(\log^{3}n)$ from $u$ to a vertex $v$ that lies in some core in $\bigcup_j\fset_j$. Moreover, path $P$ must visit the layers in a non-decreasing order, that is, if
	$u_{i}\in\Lambda_{j}$ then $u_{i+1}\in\Lambda_{\le j}$. 
	
	\item $\shortkpath(K,u,v)$: Given any pair of vertices $u$ and $v$, both of which lie in some core $K\in \bigcup_j\fset_j$, return a  simple $u$-$v$ path $P$ in $K$ of
	length at most  $(\gamma(n))^{O(q)} = \Ohat(1)$. 
\end{itemize}

We now formally state one of our main technical results - an algorithm for maintaining the \lcd data structure under edge deletions.

\begin{thm}
	[Layered Core Decomposition]\label{thm:LCD}There is a deterministic
	algorithm that, given a simple unweighted $n$-vertex $m$-edge graph $G=(V,E)$
	undergoing edge deletions, and parameters $\Delta\ge2$ and $1 \le q \le o(\log^{1/4}n)$,
	maintains a partition $(\Lambda_{1},\dots,\Lambda_{r+1})$ of $V$ into layers, 
where for all $1\leq j\leq r+1$, each vertex in $\Lambda_{j}$ has virtual degree $h_{j}$. Additionally, for each layer $\Lambda_{j}$, the algorithm maintains a
collection $\fset_{j}$ of vertex-disjoint subgraphs of the graph
$H_{j}=G[\Lambda_{j}]$, called \emph{cores}. 
The algorithm supports queries  $\shortpath(j,u,v)$ in time $O(\log n)$ if $u$ and $v$ lie in different connected components of $G[\Lambda_{\le j}]$, and in time $O(|P|(\gamma(n))^{O(q)}) = \Ohat(|P|)$ otherwise, where $P$ is the $u$-$v$ path returned. Additionally, it supports queries  $\tocore(u)$ with query time $O(|P|)$, where $P$ is the returned path, and $\shortkpath(K,u,v)$ with query time  $(\gamma(n))^{O(q)} = \Ohat(1)$. 
For all $1\leq j\leq h+1$, once a core $K$ is added to $\fset_j$ for the first time, it only undergoes edge- and vertex-deletions, until $K=\emptyset$ holds. The total number
	of cores ever added to $\fset_{j}$ throughout the algorithm is
	at most $\Ohat(n\Delta/h_{j})$.  The total update time of the algorithm is $\Ohat(m^{1+1/q}\Delta^{2+1/q}(\gamma(n))^{O(q)}) = \Ohat(m^{1+1/q}\Delta^{2+1/q})$.
\end{thm}

For intuition, it is convenient to set the parameters $\Delta=2$ and
$q=\log^{1/8}n$, which is also the setting that we use in algorithms for $\SSSP$ and for $\APSP$ in the large-distance regime. For this setting, $(\gamma(n))^{O(q)} = \Ohat(1)$, and the total update time of the algorithm is $\Ohat(m)$.

\noindent {\bf Optimality.}
The guarantees of the $\lcd$ data structure from \Cref{thm:LCD} are close to optimal in several respects.
First, the total update time of $\Ohat(m)$ and the query time for $\shortkpath$
and $\tocore$ are clearly optimal to within a subpolynomial in $n$ factor. The
length of the path returned by $\shortpath$ queries is almost optimal in the
 sense that there can exist a path $P$ in a connected component $C$ of
$G[\Lambda_{\le j}]$ whose length is $\Omega(|V(C)|/h_{j})$; 
the query time of $\Ohat(|P|)$ is almost optimal as well. The bound on the
total number of cores ever created in $\Lambda_{j}$ is also near
optimal. This is because, even in the static setting, there exist graphs 
with minimum degree $h_{j}$ that require $\tilde{\Omega}(n/h_{j})$
cores in order to guarantee the desired properties of a core decomposition.

\subsection*{Comparison with the Algorithm of \cite{fast-vertex-sparsest} and Summary of Main Challenges}
The data structure from  \cite{fast-vertex-sparsest}  supports the same set of queries, but has several significant drawbacks compared to the results of \Cref{thm:LCD}. First, the algorithm of \cite{fast-vertex-sparsest} is randomized. Moreover, it can only handle vertex deletions, as opposed to the more general and classical setting of edge deletions (which is also required in some applications to static flow and cut problems). Additionally, the total update time of the algorithm of \cite{fast-vertex-sparsest}  is $\Ohat(n^{2})$,
as opposed to the almost linear running time of $\Ohat(m)$ of our algorithm. For every index $j$, the total number
of cores ever created in $\Lambda_{j}$ can be as large as $\Ohat(n^{2}/h_{j}^{2})$ in the algorithm of \cite{fast-vertex-sparsest}, as opposed to the bound of  $\Ohat(n/h_{j})$ that we obtain; this bound directly affects the running of our algorithm for $\APSP$. Lastly, the query time for $\shortpath(j,u,v)$ 
is only guaranteed to be bounded by $\Ohat(|V(C)|)$ in \cite{fast-vertex-sparsest}, where $C$ is a connected component of $\Lambda_{\leq j}$ to which $u$ and $v$ belong, as opposed to our query time of $\Ohat(|P|)$,  where $P$ is the $u$-$v$ path returned. This faster query time is essential in order to obtain the desired query time of $\Ohat(|P|)$  in our algorithms for \SSSP and \APSP.
Next, we describe some of the challenges to achieving these improvements, and also sketch some ideas that allowed us to overcome them.

\noindent {\bf Vertex deletions versus edge deletions.}
The algorithm of \cite{fast-vertex-sparsest} maintains, for every index $1\leq j\leq r$, a variation of the core decomposition (that is based on vertex expansion) in graph 
$G[\Lambda_{j}]$. This decomposition can be computed in almost linear time $\Ohat(|E(\Lambda_{j})|)=\Ohat(nh_{j})$, which is close to the best time one can hope for, 
creating an initial set $\fset_j$ of at most $\Ohat(n/h_{j})$ cores. Since every core $K\in \fset_j$ has vertex degrees at least
$h_{j}/n^{o(1)}$, the decomposition can withstand up to $h_{j}/(2n^{o(1)})$ vertex deletions, while maintaining all its crucial properties. However, after $h_{j}/(2n^{o(1)})$
vertex deletions, some cores may become disconnected, and the core decomposition structure may no longer retain the desired properties. 
Therefore, after every batch of roughly $h_{j}/(2n^{o(1)})$ vertex deletions, the algorithm of \cite{fast-vertex-sparsest} recomputes the core
decomposition $\fset_j$ from scratch. Since there may be at most $n$ vertex-deletion operations
throughout the algorithm, the core decomposition $\fset_j$ only needs to be recomputed 
at most $\Ohat(n/h_{j})$ times throughout the algorithm, leading to the total update time of $\Ohat(n/h_{j})\cdot \Ohat(|E(\Lambda_{j})|)=\Ohat(n^{2})$.
The total number of cores that are ever added to $\fset_j$ over the course of the algorithm is then bounded by $\Ohat(n/h_{j})\cdot \Ohat(n/h_{j})=\Ohat(n^{2}/h_{j}^{2})$. 

Consider now the edge-deletion setting. Even if we are willing to allow a total update time of $\Ohat(n^{2})$, we cannot hope to perform a single computation of the decomposition $\fset_j$ in time faster than linear in $|E(\Lambda_j)|$, that is, $O(nh_j)$. Therefore, we can only afford at most $O(n/h_j)$ such re-computations over the course of the algorithm. Since the total number of edges in graph $G[\Lambda_j]$ may be as large as $\Theta(nh_j)$, our core decomposition must be able to withstand up to $h^2_j$ edge deletions. However, even after just $h_j$ edge deletions, some vertices of $\Lambda_j$ may become disconnected in graph $G[\Lambda_{\leq j}]$, and some of the cores may become disconnected as well. In order to overcome his difficulty, we 
first observe that it takes $h_{j}/n^{o(1)}$ edge deletions before a vertex in $\Lambda_{j}$
becomes ``useless'', which roughly means that it is not well-connected to other vertices in $\Lambda_{j}$. Similarly to the algorithm of \cite{fast-vertex-sparsest}, we would now like to recompute the core decomposition $\fset_j$ only after $h_{j}/(2n^{o(1)})$ vertices of $\Lambda_j$ become useless, which roughly corresponds to $h_j^2/n^{o(1)}$ edge deletions. Additionally, we employ the expander pruning technique from \cite{expander-pruning} in order to maintain the cores so that they can withstand this significant number of edge deletions.
As in \cite{fast-vertex-sparsest}, this approach can lead to $\Ohat(n^{2})$ total update time, ensuring that the total number of cores that are ever added to set $\fset_j$ is at most  $\Ohat(n^{2}/h_{j}^{2})$.

\noindent {\bf Obtaining faster total update time and fewer cores.}
Even with the modifications described above, the resulting total update time is only  $\Ohat(n^{2})$, while our desired update time is near-linear in $m$. It is not hard to see that recomputing the whole decomposition $\fset_j$ from scratch every time is too expensive, and with the $\Ohat(m)$ total update time we may only afford to do so at most $\Ohat(1)$ times. In order to overcome this difficulty, we further partition each layer $\Lambda_{j}$
into \emph{sublayers} $\Lambda_{j,1},\Lambda_{j,2},\dots,\Lambda_{j,L_{j}}$
whose size is geometrically decreasing (that is,~$|\Lambda_{j,\ell}|\approx |\Lambda_{j,\ell-1}|/2$
for all $\ell$). The core decompositions $\fset_{j,\ell}$ will be computed in each sub-layer separately, and the final core decomposition for layer $j$ that the algorithm maintains is $\fset_j=\bigcup_{\ell}\fset_{j,\ell}$. 
In general, we guarantee that, for each $\ell$, $|\Lambda_{j,\ell}| \le n_{\le j}/2^{\ell-1}$ always holds, and we recompute the core decomposition $\fset_{j,\ell}$ for sublayer at $\Lambda_{j,\ell}$ at most $\Ohat(2^\ell)$ times. We use \Cref{claim:bound edge} to show that $|E(\Lambda_{j,\ell})| \le h_j \Delta \cdot n_{\le j}/2^{\ell-1} = O(m/2^\ell)$ must hold.
Therefore, the total time for computing core decompositions within each sublayer is $\Ohat(m)$.
As there are $O(\log n)$ sublayers within a layer, the total time for computing the decompositions over all layers is $\Ohat(m)$. This general idea is quite challenging to carry out, since, in contrast to layers $\Lambda_1,\ldots,\Lambda_{r+1}$, where vertices may only move from higher to lower layers throughout the algorithm, the vertices of a single layer can move between its sublayers in a non-monotone fashion. One of the main challenges in the design of the algorithm is to design a mechanism for allowing the vertices to move between the sublayers, so that the number of such moves is relatively small.

\noindent {\bf Improving query times.}
The algorithm of \cite{fast-vertex-sparsest} supports $\shortkpath(K,u,v)$ queries, that need to return a short path inside the core $K$ connecting the pair $u,v$ of its vertices, in time $\Otil(|V(K)|)+\Ohat(1)$, returning a path of length $\Ohat(1)$; in contrast our algorithm takes time $\Ohat(1)$. The query time of $\shortkpath(K,u,v)$ in turn directly influences the query time of $\shortpath(u,v)$ queries, which in turn is critical to the final query time that we obtain for \SSSP and \APSP problems. Another way to view the problem of supporting $\shortkpath(K,u,v)$ queries is the following: suppose we are given an expander graph $K$ that undergoes edge- and vertex-deletions (in batches). We are guaranteed that after each batch of such updates, the remaining graph $K$ is still an expander, and so every pair of vertices in $K$ has a path of length $O(\poly\log n)$ connecting them. The goal is to support ``short-path'' queries: given a pair $u$, $v$ of vertices of $K$, return a path of length $\Ohat(1)$ connecting them. The problem seems interesting in its own right, and, for example, it plays an important role in the recent fast deterministic approximation algorithm for the sparsest cut problem of \cite{ChuzhoyGLNPS19}. The algorithm of \cite{fast-vertex-sparsest}, in order to process $\shortkpath(K,u,v)$ query, simply perform a breadth-first search in the core
$K$ to find the required $u$-$v$ path, leading to the high query time. Instead, we develop a more efficient algorithm for supporting short-path queries in expander graphs, that is similar in spirit and in techniques to the algorithm of \cite{ChuzhoyGLNPS19}.  This new data structure
has already found further applications to other problems \cite{BernsteinGS20scc}.

For $\shortpath(u,v)$ queries, the guarantees of \cite{fast-vertex-sparsest} are similar to our guarantees in terms of the length of the path returned, but their query processing time is too high, and may be as large as $\tilde \Omega(n)$ in the worst case. We improve
the query time to $\Ohat(|P|)$, where $P$ is the returned path, which is close to the best possible bound. This improvement is necessary in order to obtain faster algorithms for several applications to cut and flow problems that we discuss. The improvement is achieved by exploiting the improved data structure that supports $\shortkpath$ queries within the cores, and by employing a \emph{minimum} spanning
tree data structure on top of the core decomposition, instead of using dynamic connectivity as in the algorithm of \cite{fast-vertex-sparsest}. 

\noindent {\bf Randomized versus Deterministic Algorithm.}
While the algorithm of \cite{fast-vertex-sparsest} works against an adaptive adversary, it is a randomized algorithm. The two main randomized components of the algorithm are: (i) an algorithm to compute a core decomposition; and (ii) data structure that supports $\shortkpath(K,u,v)$ queries within each core. For the first component, we exploit the recent 
fast deterministic algorithm for the Balanced Cut problem of \cite{ChuzhoyGLNPS19}. For the second component, as discussed above, we design a new deterministic algorithm that support $\shortkpath(K,u,v)$ queries within the cores. These changes lead to a deterministic algorithm for the \lcd data structure.

\subsection*{Using the $\protect\lcd$ Data Structure for $\SSSP$ and $\APSP$}

With our improved implementation of the \LCD data structure, using the same approach as that of \cite{fast-vertex-sparsest}, we immediately obtain the desired algorithm for \SSSP, proving \Cref{thm: main for SSSP}.  

Our algorithm for $\APSP$ in the large-distances regime exploits the
$\lcd$ data structure in a similar way as in the algorithm for \SSSP:
We use the $\lcd$ data structure in order to ``compress'' the dense
parts of the graph. In the sparse part, instead of maintaining a single \EST, as in the algorithm for $\SSSP$, we maintain the deterministic tree
cover of \cite{GutenbergW20} (which simplifies the moving $\EST$
data structure of \cite{henzinger16}). 

Our algorithm for $\APSP$ in the small-distances regime uses a \emph{tree cover} approach, similar to previous work
\cite{BernsteinR11,HenzingerKN14_focs,henzinger16,chechik}. The key
difference is that we root each $\EST$ at one of the cores maintained by the
$\lcd$ data structure (recall that each core is a high-degree expander), instead of rooting it at a
random vertex.

The remainder of this section is dedicated to the proof of \Cref{thm:LCD}. However, the statement of this theorem is sufficient in order to obtain our results for $\SSSP$ and $\APSP$, that are discussed in \Cref{sec: SSSP} and \Cref{sec: APSP}, respectively.

\input{LCD_implementation.tex}

%% file: LCD_implementation.tex
\section*{Implementation of Layered Core Decomposition}
\label{sec: LCD implement}

In the remainder of this section, we provide the proof of \Cref{thm:LCD} by showing an implementation of the \lcd data structure, which is the central technical tool of this paper. 
We start by observing that all layers  $\Lambda_{1},\dots,\Lambda_{r}$ can be maintained in near linear time:

\begin{observation}\label{thm: maintaining virtual degrees}
	There is a deterministic algorithm, that, given an $n$-vertex $m$-edge graph $G$ undergoing edge deletions and parameter $\Delta\geq 2$, maintains the partition $(\Lambda_1,\dots,\Lambda_r)$ of $V(G)$ into layers. Additionally, for every vertex $v\in V$ and index $1\leq j\leq r+1$, the algorithm maintains a list of all neighbors of $v$ that belong to $\Lambda_j$. The total update time of the algorithm is $\tilde O(m+n)$.
\end{observation}
\begin{proof}
	We maintain the partition $(\Lambda_{1},\dots,\Lambda_{r+1})$ of $V(G)$ into layers, as graph $G$ undergoes edge deletions, as follows.
	We run $\algPS(G,h_j)$ for maintaining the vertex set $A_j$, for all $1\leq j\leq r$ in parallel. Whenever a vertex $v\in \Lambda_j$ is deleted from $A_j$ by this algorithm, we update its layer accordingly. It is easy to verify that the total update time for maintaining the partition of $V(G)$ into layers is $O((|E(G)|+|V(G)|)\cdot r)=\tilde O(m+n)$.

	Additionally, for every vertex $v\in V(G)$ and index $1\leq j\leq r+1$, we maintain a list of all neighbors of $v$ that lie in $\Lambda_j$. In order to maintain this list, whenever a vertex $u\in \Lambda_i$ is removed from set $A_{i-1}$, we inspect the lists of all neighbors of $u$, and for each such neighbor $v$, we move $u$ to the list of neighbors of $v$ corresponding to the new layer of $u$. Therefore, whenever a virtual degree of a vertex $u$ decreases, we spend $O(\deg_G(u))$ time to update the lists of its neighbors. As virtual degrees can decrease at most $r$ times for every vertex, the total update time for initializing and maintaining these lists is $O(|E(G)|+|V(G)|)\cdot r=\tilde O(m+n)$.
\end{proof}

The remainder of the section is organized as follows.
\thatchapholnew{First,  we list some known tools related to expanders in \Cref{subsec: expander tool} and then, in \Cref{subsec:shortkpath}, provide a new tool, called a short-path oracle on decremental expanders that will be useful for $\shortkpath$ queries.}
One of our key ideas is to further partition each layer $\Lambda_{j}$ into \emph{sublayers} $\Lambda_{j,1},\dots,\Lambda_{j,L_j}$. We describe the structure of the sublayers and the invariants that we maintain for each sublayer in \Cref{subsec: sublayer}. For every sublayer  $\Lambda_{j,\ell}$, the execution of the algorithm is partitioned into phases with respect to that sublayer, that we refer to as \emph{$(j,\ell)$-phases}. 
At the beginning of each $(j,\ell)$-phase, we compute a core decomposition of graph $G[\Lambda_{j,\ell}]$ and obtain a collection $\fset_{j,\ell}$ of cores for the sublayer $\Lambda_{j,\ell}$. 
\Cref{subsec: core-decomp} describes the algorithm for computing the core decompositions.
\thatchapholnew{The description of an algorithm that we use to maintain each core and to support $\shortkpath$ queries on each core is shown in \Cref{subsec:maintain core}.}
During each $(j,\ell)$-phase, vertices can move between the sublayers of layer $j$ in a non-monotone manner (in contract to the fact that every vertex can only move from higher to lower layers). We describe  how vertices are moved between sublayers in \Cref{subsec:move describe} and state the key technical lemma that bounds the total number of times vertices may  move between sublayers. We then bound the total number of cores ever created by the algorithm in \Cref{subsec:bound LCD}; this bound is crucial for our \lcd data structure. %
In \Cref{subsec: tocore exist}, we show an algorithm for processing $\tocore$ queries.
We provide additional technical details for maintaining all necessary data structures in \Cref{subsec: lcd incident ds} and \Cref{subsec: lcd update time}. Finally, we describe the algorithm for responding to $\shortpath$ queries in \Cref{subsec: shortpath}.

\subsection{Known Expander-Related Tools}
\label{subsec: expander tool}

In this subsection we describe several expander-related tools, that mostly follow from previous work, that our algorithm uses.

\paragraph{Expander Decomposition.}

The following theorem can be obtained immediately from the recent deterministic algorithm of \cite{ChuzhoyGLNPS19} for computing a (standard) expander decomposition in almost-linear time; the proof is deferred to \Cref{sec: strong expander decomp}.
As before, we denote $\gamma(n) = \exp(O({\log^{3/4}n}))=n^{o(1)}$. 

\begin{theorem}\label{thm: expander decomp}
	There is a deterministic algorithm, that, given a connected graph $G=(V,E)$ with $n$ vertices and  $m$ edges, and a parameter $0\leq \phi\leq 1$,  computes a partition of $V$ into disjoint subsets $V_1,\ldots,V_k$, such that  $\sum_{i=1}^{k}\delta(V_i)\leq \gamma(m) \cdot \phi m$, and, for all $1\leq i\leq k$, $G[V_i]$ is a strong $\phi$-expander with respect to $G$. The running time of the algorithm is $\Ohat(m)$.
\end{theorem}

\paragraph{Expander Pruning.}
In the following theorem, we consider the setting where we are given as input a graph $G=(V,E)$, with $|E|=m$. Intuitively, we hope that $G$ is a $\phi$-expander for some $0\leq \phi\leq 1$, though it may not be the case. We assume that $G$ is represented by its adjacency list. We also assume that we are given an input sequence $\sigma=(e_1,e_2,\ldots,e_k)$ of online edge deletions, and we denote by $G_i$ the graph $G$ at time $i$, that is, $G_0$ is the original graph $G$, and for all $1\leq i\leq k$, $G_i=G\setminus\set{e_1,\ldots,e_i}$. Our goal is to maintain a set $S\subseteq V$ of vertices, such that, intuitively, if we denote by $S_i$ the set $S$ at time $i$ (that is, after the deletion of the first $i$ edges from $G$), then $G[V\setminus S_i]$ is large enough compared to $G$. Moreover, if $G$ was a $\phi$-expander, then for all $i$, $G[V\setminus S_i]$ remains a strong enough expander. We also require that the set $S$ is incremental, that is, $S_{i-1}\subseteq S_i$ for all $i$. The following theorem, proved in~\cite{expander-pruning} allows us to do so.

\begin{theorem}[Restatement of Theorem 1.3 in~\cite{expander-pruning}]\label{thm: expander pruning}
	There is a deterministic algorithm, that, given an access to the adjacency list of a graph $G=(V,E)$ with $|E|=m$, a parameter $0<\phi\leq 1$, and a sequence $\sigma=(e_1,e_2,\ldots,e_k)$ of $ {k\leq \phi m/10}$ online edge deletions, maintains a vertex set $S\subseteq V$ with the following properties. Let $G_i$ be the graph $G$ after the edges $e_1,\ldots,e_i$ have been deleted from it; let $S_0=\emptyset$ be the set $S$ at the beginning of the algorithm, and for all $0<i\leq k$, let $S_i$ be the set $S$ after the deletion of $e_1,\ldots,e_i$. Then, for all $1\leq i\leq k$:
	
	\begin{itemize}
		\item $S_{i-1}\subseteq S_i$;
		\item $ {\vol_G(S_i)\leq 8i/\phi}$;
		\item $|E(S_i,V\setminus S_i)|\leq 4i$; and
		\item if  $G$ is a $\phi$-expander, then \thatchapholnew{$G_i[V\setminus S_i]$ is a strong $\phi/6$-expander with respect to $G$.}
	\end{itemize}
	
	The total running time of the algorithm is $ {O(k\log m/\phi^2)}$.
\end{theorem}

\paragraph{Embeddings.}
Let $G,W$ be two graphs with $V(W)\subseteq V(G)$. A set $\pset$ of paths in $G$ is called an \emph{embedding of $W$ into $G$} if, for every edge $e=(u,v)\in E(W)$, there is a path $\mathrm{path}(u,v)\in\pset$, such that
$\mathrm{path}(u,v)$ is a $u$-$v$ path in $G$. We say that the \emph{length} of the embedding $\pset$ is $l$ if the length of every path in $\pset$ is at most $l$, and we say that the \emph{congestion} of the embedding is $\eta$ iff every edge of $G$ participates in at most $\eta$ paths in $\pset$. If embedding $\pset$ has length $l$ and congestion $\eta$, then we sometimes call it an $(l,\eta)$-embedding, and we say that $W$ $(l,\eta)$-embeds into $G$.

The following algorithm allows us to quickly embed a smaller expander
into a given expander; the proof appears in \Cref{sec:embedding}. Recall that we denoted $\gamma(n) = \exp(O({\log^{3/4}n}))$.

\begin{theorem}
	\label{lem:embed}
	There is a deterministic algorithm that, given an $n$-vertex $m$-edge
	graph $G$ which is a $\phi$-expander, and a terminal set $T\subseteq V(G)$,
 computes a graph $W$ with $V(W)=T$ and maximum vertex degree $O(\log|T|)$ such that $W$ is a $(1/\gamma(|T|))$-expander. The algorithm also computes a
  $(l,\eta)$-embedding $\pset$ of $W$ into $G$ with $l=O(\phi^{-1}\log m)$ and $\eta=O(\phi^{-2}\log^{2}m)$. 
The running time of the algorithm is $\tilde O(m\cdot\gamma(|T|)/\phi^3)$.
\end{theorem}

\subsection{\thatchapholnew{A New Tool: Short-Path Oracle for Decremental Expanders}}
\label{subsec:shortkpath}

\thatchapholnew{Based on known expander-related tools from the previous section, we provide a new tool, that we call a \emph{short-path oracle on decremental expanders}. This will be a key tool for $\shortkpath$ queries. We believe that the techniques used in this section are of independent interest as they are quite generic. In fact, they have already been subsequently generalized to directed graphs in \cite{BernsteinGS20scc}.
We fix the following parameters that will be used throughout this section.}
We set the parameters as follows:

\begin{equation}
\gamma=\gamma(m) = \exp(O(\log^{3/4}m)) = n^{o(1)}; 
\end{equation}

and 

\begin{equation}
\phi=1/(c \cdot \gamma),
\end{equation}

where $c$ is a large enough constant.

Below, we say that a vertex set $S$ is \emph{incremental}
if vertices in $S$ can never leave $S$ as time progresses. 

\begin{theorem}
	\label{thm:shortkquery} There is a deterministic algorithm that, given
	an $m$-edge $n$-vertex $\phi$-expander $G$
	undergoing a sequence at most $\phi|E(G)|/10$ edge deletions, and a parameter $q>1$,
	maintains an incremental \emph{vertex set} $S\subseteq V(G)$, such that, if we denote by $G^{(0)}$ the graph $G$ before any edge deletions, then, for every $t>0$, after $t$ edges are deleted from $G$, $\vol_{G^{(0)}}(S)\leq O(t/\phi)$ holds \thatchapholnew{and $G\setminus S$ is a strong $\phi/6$-expander with respect to $G^{(0)}$.}
	The algorithm also supports the following query: given a pair of vertices $u,v\in V(G)\setminus S$, return a simple $u$-$v$  path $P$ in $G\setminus S$ of length at most $(\gamma(m))^{O(q)}$, with query time $(\gamma(m))^{O(q)}$. 
	The total update time of the algorithm is  $O(m^{1+1/q}(\gamma(n))^{O(q)})$.
\end{theorem}

\thatchapholnew{The remainder of this section is dedicated to proving \Cref{thm:shortkquery}.}

Throughout the algorithm, $m$ denotes the number of edges in the original $\phi$-expander graph $G$, and the parameter $\phi=1/(c\gamma(m))$ remains unchanged. As our main tools, we employ the Expander Pruning Algorithm from \Cref{thm: expander pruning}, and the algorithm from  \Cref{lem:embed}  that allows us to embed a smaller expander into a given expander. We use parameters $l=O(\log m/\phi)$ and $\eta=O(\log^{2}m/\phi^2)$. 
The idea of the algorithm is to maintain a hierarchy of expander graphs $G_1,\ldots,G_q$, where for all $1\leq i< q$, graph $G_i$ contains $\ceil{m^{i/q}}$ vertices, and it is a $\phi/6$-expander; we set $G_q=G$. We will also maintain an $(l,\eta)$-embedding $\pset_i$ of each such graph $G_i$ into graph $G_{i+1}$. Initially, for all $1\leq i<q$, both the graph $G_i$ and its embedding into $G_{i+1}$ are computed using  \Cref{lem:embed}. Additionally, we maintain an \EST in graph $G_{i+1}$, rooted at the vertex set $V(G_{i})$, with distance threshold $O(\log n/\phi)$. For every edge $e\in E(G_i)$, we will maintain a list $J_i(e)$ of all edges $e'\in E(G_{i-1})$, such that the embedding of $e'$ in $\pset_{i-1}$ contains the edge $e$; recall that $|J(e)|\leq \eta$ must hold. Whenever edge $e$ is deleted from graph $G_i$, this will trigger the deletion of all edges in its list $J_i(e)$ from graph $G_{i-1}$.
Lastly, we use the algorithm from \Cref{thm: expander pruning} in order to maintain, for every expander $G_i$, the set $S_i$ of ``pruned-out'' vertices. When set $S_i$ becomes too large, we re-initialize the graphs $G_i,G_{i-1},\ldots,G_1$.

The outcome of the algorithm is vertex set $S=S_q$, the pruned-out set that we maintain for the expander $G_q=G$. 
\thatchapholnew{Recall that  $G^{(0)}$ denotes the graph $G$ before any edge deletions. \Cref{thm: expander pruning} directly guarantees that $G\setminus S$ is a strong $\phi/6$-expander with respect to $G^{(0)}$ and $\vol_{G^{(0)}}(S) \le O(t/\phi)$ after $t$ edge deletions as desired.} %

We note that, since $G_q$ may be a high-degree graph, it is convenient to define a new graph $G'_q$, that is obtained from $G_q$ by sub-dividing every edge $e$ of $G_q=G$ with a new vertex $v_e$. We let $X=\set{v_e\mid e\in E(G)}$. It is easy to verify that $G'_q$ remains a $\phi/2$-expander. 

In \Cref{alg:core init} we describe the implementation of the algorithm $\mathtt{InitializeExpander}(i)$; the algorithm initializes the data structures for expander $G_{i-1}$, assuming that expander $G_i$ is already defined. The algorithm then recursively calls to $\mathtt{InitializeExpander}(i-1)$. At the beginning of the algorithm, we initialize the whole data structure by calling  $\mathtt{InitializeExpander}(q)$. If, over the course of the algorithm, for some $1\leq i<q$, the number of edges deleted from $G_i$ exceeds $\phi|E(G_i)|/10$, we will call $\mathtt{InitializeExpander}(i-1)$.

\begin{algorithm}
	\textbf{Assertion:} $G_{i}$ is a $\phi/6$-expander. 
	\begin{enumerate}
		\item If $i=1$, then initialize an \EST $T_{1}$ in $G_1$, rooted at an arbitrary
		vertex, with distance threshold $O(\phi^{-1}\log n)$; return. 
		\item If $i=q$, then let $X_q$ be an arbitrary subset of the set $\set{x_e\mid e\in E(G)}$ of vertices of $G'_q$ of cardinality $\ceil{m^{(q-1)/q}}$; otherwise, set $G'_i=G_i$, and let $X_i$ be any subset of $V(G'_i)$ of cardinality 
		$\ceil{m^{(i-1)/q}}$.
		\item \label{enu:core embed} Using the algorithm from \Cref{lem:embed}, compute an expander $G_{i-1}$ over vertex set $X_i$, and its $(l,\eta)$-embedding $\pset_{i-1}$ into $G'_{i}$.
		\item For every edge $e\in E(G_i)$, initialize a list $J_i(e)$ of all edges of $G_{i-1}$ whose embedding  path in $\pset_{i-1}$ contains $e$.
		\item Initialize the expander pruning algorithm from \Cref{thm: expander pruning}
		on $G_{i-1}$, that will maintain a pruned vertex set $S_{i-1}\subseteq V(G_{i-1})$. 
		
		\item \label{enu:APSP root ES tree}Initialize an ES-tree $T_{i}$ in $G'_{i}$
		rooted at $X_{i}$, with distance threshold $O(\phi^{-1}\log n)$. 
		\item Call $\mathtt{InitializeExpander}(i-1)$. 
	\end{enumerate}
	\protect\caption{$\mathtt{InitializeExpander}(i)$\label{alg:core init}}
\end{algorithm}

We denote by  $G_{i-1}^{(0)}$ the expander graph created by Procedure $\mathtt{InitializeExpander}(i)$. For all $d>0$, we denote by $G_{i-1}^{(d)}$ the graph that is obtained from  $G_{i-1}^{(0)}$ after $d$ edge deletions from $G$. 
As $d$ increases, our algorithm maintains the graph $G_{i-1}=G_{i-1}^{(d)}\setminus S_{i-1}$.
By \Cref{thm: expander pruning}, as long as $d\leq \phi|E(G_i)|/10$, graph $G_{i-1}$ remains a $(\phi/6)$-expander.

When some edge $e$ is deleted from graph $G$, we call  Algorithm $\mathtt{Delete}(q,e)$. The algorithm may recursively call to procedure $\mathtt{Delete}(i,e')$ for other expander graphs $G_i$ and edges $e'$. The algorithm  $\mathtt{Delete}(i,e')$  is shown in \Cref{alg:core delete}.

\begin{algorithm}
	\begin{enumerate}
		\item If $i=1$, delete $e$ from graph $G_{1}$. Recompute an \EST
		$T_{1}$ in graph $G_{1}$, up to depth $O(\log n/\phi)$, rooted at any vertex; return.
		\item Delete $e$ from $G_{i}$. Update the pruned-out vertex set $S_{i}$ using \Cref{thm: expander pruning}. 
		\item Let $D_{i}^{new}$ denote the set of edges that were just removed from
		$G_{i}$. That is, $D_{i}^{new}$ contains $e$ and all edges incident
		to vertices that were newly added into $S_{i}$. 
		\item For each $e\in D_{i}^{new}$, for every edge $e'\in J_i(e)$, call $\mathtt{Delete}(i-1,e')$;
				\item If the total number of edge deletions from $G_i^{(0)}$ exceeds $\phi|E(G_{i}^{(0)})|/10$, call $\mathtt{InitializeExpander}(i+1)$. 
	\end{enumerate}
	\caption{$\mathtt{Delete}(i,e)$ where $e\in E(G_{i})$\label{alg:core delete}}
\end{algorithm}

We bound the total update time of the algorithm in the following lemma.

\begin{lemma}\label{lem: oracle update time}
	The total update time of the algorithm is $O(m^{1+1/q}(\gamma(n))^{O(q)})$. 
\end{lemma}

\begin{proof}
	Fix an index $1\leq i\leq q$. We partition the execution of the algorithm into \emph{level-$i$ stages}, where each level-$i$ stage starts when $\mathtt{InitializeExpander}(i+1)$ is called, and terminates just before the subsequent call to $\mathtt{InitializeExpander}(i+1)$. Recall that, over the course of a level-$i$ stage, at most $\phi|E(G_{i}^{(0)})|/10$ edges are deleted from the graph $G_i^{(0)}$. We now bound the running time that is needed in order to initialize and maintain the level-$i$ data structure over the course of a single level-$i$ stage. This includes the following:
	
	\begin{itemize}
		\item Computing expander $G_{i}$ and its $(l,\eta)$-embedding $\pset_{i-1}$ into $G'_{i+1}$ using the algorithm from  \Cref{lem:embed}; the running time is $\tilde O(|E(G_{i+1})|\cdot\gamma(m)/\phi^3)=O\left (m^{(i+1)/q}\cdot (\gamma(n))^{O(1)}\right )$.
		\item Initializing the lists $J_{i+1}(e)$ for edges $e\in G_{i+1}$: the time to initialize all such lists is bounded by the time needed to compute the embedding $\pset_i$, which is in turn bounded by  $O\left (m^{(i+1)/q}\cdot (\gamma(n))^{O(1)}\right )$.
		\item Initializing and maintaining the \EST  $T_{i+1}$: the running time is $O(|E(G_{i+1})|\cdot \log n/\phi)\leq  O\left (m^{(i+1)/q}\cdot (\gamma(n))^{O(1)}\right )$.
		\item Running the algorithm for expander pruning on the expander $G_i$; from \Cref{thm: expander pruning}, the running time is $\tilde O(|E(G_i^{(0)})/\phi)\leq O\left (m^{i/q}\cdot (\gamma(n))^{O(1)}\right )$, since the number of edge deletions is bounded by $\phi|E(G_{i}^{(0)})|/10$.
		\item The remaining work, executed by $\mathtt{Delete}(i,e)$, for every edge $e$ that is deleted from graph $G_i$ (including edges incident to the vertices of the pruned out set $S_i$), requires $O(\eta)$ time per edge, with total time $O(|E(G_i^{(0)})|\cdot \eta)\leq  O\left (m^{i/q}\cdot (\gamma(n))^{O(1)}\right )$.
	\end{itemize}
	
	Therefore, the total time that is needed in order to initialize and maintain the level-$i$ data structure over the course of a single level-$i$ stage is  $O\left (m^{(i+1)/q}\cdot (\gamma(n))^{O(1)}\right )$
	Note that we did not include in this running time the time required for maintaining level-$(i-1)$ data structures, that is, calls to $\mathtt{InitializeExpander}(i)$ and $\mathtt{Delete}(i-1,e)$.
	
	Next, we bound the total number of level-$i$ stages. Consider some index $1< i'\leq q$, and consider a single level-$i'$ stage. Recall that, over the course of this stage, at most $d_{i'}=\phi|E(G_{i'}^{(0)})|/10$ edge deletions from graph $G_{i'}^{(0)}$ may happen. From \Cref{thm: expander pruning}, over the course of the level-$i'$ stage, the total volume of edges that are incident to the pruned-out vertices in $S_i$ is bounded $O(d_{i}/\phi)$.
	As the embedding $\pset_{i'}$ of $G_{i'-1}$ into $G'_{i}$ has congestion at most
	$\eta$, this can cause at most $O(\eta d_{i}/\phi)$ edge deletions
	from graph $G_{i'-1}^{(0)}$. As a single level-$(i'-1)$ stage requires $\phi|E(G_{i'-1}^{(0)})|/10$ edge deletions from $G_{i'-1}^{(0)}$, the number of level-$(i'-1)$ stages that are contained in a single level-$i'$ stage is bounded by:

	\[
	\frac{O(d_{i'}\cdot \eta /\phi)}{\phi|E(G_{i'-1}^{(0)})|/10}\leq \frac{O(|E(G_{i'}^{(0)})| \cdot \log^3m)}{\phi^3\cdot  |E(G_{i'-1}^{(0)})|}\leq    O(m^{1/q}\cdot(\gamma(n))^{O(1)}).
	\]
	
	Since we only need to support at most $\phi |E(G)|/10$ edge deletions from the original graph $G$, there is only a single level-$q$ stage. Therefore, for every $1\leq i<q$, the total number of level-$i$ stages is bounded by:  $O(m^{(q-i)/q}\cdot (\gamma(n))^{O(q-i)})$. We conclude that the total amount of time required for maintaining level-$i$ data structure is bounded by:
	
	\[ O\left (m^{(i+1)/q}\cdot (\gamma(n))^{O(1)}\right ) \cdot O\left (m^{(q-i)/q}\cdot (\gamma(n))^{O(q-i)}\right )\leq O\left (  m^{1+1/q}\cdot (\gamma(n))^{O(q-i)}\right ).  \]
	
	Summing this up over all $1\leq i\leq q$, we get that the total update time of the algorithm is $O\left (  m^{1+1/q}\cdot (\gamma(n))^{O(q)}\right )$, as required.
\end{proof}

Next, we provide an algorithm for responding to the short-path query between a pair $u,v$ of vertices. The algorithm calls $\mathtt{Query}(q,u,v)$, that is described in  \Cref{alg:core query}, which recursively calls $\mathtt{Query}(i,u',v')$ for $i<q$. The idea of the algorithm is simple: we use the \EST $T_q$ in graph $G_q$ in order to compute two paths: one path connecting $u$ to some vertex $u'\in X_q$, and one path connecting $v$ to some vertex $v'\in X_q$, and then recursively call the short-path query for the pair $u',v'$ of vertices in the expander $G_{q-1}$; we then use the embedding $\pset_q$ of $G_{q-1}$ into $G_q$ to convert the resulting path into a $u'$-$v'$ path in $G_q$. The final path connecting $u$ to $v$ is obtained by concatenating the resulting three paths.

\begin{algorithm}
\begin{enumerate}
	\item If $i=1$, return the unique  $u$-$v$ path in $T_{1}$. 
	\item Compute, in $T_{i}$, a unique path $Q_{uu'}$ connecting $u$ to some vertex $u'\in X_{i}$, and a unique  path $Q_{v'v}$ connecting $v$ to some vertex $v'\in X_{i}$ to $v$. 
	\item If $v'=u'$, set $R_{u'v'}=\emptyset$; otherwise set $R_{u'v'}=\mathtt{Query}(i-1,u',v')$. 
	\item Let $Q_{u'v'}$ be a path in $G_i$ obtained by concatenating, for all edges $e' \in R_{u'v'}$, the corresponding path $P(e')\in \pset_i$ from the embedding of $G_{i-1}$ into $G'_i$.
	\item Return $Q_{uv}=Q_{uu'}\circ Q_{u'v'}\circ Q_{v'v}$. (Note that for $i=q$, $Q_{u,v}$ is a path in graph $G'_q$, that was obtained from $G_q$ by subdividing its edges; it is immediate to turn it into a corresponding path in $G_q$.)
\end{enumerate}
\caption{$\mathtt{Query}(i,u,v)$ where $u,v\in V(G_{i})$\label{alg:core query}}
\end{algorithm}

The following lemma summarizes the guarantees of the algorithm for processing short-path queries.

\begin{lemma}\label{lem: oracle nonsimple path}
Given any pair of vertices $u,v\in V(G)\setminus S$, algorithm $\mathtt{Query}(q,u,v)$
returns a  (possibly non-simple)  $u$-$v$ path $Q$ in $G\setminus S$, of length $(\gamma(n))^{O(q)}$, in time $O(|Q|)$.
\end{lemma}

\begin{proof}
Let $\mathrm{Len}(i)$ be the maximum length of the path in $G_{i}$
returned by $\mathtt{Query}(i,u,v)$. As $G_{i}$ is always a $\phi/6$-expander
by \Cref{thm: expander pruning}, it is immediate to verify that the diameter of $G_{i}$
is $O(\phi^{-1}\log n)$, and so the \EST tree $T_{i}$ spans graph $G_{i}$. Consider \Cref{alg:core query}.
Let $Q_{uu'}$ and $Q_{v'v}$ be the path in $G'_{i}$ from $u$ to
$u'\in X_{{i-1}}$ and the path in $G'_{i}$ from $v'\in X_{{i-1}}$
to $v$. As $T_{i}$ spans $G'_{i}$, $Q_{uu'}$ and $Q_{v'v}$ do
exist. Let $R_{u'v'}=\mathtt{Query}(i-1,u',v')$ where $|R_{u'v'}|\le\mathrm{Len}(i-1)$.
Let $Q_{u'v'}$ be obtained by concatenating $\mathrm{path}(e')$
for each $e'\in R_{u'v'}$ where $\mathrm{path}(e')\in \pset_i$
is the corresponding embedding path of $e'$. We have $|Q_{u'v'}|\le\ell\cdot|R_{u'v'}|$.
It is clear that the concatenation $Q_{uu'}\circ Q_{u'v'}\circ Q_{v'v}$
is indeed a $u$-$v$ path in $G'_{i}$ and hence in $G_{i}$. The
length of this path is at most 
\[
\mathrm{Len}(i)=O(\phi^{-1}\log n)+O(\phi^{-1}\log n)\cdot\mathrm{Len}(i-1).
\]
Solving the recursion gives us $\mathrm{Len}(i)=(\gamma(n))^{O(i)}$.
So $\mathtt{Query}(q,u,v)$ returns a $u$-$v$ path of length $(\gamma(n))^{O(q)}$.
Observe that the query time is proportional to the number of edges on the returned path. 
\end{proof}

Lastly, observe that a path $Q$ connecting the given pair $u,v$ of vertices, that is returned by algorithm $\mathtt{Query}(q,u,v)$ may not be simple. It is easy to turn $Q$ into a simple path $Q'$, in time $O(|Q|)$, by removing all cycles from $Q$. The final path $Q'$ is guaranteed to be simple, of length $(\gamma(n))^{O(q)}$, and the query time is bounded by $(\gamma(n))^{O(q)}$, as required.

\subsection{Sublayers and Phases}
\label{subsec: sublayer}

In this subsection we focus on a single layer $\Lambda_j$, for some $1<j\leq r$.
Recall that we have denoted by $n_{\le j}$ the number of vertices
that belonged to set $\Lambda_{\le j}$ at the beginning of the algorithm, before any edges were deleted from the input graph.
We let $L_j$ be the smallest integer $\ell$, such that $n_{\le j}/2^{\ell-1}\leq h_j/2$; observe that $L_j\leq \log n$.
 We further partition vertex set $\Lambda_j$ into  subsets $\Lambda_{j,1},\Lambda_{j,2},\dots,\Lambda_{j,L_{j}}$. We refer to each resulting vertex set $\Lambda_{j,\ell}$ as \emph{sublayer
	$(j,\ell)$}. For indices $1\leq \ell\leq \ell'\leq L_j$, we say that sublayer $\Lambda_{j,\ell}$ is \emph{above} sublayer $\Lambda_{j,\ell'}$.
The last sublayer $\Lambda_{j,L_{j}}$, that we also denote for convenience by $\Lambda_{j}^{-}$,  is called the \emph{buffer}
sublayer. For convenience, we also use the shorthand $\Lambda_{j,\le\ell}=\bigcup_{\ell'=1}^{\ell}\Lambda_{j,\ell'}$, and we define $\Lambda_{j,<\ell},\Lambda_{j,\ge\ell},\Lambda_{j,>\ell}$ 
similarly.%

{} We will ensure that throughout the algorithm, the following invariant always holds:
\begin{properties}{I}
	\item for all $1\leq \ell\leq L_j$, $|\Lambda_{j,\geq \ell}|\leq n_{\le j}/2^{\ell-1}$. \label{inv: size of every sublayer}
\end{properties}

At the beginning of the algorithm, we set $\Lambda_{j,1}=\Lambda_{j}$
and $\Lambda_{j,\ell'}=\emptyset$ for all $1<\ell'\leq L_j$.
Consider now some sublayer $\Lambda_{j,\ell}$, for $\ell<L_j$. We partition the execution of our algorithm into phases with respect to this sublayer, that we refer to as \emph{level-$(j,\ell)$
	phases}, or \emph{$(j,\ell)$-phases}. Whenever Invariant \ref{inv: size of every sublayer} for the next sublayer  $\Lambda_{j,\ell+1}$ is violated (that is, $|\Lambda_{j,\geq \ell+1}|$ exceeds $n_{\le j}/2^{\ell}$), we terminate the current $(j,\ell)$-phase and start a new phase. 

Consider now some time $t$ during the execution of the algorithm, when, for some $1\leq \ell<L_j$, a $(j,\ell)$-phase is terminated. Let $\ell'$ be the smallest index for which the $(j,\ell')$-phase is terminated at time $t$. We then set $\Lambda_{j,\ell'}=\Lambda_{j,\ge\ell'}$, and for all $\ell'<\ell''\leq L_j$, we set 
 $\Lambda_{j,\ell''}=\emptyset$.

Throughout the algorithm, for every vertex $u$, we denote by $\deg_{\le(j,\ell)}(u)=|E(u,\Lambda_{<j}\cup\Lambda_{j,\le\ell})|$
 the number of neighbors of $u$ that lie in sublayer-$(j,\ell)$ or above (including in layers that are above layer $j$). By the definition, $\deg_{\le(j,\ell)}(u) \le \deg_{\le j}(u)$. However, since, at the beginning of each $(j,\ell)$-phase, we set $\Lambda_{j,\ell'}\gets\emptyset$ for all $\ell'>\ell$, we obtain the following immediate observation:

\begin{observation}\label{obs:no downward degree}
	For all $1\leq \ell< L_j$, for every vertex $u$, at the beginning of each $(j,\ell)$-phase, $\deg_{\le(j,\ell)}(u) = \deg_{\le j}(u)$.
\end{observation}

Let $H_{j,\ell}=G[\Lambda_{j,\ell}]$ be the subgraph of $G$ induced
by the vertices of the sublayer $\Lambda_{j,\ell}$. We refer to  $H_{j,\ell}$ as the
\emph{level-$(j,\ell)$ graph}. We will also ensure that throughout the algorithm the following invariant holds:

\begin{properties}[1]{I}
	\item For all $1\leq \ell<L_{j}$, for each level-$(j,\ell)$ phase, graph $H_{j,\ell}$ only undergoes deletions of edges and vertices over the course of the phase. \label{inv: H is decremental}
\end{properties}

Therefore, we say that graph $H_{j,\ell}$ is \emph{decremental} within each $(j,\ell)$-phase.
Note that the graph $H_{j,L_{j}}$ that corresponds to the buffer sublayer $\Lambda_j^-$ may undergo both insertions and deletions
of edges and vertices. As time progresses, some vertices $v$ whose
virtual degree $\td(v)$ was greater than $h_{j}$ may have their
virtual degree decrease to $h_{j}$. In order to preserve the above invariant, we always insert such vertices $v$ into the buffer sublayer $\Lambda_{j}^{-}=\Lambda_{j,L_{j}}$; additional vertices may also be moved from higher sub-layers $\Lambda_{j,\ell}$ to the buffer sub-layer over the course of a $(j,\ell)$-phase.

\subsection{Initialization of a Sublayer: Core Decomposition}
\label{subsec: core-decomp}

Consider now some non-buffer sub-layer $\Lambda_{j,\ell}$,  with $\ell<L_{j}$. At the beginning of every $(j,\ell)$-phase, if $\Lambda_{j,\ell}\neq \emptyset$, we compute a \emph{core decomposition}
of the graph $H_{j,\ell}$. This is one of the key subroutines in our $\lcd$ data structure.
The following theorem provides the algorithm for computing the core decomposition of a sub-layer.

\begin{theorem}
	[Core Decomposition of Sublayer]\label{thm: core decomposition}There is a deterministic
	algorithm, that, given a level-$(j,\ell)$ graph $H_{j,\ell}=G[\Lambda_{j,\ell}]$,
 computes a collection $\fset_{j,\ell}$
	of vertex-disjoint subgraphs of $H_{j,\ell}$, called \emph{cores}, such that each core $K \in \fset_{j,\ell}$ is a $\phi$-expander, and  for every vertex  $u\in V(K)$, $\deg_{K}(u)\ge\phi \cdot \deg_{\le(j,\ell)}(u)/\gapdegree$. Moreover, if we denote by $U_{j,\ell}=\Lambda_{j,\ell}\setminus\left (\bigcup_{K\in\fset_{j,\ell}}V(K)\right )$
		 the set of all vertices of $\Lambda_{j,\ell}$ that do not belong to any core, then there is an orientation of the edges of the graph  $G[U_{j,\ell}]$, such that the resulting directed graph $\dset_{j,\ell}$ is a directed acyclic graph (DAG), and, for every vertex $u\in U_{j,\ell}$, $\indeg_{\dset_{j,\ell}}(u)\le\deg_{\le(j,\ell)}(u)/\gapdegree$.
	The running time of the algorithm is $\Ohat(|E(H_{j,\ell})|)$.
\end{theorem}

\begin{proof}
	We use the following lemma, whose proof follows easily from \Cref{thm: expander decomp}.
	\begin{lemma}
		\label{thm: one phase of core decomp}There is a deterministic algorithm,
		that given a subgraph $H'_{j,\ell}\subseteq H_{j,\ell}$, such that
		every vertex $u\in V(H'_{j,\ell})$ has degree at least $\deg_{\le(j,\ell)}(u)/\gapdegree$ in $H'_{j,\ell}$,
		in time $\Ohat(|E(H_{j,\ell})|)$ computes a collection $\fset$
		of vertex-disjoint subgraphs of $H'_{j,\ell}$ called cores, such that each core $K\in\fset$ a $\phi$-expander and, for every vertex $u\in V(K)$, $\deg_{K}(u)\ge\phi\deg_{\le(j,\ell)}(u)/\gapdegree$. Moreover, $\sum_{K\in\fset}|E(K)|\ge3|E(H'_{j,\ell})|/4$.
	\end{lemma}
	
	\begin{proof}
		We apply \Cref{thm: expander decomp} to every connected component of graph $H'_{j,\ell}$,  with parameter $\phi$. Let $(V_1,\ldots,V_k)$ be the resulting partition of $V(H'_{j,\ell})$. For each $1\leq i\leq k$, we the define a core $K_i=H'_{j,\ell}[V_i]$.
		Observe first that $\sum_{i=1}^k\delta(V_i)\leq \gamma \cdot \phi |E(H'_{j,\ell})| \le |E(H'_{j,\ell})|/4$ by our choice of $\phi$. Therefore, $\sum_{i=1}^k|E(K_i)|\geq 3|E(H'_{j,\ell})|/4$. We are also guaranteed that for all $1\leq i\leq k$, graph $K_i$ is a strong $\phi$-expander with respect to $H'_{j,\ell}$. Lastly, if $K_i$ contains more than one vertex, then, from Observation~\ref{obs: high degrees in strong expander}, every vertex $u$ of $K_i$ has degree at least $\phi \deg_{H'_{j,\ell}}(u) \ge \phi \deg_{\le (j,\ell)}(u)/\gapdegree$ in $K_i$. We return a set $\fset$ containing all graphs $K_i$ with $|V(K_i)|>1$. The running time of the algorithm is $\Ohat(|E(H'_{j,\ell})|)$ by \Cref{thm: expander decomp}.
	\end{proof}

	We are now ready to complete the proof of \Cref{thm: core decomposition}. Our algorithm is iterative. 
	We start with $\fset_{j,\ell}\gets\emptyset$ and $H'_{j,\ell}\gets H_{j,\ell}$.
	Consider the following \emph{trimming} process similar to the one in \cite{KawarabayashiT19}: while there is a vertex $u\in V(H'_{j,\ell})$
	with $\deg_{H'_{j,\ell}}(u)<\deg_{\le(j,\ell)}(u)/\gapdegree$, delete $u$ from
	$H'_{j,\ell}$. We say that graph $H'_{j,\ell}$ is \emph{trimmed} if, for all $u\in V(H'_{j,\ell})$
	$\deg_{H'_{j,\ell}}(u)\ge\deg_{\le(j,\ell)}(u)/\gapdegree$. While $H'_{j,\ell}\neq \emptyset$, we perform iterations, each of which consists of the following steps: 
	
	\begin{enumerate}
	\item Trim the current graph $H'_{j,\ell}$;
	\item Apply the algorithm from	\Cref{thm: one phase of core decomp} to graph $H'_{j,\ell}$, to obtain a collection
	$\fset$ of cores;
	\item For all $K\in\fset$, delete all vertices of $K$ from $H'_{j,\ell}$;
	\item Set
	$\fset_{j,\ell}\gets\fset_{j,\ell}\cup\fset$. 
	\end{enumerate}
	
	Note that, throughout the algorithm, the graph $H'_{j,\ell}$ that serves as input to \Cref{thm: one phase of core decomp} has vertex degrees at least $\deg_{\le(j,\ell)}(u)/\gapdegree$ due to the trimming operation, so it is a valid input to the lemma. It is also immediate to see that the number of iterations is bounded by $O(\log n)$. Indeed, let $H$ be the graph $H'_{j,\ell}$
	at the beginning of some iteration, and let $H'$ be the graph $H'_{j,\ell}$ at the beginning of the next
	iteration. Note that $H'$ is a subgraph of $H\setminus(\bigcup_{K\in\fset}E(K))$.
	As $|\bigcup_{K\in\fset}E(K)|\ge3|E(H)|/4$, we conclude that $|E(H')|\le|E(H)|/4$.
	Therefore, after $O(\log n)$ iterations, $H'_{j,\ell}$ becomes empty and the algorithm terminates. It is easy to see that the trimming step takes time $O(|E(H'_{j,\ell})|)$.
	Therefore, the total running time of the algorithm is $\Ohat(|E(H_{j,\ell})|)$. 
	
	Consider now the final set $\fset_{j,\ell}$ of cores computed by the algorithm. We now show that it has all required properties. From \Cref{thm: one phase of core decomp}, we are guaranteed that each core $K \in \fset_{j,\ell}$ is a $\phi$-expander, and  that for every vertex  $u\in V(K)$, $\deg_{K}(u)\ge\phi \cdot \deg_{\le(j,\ell)}(u)/\gapdegree$. Observe that every 
	vertex in set $U_{j,\ell}=V(H_{j,\ell}\setminus\bigcup_{K\in\fset_{j,\ell}}K)$ was deleted from graph $H'_{j,\ell}$ by the trimming procedure at some time $t$ in the algorithm's execution. Therefore, at time $t$, $\deg_{H'_{j,\ell}}(u)<\deg_{\le(j,\ell)}(u)/\gapdegree$ held. We orient
	all edges that belonged to graph $H'_{j,\ell}$ at time $t$ and are incident to $u$ towards $u$. This provides an orientation of all edges in graph $G[U_{j,\ell}]$, which in turn defines a directed graph $\dset_{j,\ell}$. From the above discussion, for every vertex of $\dset_{j,\ell}$, $\indeg_{\dset_{j,\ell}}(u)<\deg_{\le(j,\ell)}(u)/\gapdegree$ holds. Moreover, it is easy to see that graph
	$\dset_{j,\ell}$ is a DAG, because the order in which the vertices of $U_{j,\ell}$ were deleted from $H'_{j,\ell}$ by the trimming procedure defines a valid topological ordering of the vertices of $\dset_{j,\ell}$.
\end{proof}

\subsection{Maintaining Cores and Supporting $\shortkpath$ Queries}

\label{subsec:maintain core}

In this subsection, we describe an algorithm for maintaining the cores, and for supporting queries $\shortkpath(K,u,v)$: given any pair of vertices $u$ and $v$, both of which lie in some core $K\in \bigcup_j\fset_j$, return a  simple $u$-$v$ path $P$ in $K$ of length at most  $(\gamma(n))^{O(q)} = \Ohat(1)$, in time $(\gamma(n))^{O(q)} = \Ohat(1)$.

When we invoke the algorithm from \Cref{thm: core decomposition} for computing a core decomposition of sublayer $\Lambda_{j,\ell}$ at the beginning of a  $(j,\ell)$-phase, we  say that the core decomposition \emph{creates} the cores in the set $\fset_{j,\ell}$ that it computes.  Our algorithm only creates new cores through the algorithm from  \Cref{thm: core decomposition}, which may only be invoked at the beginning of a $(j,\ell)$-phase.

Throughout the algorithm, we denote $\fset_{j}=\fset_{j,1}\cup\dots\cup\fset_{j,L_{j}-1}$, and we refer to graphs in $\fset_j$  as \emph{cores for layer $\Lambda_j$}, or \emph{cores for graph $H_j$} (recall that we have defined $H_j=G[\Lambda_j]$).
For convenience, we also use shorthand notation $\fset_{j,\le\ell}=\fset_{j,1}\cup\dots\cup\fset_{j,\ell}$
and $\fset_{\le j}=\fset_{1}\cup\dots\cup\fset_{j}$. We define $\fset_{\ge j}$ and
$\fset_{j,\ge\ell}$ analogously. Let $\hat{K}_{j}=\bigcup_{K\in\fset_{j}}K$, and denote
$\hat{K}_{\le j}=\bigcup_{K\in\fset_{\le j}}K$. We define notation $\hat K_{\geq j}, \hat K_{j,\leq \ell}$ and $\hat K_{j,\geq \ell}$ analogously.

In order to maintain the cores and to support the $\shortkpath(K,u,v)$ queries for each such
 core $K$, we do the following.
  Consider a pair $1\leq j\leq r$, $1\leq \ell<L_j$ of indices, and some core $K\in\fset_{j,\ell}$, that was created when the core decomposition algorithm from \Cref{thm: core decomposition} was invoked for layer $\Lambda_{j,\ell}$, at the beginning of some $(j,\ell)$-phase. Let $K^{(0)}$ denote the
core $K$ right after it is created, before any edges are deleted from $K$;
recall that $K^{(0)}$ is a $\phi$-expander. We use the algorithm from \Cref{thm:shortkquery}
on graph $K^{(0)}$, as it undergoes edge deletions, with the parameter $q$ that serves as input to \Cref{thm:LCD}, to maintain the vertex set $S^{K}\subseteq V(K^{(0)})$. 
 Whenever, over the course of the current $(j,\ell)$-phase, an edge is deleted from graph $G$ that belongs to $K^{(0)}$, we add this edge to the sequence of edge deletions from graph $K^{(0)}$, and update the set $S^K$ of vertices using the algorithm from \Cref{thm:shortkquery}  accordingly. At any point in the current $(j,\ell)$-phase, if $A^{K}\subseteq E(K^{(0)})$ is the set of edges of $K^{(0)}$ that were deleted from $G$ so far, and $S^K$ is the current vertex set maintained by the algorithm from  \Cref{thm:shortkquery}, then we set the current core corresponding to $K^{(0)}$ to be the graph obtained from $K^0$ by deleting the edges of $A^K$ and the vertices of $S^K$ from it; in other words, $K=(K^{(0)}\setminus A^K)\setminus S^k$. We refer
to the resulting graph $K$ as a core throughout the phase. Whenever the number of deleted
edges in $A^{K}$ exceeds $\phi|E(K^{(0)})|/10$, we set $S^{K}=V(K^{(0)})$, which effectively set $K = \emptyset$; at this point we say that core $K$ is \emph{destroyed}. Each destroyed core is removed from $\fset_{j,\ell}$.

From this definition of the core $K$, from the time it is created and until it is destroyed, it may only undergo deletions of edges and vertices.
In addition to the deletion of edges of $K$ due to the edge deletions from the input graph $G$, we also delete vertices of $S^K$ from $K$. Whenever a vertex
$v\in V(K)$ is deleted from $K$ (that is, $v$ is added to $S^{K}$),
we say that $v$ is \emph{pruned out} of $K$. When there are more
than $\phi|E(K^{(0)})|/10$ edge deletions in $A^K$, all vertices of $K$ are
pruned out and so $K$ is destroyed.

Therefore, we can now use \Cref{thm:shortkquery} in order to support queries $\shortkpath(K,u,v)$ for each core $K$: given a pair $u,v\in V(K)$ of vertices of $K$, return a simple $u$-$v$ path $P$ in $K$ of length at most $(\gamma(m))^{O(q)}$ in time $(\gamma(m))^{O(q)}$. We now provide a simple observation about the maintained cores.

\begin{prop}\label{prop: core property}
	For every core $K$, from the time $K$ is created and until it is destroyed, $|V(K)|\ge\Omega(\phi^{2}h_{j})$ holds.
\end{prop}
\begin{proof}
	By \Cref{thm:shortkquery}, $K$ is a strong $\phi/6$-expander w.r.t.~$K^{(0)}$.
	Let $u\in V(K)$ be a vertex minimizing $\deg_{K^{(0)}}(u)$. In particular,
	$\deg_{K^{(0)}}(u)\le\vol_{K^{(0)}}(V(K))/2$ must hold. By \Cref{obs: high degrees in strong expander v2},
	$\deg_{K}(u)\ge(\phi/6)\cdot\deg_{K^{(0)}}(u)$. Since, at the beginning of the $(j,\ell)$-phase  
	\begin{align*}
	\deg_{K^{(0)}}(u) & \ge\phi\deg_{\le(j,\ell)}(u)/\gapdegree & \text{by \Cref{thm: core decomposition}}\\
	& =\phi\deg_{\le j}(u)/\gapdegree & \text{by \Cref{obs:no downward degree}}\\
	& \ge\phi h_{j}/\gapdegree & \text{by \Cref{obs: virt vs reg degrees connections},}
	\end{align*}
	
	held
	we conclude that $|V(K)|\ge\deg_{K}(u)=\Omega(\phi^{2}h_{j})$ (using the fact that the graph is simple).
\end{proof}

\thatchapholnew{
We use the following observation in order to bound the number of cores \emph{at any point during the algorithm's execution}. Later in \Cref{subsec:bound LCD}, we will give another bound for the total number of cores \emph{ever created} by the algorithm.}

\begin{observation}
	\label{obs:bound core moment} For all $1\leq j\leq r$ and $1\leq \ell<L_j$, at any time over the course of the algorithm,
	\thatchapholnew{$|\fset_{j,\ell}|\le O(|\Lambda_{j,\ell}|/(\phi^2 h_{j}))$, and $|\fset_{\le j}|\le O(n_{\le j}/(\phi^2 h_{j}))$ must hold. 
	Moreover,
	if $C$ is a connected component of $G[\Lambda_{\le j}]$, and
	$\fset_{\le j}^{C}=\{K\in\fset_{\le j}\mid K\subseteq C\}$ is the collection
	of cores in $\fset_{\le j}$ that are contained in $C$, then $|\fset^C_{\le j}|\le O(|V(C)|/(\phi^2 h_{j}))$.}
\end{observation}

\begin{proof}
	\thatchapholnew{Consider a set $\fset_{j,\ell}$ of remaining cores in sublayer $\Lambda_{j,\ell}$
	which is not destroyed yet. Note again that new cores in $\fset_{j,\ell}$
	may only be created when the algorithm from \Cref{thm: core decomposition}
	is employed on sublayer $\Lambda_{j,\ell}$. In the beginning of a
	$(j,\ell)$-phase, all cores are vertex disjoint by \Cref{thm: core decomposition}.
	Moreover, each core undergoes deletions only so it remains disjoint and it contains $\Omega(\phi^{2}h_{j})$ vertices at any point
	of time before it is destroyed by \Cref{prop: core property}. So we
	the number of cores is at most $|\fset_{j,\ell}|=O(|\Lambda_{j,\ell}|/(\phi^{2}h_{j}))$
	at any point of time.}
	By summing up the above bound over all sublayers in layers $1,\ldots,j$, and noting that $h_1, h_2,\ldots, h_j$ form a geometrically decreasing sequence, and that $|\Lambda_{\leq j}|\leq n_{\leq j}$ holds at all times, we get that $|\fset_{\le j}|\le O(n_{\le j}/(\phi^2 h_{j}))$.
	
	Lastly, consider some connected component $C$ of $G[\Lambda_{\le j}]$, and let
	$\fset_{\le j}^{C}=\{K\in\fset_{\le j}\mid K\subseteq C\}$. For an index $1\leq \ell<L_j$, let $\fset_{j,\ell}^{C}=\{K\in\fset_{j,\ell}\mid K\subseteq C\}$.
	Using the same argument, \thatchapholnew{$|\fset_{j,\ell}^{C}|\le O(|\Lambda_{j,\ell}\cap V(C)|/(\phi^2 h_{j}))$.
	By summing up over all sub-layers of layers $1,\ldots,j$, we conclude that $|\fset^C_{\le j}|\le O(|V(C)|/(\phi^2 h_{j}))$.}
\end{proof}

\subsection{Maintaining the Structure of the Sublayers}
\label{subsec:move describe}

In this subsection we provide additional details regarding the sub-layers of each layer $\Lambda_j$, and in particular we describe how vertices move between the sublayers. Throughout this subsection, we fix an index $1\leq j\leq r$.

Consider an index $1\leq \ell<L_j$. Throughout the algorithm, we maintain a partition of the vertices of the (non-buffer) sub-layer $\Lambda_{j,\ell}$ into two subsets: 
set $\hat{K}_{j,\ell}$ that contains all vertices currently lying in the cores of $\fset_{j,\ell}$, so $\hat{K}_{j,\ell}=\bigcup_{K\in\fset_{j,\ell}}V(K)$, and set $U_{j,\ell}$ containing all remaining vertices of $\Lambda_{j,\ell}$.
(We note that previously, we defined $\hat{K}_{j,\ell}$ to denote the graph $\bigcup_{K\in\fset_{j,\ell}}K$;
we slightly abuse the notation here by letting $\hat{K}_{j,\ell}$ denote the set of vertices of this graph). For every vertex $u\in\Lambda_{j,\ell}$,
let $\deg_{\le(j,\ell)}^{(0)}(u)$ and $\deg_{\le j}^{(0)}(u)$ denote
$\deg_{\le(j,\ell)}(u)$ and $\deg_{\le j}(u)$ at the beginning of
the current $(j,\ell)$-phase, respectively. Recall that, from \Cref{obs:no downward degree}, $\deg_{\le(j,\ell)}^{(0)}(u)=\deg_{\le j}^{(0)}(u)$.
We maintain the following invariant:

\begin{properties}[2]{I}
\item For every vertex $u\in U_{j,\ell}$, $\deg_{\le(j,\ell)}(u)\ge\deg_{\le(j,\ell)}^{(0)}(u)/\gapU$ holds throughout the execution of a $(j,\ell)$-phase. \label{inv: high degree of non-core vertices}
\end{properties}

We now consider the buffer sublayer $\Lambda_{j}^{-}$. The vertices of the buffer sublayer are partitioned into three disjoint subsets: $\hat{K}_{j}^{-}$, $U_{j}^{-}$, and $D_{j}^{-}$, that are defined as follows. First, for all $1\leq \ell<L_j$, whenever any vertex $u$ is pruned out of any core $K\in \fset_{j,\ell}$  by the core pruning algorithm from \Cref{thm: expander pruning} over the course of the current $(j,\ell)$-phase, vertex $u$ is deleted from sublayer $\Lambda_{j,\ell}$ and is added to the buffer sublayer $\Lambda_j^-$, where it joins the set $\hat K_j^-$ (recall that, once the current $(j,\ell)$-phase terminates, we set $\Lambda_j^-=\emptyset$).
Additionally,  whenever Invariant \ref{inv: high degree of non-core vertices} is violated for any vertex $u\in U_{j,\ell}$, we delete $u$ from $\Lambda_{j,\ell}$, and add it to $\Lambda_j^-$, where it joins the set $U_{j}^{-}$. Lastly, for all $j'<j$, whenever a vertex $u\in \Lamda_{j'}$ has its virtual degree decrease from $h_{j'}$ to $h_j$, vertex $u$ is added to layer $\Lambda_j$, into the buffer sublayer $\Lambda_j^-$, where it joins the set $D_{j}^{-}$. Similarly, whenever a vertex $u\in \Lambda_j$ has its virtual degree decrease below  $h_{j}$, we delete it
from $\Lambda_{j}$ and move it to the appropriate layer, where it joints the corresponding buffer sub-layer.
 
Whenever a vertex is added to the buffer sublayer $\Lambda_j^-$, we say that a \emph{move into $\Lambda_j^-$} occurs. The following lemma, that is key to the analysis of the algorithm, bounds the number of such moves. Recall that we used $n_{\le j}$ to denote the number of vertices  
in $\Lambda_{\le j}$ at the beginning of the algorithm, before any edges are deleted from $G$. 

\begin{lemma}
	\label{lem:bound move}For all $1\leq j\leq r$, the total number of moves into $\Lambda_{j}^-$ over the course of the entire algorithm is at most $O(n_{\le j}\Delta / \phi^3)$ = $\Ohat(n_{\le j}\Delta)$.
\end{lemma}

We defer the proof of the lemma to \Cref{subsec:move analysis}, after we show, in \Cref{subsec:bound LCD}, several immediate useful consequences of the lemma.

\subsection{Bounding the Number of Phases and the Number of Cores}

\label{subsec:bound LCD}

In this subsection we use \Cref{lem:bound move} to bound the number of phases and the number of cores for each sublayer $\Lambda_{j,\ell}$.
Recall that in \Cref{subsec:maintain core} we have described an algorithm for maintaining each core $K\in\fset_{j,\ell}$ over the course of a $(j,\ell)$-phase. 
The set $\fset_{j,\ell}$ of cores is initialized using the Core Decomposition algorithm from \Cref{thm: core decomposition}, at the beginning of a $(j,\ell)$-phase. After that, every core $K\in \fset_{j,\ell}$ only undergoes edge and vertex deletions, over the course of the $(j,\ell)$-phase. Once $K=\emptyset$ holds, or a new $(j,\ell)$-phase starts, we say that the core is \emph{destroyed}. Recall that $\fset_{j}=\fset_{j,1}\cup\dots\cup\fset_{j,L_{j}-1}$
denotes the set of all cores in $\Lambda_{j}$ (we do not perform core decomposition in the buffer sublayer $\Lambda_{j,L_{j}}$). In this section, using \Cref{lem:bound move}, we
bound the total number of cores in $\Lambda_{j}$ that are ever created over the course of the algorithm, and also the number of $(j,\ell)$-phases, for all $1\leq\ell<L_j$, in the next two lemmas.

\begin{lemma}
	\label{cor:bound phase}For  all $1\leq j\leq r$ and $1\leq \ell<L_j$, the total number of $(j,\ell)$-phases over the course of the algorithm is at most $\Ohat(2^{\ell}\Delta)$.
\end{lemma}

\begin{proof}
	Fix a pair of indices $1\leq j\leq r$, $1\leq \ell<L_j$.
	Recall that we start a new $(j,\ell)$-phase only when Invariant  \ref{inv: size of every sublayer} is violated for sublayer $(j,\ell+1)$, that is, $|\Lambda_{j,\geq \ell+1}|> n_{\le j}/2^{\ell}$ holds. At the beginning of a $(j,\ell)$-phase, we set $\Lambda_{j,\ell'}=\emptyset$ for all $\ell'>\ell$, and so in particular, $\Lambda_{j,\geq \ell+1}=\emptyset$. The only way that new vertices are added to set $\Lambda_{j,\geq \ell+1}$ is when new vertices join the buffer layer $\Lambda_j^-$, that is, via a move into the buffer sublayer. Therefore, at least $n_{\le j}/2^{\ell}$ moves into the buffer sublayer $\Lambda_j^-$ are required before the current $(j,\ell)$-phase terminates. 
	Since, from \Cref{lem:bound move}, the total number of moves into $\Lambda_j^-$ is bounded by 
	$\Ohat(n_{\le j}\Delta)$, the total number of $(j,\ell)$-phases is bounded by $\Ohat(2^{\ell}\Delta)$.
\end{proof}

Lastly, we bound the total number of cores in $\fset_j$ that are ever created over the course of the algorithm in the following lemma.
\begin{lemma}
	\label{cor:bound core}The total number of cores  created in layer
	$\Lambda_j$ over the course of the entire algorithm is at most $\Ohat(n_{\le j}\cdot\Delta/h{}_{j})$. 
\end{lemma}

\begin{proof}
	Consider an index $1\leq \ell<L$. 
	\thatchapholnew{By \Cref{obs:bound core moment}, at the beginning of every $(j,\ell)$-phase, $|\fset_{j,\ell}|\le O(|\Lambda_{j,\ell}|/(\phi^2 h_{j}))=\Ohat(n_{\le j}/(2^{\ell}h_{j}))$ holds} (we have used Invariant  \ref{inv: size of every sublayer} to bound $|\Lambda_{j,\ell}|$ by $n_{\leq j}/2^{\ell-1}$).
	Since the total number of $(j,\ell)$-phases over the course of the algorithm is bounded by $\Ohat(2^{\ell}\Delta)$, the total number of cores that are ever created in sublayer $\Lambda_{j,\ell}$ is bounded by $\Ohat(n_{\le j}\Delta/h{}_{j})$. The
	claim follows by summing over all $L_{j}-1=O(\log n)$ sublayers.
\end{proof}

\subsection{Bounding the Number of Moves into the Buffer Sublayers: Proof of \Cref{lem:bound move}}
\label{subsec:move analysis}

The goal of this subsection is to prove \Cref{lem:bound move}. Throughout this subsection, we fix an index $1\leq j\leq r$. Our goal is to prove that the total number of moves into the buffer sublayer $\Lambda_j^-$ over the course of the entire algorithm is bounded by $\Ohat(n_{\le j}\Delta)$.
We partition all moves into the buffer sublayer $\Lambda_{j}^-$ into three types. A move of a vertex $u$ into sublayer $\Lambda_j^-$ is of \emph{type-$D$}, if $u$ is added to set  $D_{j}^{-}$; it is of type-$K$, if $u$ is added to  $K_{j}^{-}$; and it is of type-$U$ if $u$ is added to  set $U_{j}^{-}$.
We now bound the number of moves of each type separately.

\paragraph*{Type-$D$ Moves.}

Recall that a vertex $u$ is added to $D_j^-$ only if its virtual degree decreases from some value $h_j'$ for $j'<j$ to $h_j$. Since virtual degrees of all vertices only decrease, such a vertex must lie in $\Lambda_{\leq j}$ at the beginning of the algorithm, and each such vertex $u$ may only be added to set $D^-_j$ once over the course of the algorithm. Therefore, the total number of type-$D$ moves into $\Lambda_j^-$ is bounded by $n_{\le j}$.

\paragraph*{Type-$K$ Moves.}
To bound the number of type-$K$ moves, it is convenient to assign
\emph{types} to edge deletions. Consider an index $1\leq \ell<L_j$ and the corresponding graph $H_{j,\ell}=G[\Lambda_{j,\ell}]$.
Let $e$ be an edge deleted from $H_{j,\ell}$. We assign to the edge $e$ one of the following four deletion types.

\begin{itemize}
	\item If $e$ is deleted by the adversary (that is, $e$ is deleted as part of the deletion sequence of the input graph $G$), then this deletion is of \emph{type-$A$};
	
	\item If $e$ is deleted from $H_{j,\ell}$ because the virtual degree of one of its endpoints falls below $h_{j}$ (and so that endpoint is deleted from $\Lambda_{j}$), then this deletion is 
	of \emph{type-$D$};

	\item If an endpoint of $e$ belonged to some core $K\in \fset_{j,\ell}$, and is then pruned out of that core (and so that endpoint
	is removed from $\Lambda_{j,\ell}$ and added to $K_{j}^{-})$, then the deletion of edge $e$ is of \emph{type-$K$};
	\item Lastly, if an endpoint $u$ of $e$ lies in $u\in U_{j,\ell}$, and Invariant \ref{inv: high degree of non-core vertices} stops holding for $u$, that is,
	$\deg_{\le(j,\ell)}(u)<\deg_{\le(j,\ell)}^{(0)}(u)/\gapU$  holds (and so 
	$u$ is removed from $\Lambda_{j,\ell}$ and added to $U_{j}^{-}$), then the deletion of $e$ is of \emph{type-$U$}.
\end{itemize}

Observe that every edge deletion from a graph $H_{j,\ell}$ executed over the course of a $(j,\ell)$-phase must fall under one of these four categories. As the algorithm progresses, the same vertex may be added to and deleted from sublayer $\Lambda_{j,\ell}$ multiple times. Therefore, an edge that is deleted from $H_{j,\ell}$ may be re-added to the new graph $H_{j,\ell}$ at the beginning of one of the subsequent $(j,\ell)$ phases.
Next, we bound the total number of edge deletions from all graphs $H_{j,\ell}$ over the course of the entire algorithm, for the first three types. Notice that these edge deletions ignore the deletions of edges whose endpoints belong to different sublayers.

The following simple observation bounds the number of type-$A$ and type-$D$ deletions.

\begin{observation}\label{obs: bound edge deletions of types A and D}
	The total number of type-$A$ and type-$D$ deletions from all graphs $H_{j,\ell}$ for all $1\leq \ell<L_j$, over all $(j,\ell)$-phases is bounded by $n_{\le j}h_{j}\Delta$.
\end{observation}
\begin{proof}
Observe that each type-$A$ deletion corresponds to a deletion of an edge whose both endpoints are contained in $\Lambda_j$ from the input graph $G$; each such edge may only be deleted once over the course of the algorithm. From \Cref{claim:bound edge}, the total number of such edges is bounded by $n_{\le j}h_{j}\Delta$.

If an edge $e$ is deleted in a type-$D$ deletion from some graph $H_{j,\ell}$, then both its endpoints lie in $\Lambda_j$, and, after the deletion, one of the endpoints of $e$ is removed from layer $\Lambda_j$ forever. Therefore, every edge may be deleted at most once in a type-$D$ deletion, and the number of all such deletions is bounded by the total number of edges whose both endpoints are contained in $\Lambda_j$ over the course of the algorithm, which is again bounded by $n_{\le j}h_{j}\Delta$ from  \Cref{claim:bound edge}.
\end{proof}

We now proceed to bound the total number of type-$K$ edge deletions.

\begin{lemma}
	\label{lem:bound deletion typeK}  The total numbers of type-$K$ edge deletions 
	 from all graphs $H_{j,\ell}$ for all $1\leq \ell<L_j$, over all $(j,\ell)$-phases is bounded by $O(n_{\le j}h_{j}\Delta/\phi^{2})$.
\end{lemma}

\begin{proof}
	Consider an index $1\leq \ell<L_j$, and some $(j,\ell)$-phase. 
	 Let $H_{j,\ell}^{(0)}$
	denote  the graph $H_{j,\ell}$ at the beginning of the $(j,\ell)$-phase, and let
	$K^{(0)}$ denote an arbitrary core $K\in \fset_{j,\ell}$ at the beginning of that phase. Let $S^{K}$ be the set of vertices that are pruned out of $K^{(0)}$ by
	\Cref{thm:shortkquery}. By the definition of type-$K$ deletion, it
	is enough to bound $\vol_{H_{j,\ell}^{(0)}}(S^{K})$, summing
	over cores $K$ created in $H_{j,\ell}$ over the course of the algorithm.
	
	In order to bound $\vol_{H_{j,\ell}^{(0)}}(S^{K})$ for a single core $K\in \fset_{j,\ell}$, recall that, from \Cref{thm: core decomposition}, for every vertex $u\in V(K^{(0)})$
	$\deg_{K^{(0)}}(u)\ge\phi\deg_{H_{j,\ell}^{(0)}}(u)/\gapdegree$ holds. 
	Therefore, $\vol_{H_{j,\ell}^{(0)}}(S^{K})\le\gapdegree\cdot\vol_{K^{(0)}}(S^{K})/\phi$.
	Moreover, by \Cref{thm:shortkquery}, after $t$ edge deletions from $K^{(0)}$ (that include type-$A$ and type-$D$ deletions, but exclude type-$U$ deletions, as edges deleted this way must lie outside of the core), we have $\vol_{K^{(0)}}(S^{K})\leq O(t/\phi)$.
	From \Cref{obs: bound edge deletions of types A and D}, the total number of type-$A$ and type-$D$ edge deletions, in all graphs $H_{j,\ell}$,
	for all $1\le\ell<L_{j}$ and all $(j,\ell)$-phases, is at most $n_{\le j}h_{j}\Delta$. We conclude that the sum of all volumes $\vol_{H_{j,\ell}^{(0)}}(S^{K})$, over all $1\leq \ell<L_j$, and over all cores ever created in $\fset_{j,\ell}$, is bounded by: $\frac{\gapdegree}{\phi}\cdot O(n_{\le j}h_{j}\Delta/\phi)=O(n_{\le j}h_{j}\Delta/\phi^{2})$.
\end{proof}

\begin{corollary}\label{cor: num of K-moves}
	 The total number of type-$K$ moves into $\Lambda_j^-$ over the course of the algorithm is bounded by $O(n_{\le j}\Delta/\phi^{3})$.
	 \end{corollary}
 \begin{proof}	
 	Consider some sublayer $\Lambda_{\ell,j}$, and the graph $H_{j,\ell}^{(0)}$ at the beginning of some $(j,\ell)$-phase. Let $K^{(0)}\in \fset_{j,\ell}$ be some core that was created at the beginning of that phase, and let $u\in V(K^{(0)})$ be any vertex of the core. Recall that, from \Cref{thm: core decomposition}, 	$\deg_{K^{(0)}}(u)\ge\phi\deg_{\leq (j,\ell)}^{(0)}(u)/\gapdegree$, and from \Cref{obs:no downward degree}, $\deg_{\leq (j,\ell)}^{(0)}(u)=\deg_{\leq j}^{(0)}(u)\ge h_{j}$. Therefore, $\deg_{K^{(0)}}(u)\ge \phi h_j/\gapdegree$. If vertex $u$ is moved to $\Lambda_j^-$ via a type-$K$ move, then each of the $\Omega(\phi h_j)$ edges of $K^{(0)}$ incident to $u$ must have been deleted as part of $A$, $D$, or $K$-type deletion from $H_{j,\ell}$. Since the total number of all deletions of types $A$, $D$ and $K$, from all graphs $H_{j,\ell}$ for $1\leq \ell<L_j$, over the course of the whole algorithm, is bounded by $O(n_{\le j}h_{j}\Delta/\phi^{2})$, we get that the total number of $K$-type moves into the buffer layer $\Lambda_j^-$ over the course of the entire algorithm is bounded by $O(n_{\le j}\Delta/\phi^{3})$.
 \end{proof}

\paragraph*{Type-$U$ Moves.}
We further partition type-$U$ moves into two subtypes. Consider a time in the algorithm's execution, when some vertex $u$ is added to set $U_j^-$, so it is added to $\Lambda_j^-$ via a $U$-move, and assume that $u$ was moved from sublayer $\Lambda_{j,\ell}$. We say that this move is of type-$U_{1}$ if $|E_{G}(u,\Lambda_{j,>\ell})|<2\deg_{\le(j,\ell)}(u)$ held
right before $u$ is moved, and we say that it is of type-$U_{2}$ otherwise. We bound
the number of moves of both subtypes separately. 

We first bound the number of type-$U_{1}$ moves. Recall that, whenever $u$
is type-$U_{1}$ moved:
\begin{align*}
\deg_{\le j}(u) & =\deg_{\le(j,\ell)}(u)+|E_{G}(u,\Lambda_{j,>\ell})|\\
& <3\deg_{\le(j,\ell)}(u)\\
& <3\deg_{\le(j,\ell)}^{(0)}(u)/\gapU\\
& =3\deg_{\le j}^{(0)}(u)/\gapU
\end{align*}

(we have used the fact that Invariant \ref{inv: high degree of non-core vertices} is violated when $u$ is moved to $\Lambda_j^-$, and \Cref{obs:no downward degree} for the last equality.)
Therefore, the number of neighbors of $u$ in $\Lambda_{\le j}$ has
reduced by a constant factor compared to the beginning of the current $(j,\ell)$-phase. Note that $\deg_{\le j}(u)$ may never increase, as virtual degrees of vertices may only decrease. Therefore, a vertex may be moved to $\Lambda_j^-$ via a type-$U_1$ move at most $O(\log n)$ times.
The total number of type-$U_{1}$ moves into $\Lambda_j^-$ over the course of the entire algorithm is then bounded by $O(n_{\le j}\log n)$. 

It remains to bound the the total number of type-$U_{2}$ moves into $\Lambda_j^-$. We
will show that the total number of such moves is bounded by $O(n_{\le j}\Delta/\phi)$ over the course of the algorithm.
For all $1\leq \ell<L_j$, we define an edge set $\Pi_{j,\ell}$, that contains all edges $e=(u,v)$ with $u\in \Lambda_{j,\ell}$ and $v\in \Lambda_{j,>\ell}$. Intuitively, set $\Pi_{j,\ell}$ contains all ``downward edges'' from vertices in $\Lambda_{j,\ell}$ that lie \emph{within the layer $j$}.
Note that $\Pi_{j,L_{j}}=\emptyset$. We first bound the total number of edges that ever belonged to each such set  $\Pi_{j,\ell}$ over the course of the algorithm. (Note that an edge may be added several times to $\Pi_{j,\ell}$ over the course of the algorithm; we count them as separate edges). We will then use this bound in order to bound the total number of the $U_2$-moves.

Note that a new edge $e=(u,v)$ may only be added to set $\Pi_{j,\ell}$ in the following cases:

\begin{itemize}
	\item \textbf{(Type-$D$ addition):} when some vertex $v\in \Lambda_{<j}$ is
	type-$D$ moved to $D_{j}^{-}\subseteq\Lambda_{j}^-$. We denote the set of all edges added to $\Pi_{j,\ell}$ in a type-$D$ addition by $\Pi^D_{j,\ell}$.
	\item \textbf{(Type-$K$ addition):} when some vertex $v\in \Lambda_{j,\le\ell}$
	is type-$K$ moved to $K_{j}^{-}\subseteq\Lambda_{j}^-$. We denote the set of all edges added to $\Pi_{j,\ell}$ in a type-$K$ addition by $\Pi^K_{j,\ell}$.
	\item \textbf{(Type-$U_{1}$ addition):} when some vertex  $v\in \Lambda_{j,\le\ell}$
	is type-$U_{1}$ moved to $U_{j}^{-}\subseteq\Lambda_{j}^-$. We denote the set of all edges added to $\Pi_{j,\ell}$ in a type-$U_1$ addition by $\Pi^{U_1}_{j,\ell}$.
	
	\item \textbf{(Type-$U_{2}$ addition):} when  some vertex $v\in \Lambda_{j,\le\ell}$
	is type-$U_{2}$ moved to $U_{j}^{-}\subseteq\Lambda_{j}^-$. 
	We denote the set of all edges added to $\Pi_{j,\ell}$ in a type-$U_2$ addition by $\Pi^{U_2}_{j,\ell}$.
\end{itemize}

Observe that the core decomposition algorithm from \Cref{thm: core decomposition} that is performed at the beginning of a $(j,\ell)$-phase does not add any new edges to any set $\Pi_{j,\ell'}$, since, for $\ell'>\ell$, we set $\Lambda_{j,\ell'}=\emptyset$, and for $\ell'\leq \ell$, vertex set $\Lambda_{j,\geq \ell'}$ remains unchanged.

Let $\inc_{j,\ell}$ denote the total  number of edges ever added to $\Pi_{j,\ell}$, and for every addition type  $X\in\{D,K,U_{1},U_{2}\}$, we denote by $\inc^X_{j,\ell}$ the total number of edges ever added to $\Pi_{j,\ell}^X$. Clearly, $\inc_{j,\ell}=\sum_{X\in\{D,K,U_{1},U_{2}\}}\inc_{j,\ell}^{X}$.
For convenience, for all $X\in\{D,K,U_{1},U_{2}\}$,
we denote by $\Pi_{j}^{X}=\bigcup_{1\leq \ell<L_j}\Pi_{j,\ell}^{X}$.
Let $\inc_{j}$ denote
the total number of edges ever added to any of the sets in $\set{\Pi_{j,\ell}\mid 1\leq \ell<L_j}$. Similarly, for all $X\in\{D,K,U_{1},U_{2}\}$, we denote by $\inc^X_{j}$ the total number of edges ever added to any of the sets in $\set{\Pi^X_{j,\ell}\mid 1\leq \ell<L_j}$.
Observe that $\inc_{j}=\sum_{X\in\{D,K,U_{1},U_{2}\}}\inc_{j}^{X}$.
We start by bounding $\inc_{j}^{D}$, $\inc_{j}^{U_1}$, and $\inc_{j}^{K}$:

\begin{lemma}
	\label{lem:bound increase type easy}$\inc_{j}^{D}\le O(n_{\le j}h_{j}\Delta)$,
	$\inc_{j}^{U_{1}}\le O(n_{\le j}h_{j}\log n)$ and $\inc_{j}^{K}\leq O(n_{\le j}h_{j}\Delta/\phi^{2})$. 
\end{lemma}

\begin{proof}
	When a vertex $u$ is moved to $\Lambda_j^-$, its contribution to $\inc_{j}$ is at
	most $|E(u,\Lambda_{j})|$. By \Cref{claim:bound edge}, the total
	number of edges $e$, such that, at any point during the algorithm's execution, both endpoints of $e$ were contained  $\Lambda_{j}$, is
	at most $n_{\le j}h_{j}\Delta$. As each vertex $u$ can 
	moved to $\Lambda_{j^-}$ in a type-$D$ move only once,  $\inc_{j}^{D}\le O(n_{\le j}h_{j}\Delta)$ must hold.
	Similarly, as each vertex $u$ can be moved to $\Lambda_{j}^{-}$ in a type-$U_1$ move
	at most $O(\log n)$ times, $\inc_{j}^{U_1}\le O(n_{\le j}h_{j}\Delta\log n)$ must hold.
	
	Next, observe that $\inc_{j,\ell}^{K}$ is precisely the total number
	of type-$K$ edge deletions from graphs $H_{j,1},\ldots,H_{j,\ell}$ over the course of the algorithm.
	By \Cref{lem:bound deletion typeK}, we can bound $\inc_{j}^{K}\leq O(n_{\le j}h_{j}\Delta/\phi^{2})$.
\end{proof}

We use the next lemma to bound $\inc_{j,\ell}^{U_{2}}$.
\begin{lemma}
	\label{lem:bound increase typeU}For all $1\leq \ell<L_j$,
	$\inc_{j,\ell}^{U_{2}}\le\inc_{j,\ell}^{D}+\inc_{j,\ell}^{K}+\inc_{j,\ell}^{U_{1}}$.
\end{lemma}
\begin{proof}
	Fix an index $1\leq \ell<L_j$. We denote by $\hat \Pi_{j,\ell}$ the set of all edges that were ever present in set $\Pi_{j,\ell}$ over the course of the algorithm, and by $\hat \Pi_{j,\ell}^{U_2}$ the set of all edges that were ever present in set $\Pi_{j,\ell}^{U_2}$ over the course of the entire algorithm; recall that $\hat \Pi_{j,\ell}^{U_2}\subseteq \hat \Pi_{j,\ell}$. We next show that $|\hat \Pi_{j,\ell}^{U_2}|\leq |\hat \Pi_{j,\ell}|/2$. 
	
	In order to do so, we assign, to every edge $e\in \hat \Pi_{j,\ell}^{U_2}$, two edges $e^1,e^2\in \hat \Pi_{j,\ell}$ that are \emph{responsible} for $e$. We will ensure that every edge in $\hat \Pi_{j,\ell}$ is responsible for at most one edge in $\hat \Pi^{U_2}_{j,\ell}$. This will immediately imply that $|\hat \Pi_{j,\ell}^{U_2}|\leq |\hat \Pi_{j,\ell}|/2$.
	
	Consider now some $(j,\ell)$-phase, and some vertex $u\in \Lambda_{j,\ell}$, that is moved to $\Lambda_j^-$ via a $U_2$-move some time during the $(j,\ell)$-phase. 
	 At the beginning of the $(j,\ell)$-phase, $\Lambda_{j,\ell'}=\emptyset$ held
	for all $\ell'>\ell$. At the time when $u$ is moved to $\Lambda_j^-$, from the definition of a $U_2$-move, $|E_{G}(u,\Lambda_{j,>\ell})|\geq 2\deg_{\le(j,\ell)}(u)$ held. The edges that are added to $\hat \Pi_{j,\ell}^{U_2}$ due to the move of $u$ to $\Lambda_j^-$ are the edges of $E_{G}(u,\Lambda_{j, \ell})$. On the other hand, each edge $(u,v)\in E_{G}(u,\Lambda_{j,>\ell})$ belonged to set $\Pi_{j,\ell}$ before the move of $u$, and is removed from that set afterwards. Moreover, vertex $v$ must have been moved to $\Lambda_j^-$ at some time during the course of the current $(j,\ell)$-phase. Since $|E_{G}(u,\Lambda_{j,>\ell})|\geq 2\deg_{\le(j,\ell)}(u)\geq 2\deg_{(j,\ell)}(u)$, we can select, for every edge $e\in E_{G}(u,\Lambda_{j, \ell})$ arbitrary two edges $e^1,e^2\in E_{G}(u,\Lambda_{j,>\ell})$ that become responsible for $e$, such that every edge of $E_{G}(u,\Lambda_{j,>\ell})$ is responsible for at most one edge of $E_{G}(u,\Lambda_{j, \ell})$. It is clear from this process that every edge in $\hat \Pi_{j,\ell}$ is responsible for at most one edge in $\hat \Pi^{U_2}_{j,\ell}$. We conclude that $|\hat \Pi_{j,\ell}^{U_2}|\leq |\hat \Pi_{j,\ell}|/2$ holds, and so $\inc_{j,\ell}^{U_2}\leq \inc_{j,\ell}/2$. Since $\inc_{j,\ell}=\sum_{X\in\{D,K,U_{1},U_{2}\}}\inc_{j,\ell}^{X}$, we get that $\inc_{j,\ell}^{U_2}\leq \inc_{j,\ell}^{D}+\inc_{j,\ell}^{K}+\inc_{j,\ell}^{U_{1}}$.
	\end{proof}

Combining \Cref{lem:bound increase type easy} and \Cref{lem:bound increase typeU}, we obtain the following corollary.

\begin{corollary}
	\label{cor:bound increase}  $\inc_{j}\leq O(n_{\le j}h_{j}\Delta/\phi^{3})$. 
\end{corollary}

Lastly, the following corollary allows us to bound the total number of $U_2$-moves.

\begin{corollary}
	\label{cor:U2 moves}The total number of $U_2$-moves into the buffer layer $\Lamda_j^-$, over the course of the entire algorithm, is bounded by $O(n_{\le j}\Delta/\phi^{3})$. 
\end{corollary}
\begin{proof}
	Consider some index $1\leq \ell<L_j$, some $(j,\ell)$-phase, and some vertex $u\in \Lambda_{j,\ell}$, that is moved to $\Lambda_j^-$ via a $U_2$-move some time during that $(j,\ell)$-phase. 
	From the definition of a $U_2$-move, when $u$ was moved to $\Lambda_j^-$, $|E_{G}(u,\Lambda_{j,>\ell})|\geq 2\deg_{\le(j,\ell)}(u)$ held. Moreover,  each edge $(u,v)\in E_{G}(u,\Lambda_{j,>\ell})$ belonged to set $\Pi_{j,\ell}$ before the move of $u$, and is removed from that set afterwards. Since $|E_{G}(u,\Lambda_{j,>\ell})|\ge2\deg_{\le(j,\ell)}(u)$
and $|E_{G}(u,\Lambda_{j,>\ell})|+\deg_{\le(j,\ell)}(u)=\deg_{\le j}(u)$,
we get that $|E_{G}(u,\Lambda_{j,>\ell})|\ge\deg_{\le j}(u)/3$. From the definition of virtual degrees, $\deg_{\le j}(u)\geq h_j$ must hold, and so $|E_{G}(u,\Lambda_{j,>\ell})|\ge h_{j}/3$. Therefore, for every vertex that is added to $\Lambda_j^-$ via a $U_2$-move, we delete at least $h_j/3$ edges from set $\Pi_{j,\ell}$. Since, as shown above, the total number of edges that are ever added to sets $\Pi_{j,\ell}$, for all $1\leq\ell<L_j$ is bounded by  $O(n_{\le j}h_{j}\Delta/\phi^{3})$, the total number of $U_2$-moves into $\Lambda_j$ over the entire course of the algorithm is bounded by $O(n_{\le j}\Delta/\phi^{3})$.
\end{proof}

To summarize, we have partitioned all moves into the buffer sublayer $\Lambda_j^-$ into four types: $D,A,K$ and $U$, and we showed that the total number of moves of each type, over the course of the algorithm, is bounded by $O(n_{\le j}\Delta/\phi^{3})$. Therefore, the total number of moves into $\Lambda_j^-$ over the course of the algorithm is at most $O(n_{\le j}\Delta/\phi^{3})\leq \Ohat(n_{\le j}\Delta)$.

\subsection{Existence of Short Paths to the Cores}
\label{subsec: tocore exist}
The main result of this subsection is summarized in the following lemma, which shows that, throughout the algorithm, for every vertex $v\in V(G)$, there is a path of length $O(\log^{3}n)$, connecting $v$ to some vertex that lies in one of the cores of $\bigcup_j\fset_j$. This fact will be used to process $\tocore$ queries. 

\begin{lemma}
	\label{cor:up path}Throughout the algorithm, for each vertex $u\in V(G)$,
	there is a path $P_{u}$ of length at most $O(\log^{3}n)$, connecting $u$ to a vertex $v$ lying in set $\hat{K}=\bigcup_{j,\ell}\hat{K}_{j,\ell}$. Moreover, the
	path $P_{u}=\{u=u_{1},u_{2},\dots,u_{k}\}$ is non-decreasing with respect to the sublayers, that is, if $u_{i}\in\Lambda_{j,\ell}$,
	then $u_{i+1}\in\Lambda_{<j}\cup\Lambda_{j,\le\ell}$. 
\end{lemma}

\begin{proof}
We use the following two claims.
\begin{claim}
	\label{lem:up path buffer}Throughout the algorithm, for all $1\leq j\leq r$, every vertex $u$ in the buffer sublayer $\Lambda_{j,L_{j}}$
	has a neighbor in $\Lambda_{<j}\cup\Lambda_{j,<L_{j}}$.
\end{claim}

\begin{proof}
	Recall that Invariant \ref{inv: size of every sublayer} guarantees that $|\Lambda_{j,L_j}|\leq n_{\le j}/2^{L_j-1}$, and, from the choice of $L_j$, $n_{\le j}/2^{L_j-1}\leq h_j/2$. From the definition of virtual degrees of vertices, for every vertex $u\in \Lambda_j^-$,
	$|E_{G}(u,\Lambda_{\le j})\ge h_{j}$ must hold. Therefore, $u$ must have a neighbor in $\Lambda_{<j}\cup\Lambda_{j,<L_{j}}$. 
\end{proof}
\begin{claim}
	\label{lem:up path non buffer}Throughout the algorithm, for all $1\leq \ell<L_j$, every
	vertex $u\in U_{j,\ell}$ has a path $P_{u}$ of length at most $O(\log n)$ connecting it
	to a vertex in $\Lambda_{<j}\cup\Lambda_{j,<\ell}\cup\hat{K}_{j,\ell}$,
	such that every inner vertex on the path belongs to $U_{j,\ell}$.
\end{claim}

\begin{proof}
Fix an index $1\leq \ell<L_j$, and consider the sublayer $\Lambda_{j,\ell}$ over the course of some $(j,\ell)$-phase.
	Below, we add a superscript $(0)$ to an object to denote that object
	at the beginning of the phase.
	Recall that the algorithm for computing a core decomposition from \Cref{thm: core decomposition} ensured
	then there is an orientation of the edges of the graph  $G[U_{j,\ell}]$, such that the resulting directed graph $\dset_{j,\ell}$ is a DAG, and, for every vertex $u\in U_{j,\ell}$, $\indeg_{\dset_{j,\ell}}(u)\le\deg_{\le(j,\ell)}^{(0)}(u)/\gapdegree$. We denote by $\dset_{j,\ell}^{(0)}$ the graph $\dset_{j,\ell}$ at the beginning of the phase.

 Let $\overline{\dset}_{j,\ell}^{(0)}$
	be the directed graph obtained from $\dset_{j,\ell}^{(0)}$, by adding the set $\Lambda_{<j}^{(0)}\cup\Lambda_{j,<\ell}^{(0)}\cup\hat{K}_{j,\ell}^{(0)}$
	of vertices to it, and all edges present in graph $G$ at the beginning of the current $(j,\ell)$-phase, connecting the vertices of $U_{j,\ell}^{(0)}$ to the vertices of $\Lambda_{<j}^{(0)}\cup\Lambda_{j,<\ell}^{(0)}\cup\hat{K}_{j,\ell}^{(0)}$. We orient these edges away from the vertices of  $U_{j,\ell}$. Therefore, for every vertex $u\in U_{j,\ell}$,  $\indeg_{\overline{\dset}_{j,\ell}^{(0)}}(u)=\indeg_{\dset_{j,\ell}^{(0)}}(u)$
	holds.
	
	Let $\overline{\dset}_{j,\ell}$ denote the graph $\overline{\dset}_{j,\ell}^{(0)}$
	at some time during the execution of the $(j,\ell)$-phase. Whenever an edge incident
	to a vertex of $\Lambda_{j,\ell}$ is deleted by the algorithm, we delete this edge from graph $\overline{\dset}_{j,\ell}$ as well. Whenever a vertex of $\Lambda_{j,\ell}$ is removed from this set, we delete such a vertex and all its incident edges from $\overline{\dset}_{j,\ell}$. From the definition of $\overline{\dset}_{j,\ell}$,
	for every vertex  $u\in U_{j,\ell}$, $\indeg_{\overline{\dset}_{j,\ell}}(u)+\outdeg_{\overline{\dset}_{j,\ell}}(u)=\deg_{\le(j,\ell)}(u)$
	holds at all times. 
	
	Observe that, for every $u\in U_{j,\ell}$:
	\begin{align*}
	\indeg_{\overline{\dset}_{j,\ell}}(u) & \le\indeg_{\overline{\dset}_{j,\ell}^{(0)}}(u) & \text{as }\overline{\dset}_{j,\ell}\subseteq\overline{\dset}_{j,\ell}^{(0)}\\
	& \le\deg_{\le(j,\ell)}^{(0)}(u)/\gapdegree & \text{by the property of }\overline{\dset}_{j,\ell}^{(0)}\\
	& \le\deg_{\le(j,\ell)}(u)/3 & \text{by Invariant \ref{inv: high degree of non-core vertices}},
	\end{align*}
	and so
	\begin{equation}
	\outdeg_{\overline{\dset}_{j,\ell}}(u)=\deg_{\le(j,\ell)}(u)-\indeg_{\overline{\dset}_{j,\ell}}(u)\ge2\cdot\indeg_{\overline{\dset}_{j,\ell}}(u).\label{eq:out more than in}
	\end{equation}
	For any vertex set $S\subseteq V(\overline{\dset}_{j,\ell})$, let $\invol_{\overline{\dset}_{j,\ell}}(S)=\sum_{u\in S}\indeg_{\overline{\dset}_{j,\ell}}(u)$,
	$\outvol_{\overline{\dset}_{j,\ell}}(S)=\sum_{u\in S}\outdeg_{\overline{\dset}_{j,\ell}}(u)$,
	and $\vol_{\overline{\dset}_{j,\ell}}(S)=\invol_{\overline{\dset}_{j,\ell}}(S)+\outvol_{\overline{\dset}_{j,\ell}}(S)$.
	For a vertex set $S \subseteq V(\overline{\dset}_{j,\ell})$, we denote by $S'$ the set of vertices containing all vertices of $S$, and all vertices $v\in V(\overline{\dset}_{j,\ell})$, such that edge $(u,v)$ with $u\in S$ belongs to the graph $\overline{\dset}_{j,\ell}$. In other words, $S'$ is an ``out-ball'' around $S$ of radius 1.
	
	Next, we show that, for any vertex set $S\subseteq U_{j,\ell}$, $\vol_{\overline{\dset}_{j,\ell}}(S')\ge\frac{4}{3}\vol_{\overline{\dset}_{j,\ell}}(S)$.
	Indeed:
	
	\begin{align*}
	\vol_{\overline{\dset}_{j,\ell}}(S') & =\vol_{\overline{\dset}_{j,\ell}}(S)+\vol_{\overline{\dset}_{j,\ell}}(S'\setminus S)\\
	& \ge\vol_{\overline{\dset}_{j,\ell}}(S)+|E_{\overline{\dset}_{j,\ell}}(S,S'\setminus S)|\\
	& =\vol_{\overline{\dset}_{j,\ell}}(S)+|E_{\overline{\dset}_{j,\ell}}(S,S')|-|E_{\overline{\dset}_{j,\ell}}(S,S)|\\
	& \ge\vol_{\overline{\dset}_{j,\ell}}(S)+\outvol_{\overline{\dset}_{j,\ell}}(S)-\invol_{\overline{\dset}_{j,\ell}}(S),
	\end{align*}
	where the last inequality follows from the fact that $|E_{\overline{\dset}_{j,\ell}}(S,S')| = \outvol_{\overline{\dset}_{j,\ell}}(S)$ and $|E_{\overline{\dset}_{j,\ell}}(S,S)| \le \invol_{\overline{\dset}_{j,\ell}}(S)$.
	From  \Cref{eq:out more than in}, $\outvol_{\overline{\dset}_{j,\ell}}(S)\ge2\invol_{\overline{\dset}_{j,\ell}}(S)$.
	Therefore, $\outvol_{\overline{\dset}_{j,\ell}}(S)-\invol_{\overline{\dset}_{j,\ell}}(S)\ge\vol_{\overline{\dset}_{j,\ell}}(S)/3$.
	We conclude that: 
	\[
	\vol_{\overline{\dset}_{j,\ell}}(S')\ge\vol_{\overline{\dset}_{j,\ell}}(S)+\vol_{\overline{\dset}_{j,\ell}}(S)/3=\frac{4}{3}\vol_{\overline{\dset}_{j,\ell}}(S).
	\]

	It is now easy to see that, for	any vertex $u\in U_{j,\ell}$,   the distance from $u$
 to $\Lambda_{<j}\cup\Lambda_{j,<\ell}\cup\hat{K}_{j,\ell}$ in  $\overline{\dset}_{j,\ell}$
	is bounded by $O(\log n)$. Otherwise, we can grow a ball from $u$, of volume
	$(4/3)^{\Omega(\log n)}\ge\poly(n)$, leading to a contradiction.
\end{proof}

We are now ready to complete the proof of \Cref{cor:up path}. Suppose we are given a vertex $u\in V(G)$, and assume that it lies in some sublayer $\Lambda_{j,\ell}$, for $1\leq j\leq r$ and $1\leq \ell\leq L_j$. We start with $v_1=u$, and then repeatedly apply
\Cref{lem:up path buffer} and \Cref{lem:up path non buffer} to the current vertex $v_i$, until
we reach a vertex that lies in some core $\hat{K}$. For each application of the
lemmas, starting from some vertex $v_i$, we  obtain a path connecting $v_i$ to either a vertex that lies in some core, or a vertex that belongs to some sublayer lying above the sublayer of $v_i$. In either case, the inner vertices of the path are contained
in the sublayer of $v_i$. The final path $P_{u}$ is obtained
by concatenating all resulting paths. As there are $O(\log^{2}n)$ sublayers and each
path that we compute has length at most $O(\log n)$, the length of $P_{u}$ is bounded by $O(\log^{3}n)$. Moreover, it is easy to verify that the path visits the sublayers in a non-decreasing order.
\end{proof}

\subsection{The Incident-Edge Data Structures}

\label{subsec: lcd incident ds}

In order to efficiently construct and maintain the graphs $H_{j,\ell} = G[\Lambda_{j,\ell}]$, we maintain the following sets of edges:

\begin{itemize}
	\item for every vertex $u\in V$ and index $1\leq j\leq r$, edge set $\edge_{j}(u)=E_{G}(u,\Lambda_{j})$.
	\item for every pair of indices $1\leq j\leq r$, $1\leq \ell\leq L_j$, for every vertex $u\in\Lambda_{j,\ell}$, edge set $\edge_{j,>\ell}(u)=E_{G}(u,\Lambda_{j,>\ell})$,
	and, for each $1\leq \ell'\le\ell$, edge set $\edge_{j,\ell'}(u)=E_{G}(u,\Lambda_{j,\ell'})$.
\end{itemize}

Intuitively, the data structures are defined in this way because we cannot afford
to maintain the edge set $E_{G}(u,\Lambda_{j,\ell'})$ for every vertex $u\in \Lambda_{j,\ell}$, for all $1\leq \ell< \ell'$, explicitly.
This is because the vertices of $\Lambda_{j,>\ell}$ may move
between the sublayers that lie below $\Lambda_{j,\ell}$ too frequently. 

Throughout, the notation $\edge(\cdot)$ is used for the sets of edges that are explicitly
maintained by the data structure, which we distinguish from subsets of edges of $G$, for which notation $E_{G}(\cdot)$ is used. 

Consider some vertex $u\in V(G)$, and let $\Lambda_{j,\ell}$ be the sublayer containing $u$. We refer to
the edge sets $$\{\edge_{j'}(u)\}_{1\leq j'\leq r},\{\edge_{j,\ell'}(u)\}_{\ell'\le\ell},\edge_{j,>\ell}(u)$$
 as the \emph{incident-edge data structure} of vertex $u$. Below,
we show that the incident-edge data structures for all vertices of $G$ can be maintained in total update time $\tilde{O}(m\Delta^{2})$ over the course of the algorithm's execution.

Recall that, from \Cref{thm: maintaining virtual degrees}, the layers $\Lambda_{1},\dots,\Lambda_{r}$, 
and the edge sets $\edge_{j}(u)$ for all vertices $u\in V(G)$ and layers $1\leq j\leq r$
can be maintained over the course of the algorithm, in time $\tilde{O}(m+n)$. 

Next, we fix an index $1\leq j\leq r$, and show how to maintain the edge sets $\{\edge_{j,\ell'}(u)\}_{\ell'\le\ell},\edge_{j,>\ell}(u)$ 
for all $1\leq \ell\leq L_j$ and $u\in\Lambda_{j,\ell}$. 

At the beginning of the algorithm, all vertices of $\Lambda_j$ lie in the sublayer $\Lambda_{j,1}$. For every vertex $u\in \Lambda_{j,1}$, set $\edge_{j,>1}(u)=\emptyset$, and edge set $\edge_{j,1}(u)$ contains all edges in $E(H_{j,1})$ that are incident to $u$. All these edge sets can be initialized in time $O(|E(H_{j,1})|)$

There are two cases when we need to update
these sets: (1) when a vertex is moved to $\Lambda_{j}^{-}$ and (2)
when we set $\Lambda_{j,\ell}\gets\Lambda_{j,\ge\ell}$ and initialize
the sublayer $\Lambda_{j,\ell}$ for some $\ell$.

For the first event, consider some sublayer $\Lambda_{j,\ell}$, and a vertex $u\in\Lambda_{j,\ell}$
that is moved to the buffer sublayer $\Lambda_{j}^{-}$. We need to partition the edges from the original edge set
$\edge_{j,>\ell}(u)=E_{G}(u,\Lambda_{j,>\ell})$ into edge sets $\edge_{j,\ell+1}(u),\dots,\edge_{j,L_{j}}(u)$,
where $\edge_{j,\ell'}(u)\gets E_{G}(u,\Lambda_{j,\ell'})$ for each
$\ell'>\ell$. Also, for each vertex $v\in\{w\mid(u,w)\in E_{G}(u,\Lambda_{j,>\ell})\}$,
we need to move the edge $(u,v)$ from $\edge_{j,\ell}(v)$ to $\edge_{j,>\ell'}(v)$, where $\ell'$ is the index of the sublayer $\Lambda_{j,\ell'}$ containing $v$; if $v \in \Lambda_{j,L_{j}}$, then we instead move $(u,v)$ from $\edge_{j,\ell}(v)$ to $\edge_{j,L_j}(v)$.
All these operations can be done in time $|E_{G}(u,\Lambda_{j,>\ell})|$.

Recall that every time we move $u$ from some sublayer $\Lambda_{j,\ell}$ to the buffer sublayer $\Lambda_{j}^{-}$,
we remove $|E_{G}(u,\Lambda_{j,>\ell})|$ edges from the edge set $\Pi_{j,\ell}$ that we defined in \Cref{subsec:bound LCD}. Therefore, we can charge the total time needed to update the incident-edge data structures due to moves of vertices into $\Lambda_j^-$ to the set $\bigcup_{j=1}^{L_j}\hat \Pi_{j,\ell}$ of edges. From  \Cref{cor:bound increase}, the total number of edges in this set is $\inc_{j}\leq O(n_{\le j}h_{j}\Delta/\phi^{3})$. Therefore, the total time spent on
 updating the incident-edge data structures due to moves of vertices into
to $\Lambda_{j}^-$ is at most $O(n_{\le j}h_{j}\Delta / \phi^3)$.
It is also possible that a vertex $u \in \Lambda_{<j}$ is moved to $\Lambda_{j}^{-}$. However, this only happen once per vertex, so the total update time due to such moves is $O(m)$.

For the second event, fix some index $1\leq \ell<L_j$, and consider the time when a new $(j,\ell)$-phase starts. Recall that we set $\Lambda_{j,\ell}\gets\Lambda_{j,\ge\ell}$.
For every vertex $u$ that originally lied in $\Lambda_{j,\ge\ell}$,
we need to set $\edge_{j,\ell}(u)\gets\edge_{j,\ell}(u)\cup\edge_{j,>\ell}(u)$.
This can be done, for all such vertices $u$, in total time $O(|E_{G}(\Lambda_{j,\ge\ell})|)$.
Recall that, from Invariant \ref{inv: size of every sublayer},  $|\Lambda_{j,\ge\ell}|\leq O(n_{\le j}/2^{\ell})$, and that, from  \Cref{claim:bound edge}, edge set $E_G(\Lambda_j)$ has an $(h_j\Delta)$-orientation. It is then easy to see that edge set  $E_{G}(\Lambda_{j,\ge\ell})$ has an $(h_j\cdot \Delta)$-orientation as well, and so 
$|E_{G}(\Lambda_{j,\ge\ell})|\leq O(n_{\le j}h_{j}\Delta/2^{\ell})$.
By \Cref{cor:bound phase}, there are at most $\Ohat (2^{\ell}\Delta)$
$(j,\ell)$-phases over the course of the entire algorithm, and so the total time that we need to spend on updating the incident-edge data structure due to the initialization of the sublayer $\Lambda_{j,\ell}$ is bounded by $\Ohat (n_{\le j}h_{j}\Delta^{2})$.
Summing over all sublayers of layer $\Lambda_j$, the total time that the algorithm spends on maintaining the edge-incident data structures for vertices lying in layer $\Lambda_{j}$ is bounded by $O(n_{\le j}h_{j}\Delta^{2} \log (n) / \phi^3+m)$.

Lastly, observe that for all $1\leq j\leq r$, $n_{\le j}h_{j}\leq O(m)$. Therefore, the total
time that is needed to maintain the incident-edge data structure for all vertices og $G$ is at most:

 $$\sum_{j}\left(O(n_{\le j}h_{j}\Delta/\phi)+\Ohat (n_{\le j}h_{j}\Delta^{2}+m)\right)=\Ohat(m\Delta^{2}).$$

\subsection{Total Update Time, and Data Structures to Support $\protect\shortkpath$ and $\protect\tocore$
	Queries}

\label{subsec: lcd update time}

In this subsection, we provide some additional data structures that are needed to support $\protect\shortkpath$ and $\protect\tocore$
Queries, and analyze the total update time of the main algorithm
for \Cref{thm:LCD}. We start by analyzing the total update time required for maintaining all data structures that we have described so far.  Then, we describe additional data structures
that we maintain for supporting $\shortkpath$, $\tocore$, and $\shortpath$
queries, and analyze their update time. 

\paragraph*{Maintaining the Sublayers.}

Recall that, from  \Cref{thm: maintaining virtual degrees}, the total update time that is needed to maintain the partition of $V(G)$ into $\Lambda_{1},\dots,\Lambda_{r}$ is $\tilde{O}(m)$.
 As shown in \Cref{subsec: lcd incident ds},the edge-incident data structures for all vertices require total update time $\Ohat(m\Delta^{2})$.

Consider now some index $1\leq j\leq r$.
For the buffer sublayer $\Lambda_{j}^{-}$, we do not  need to maintain any
additional data structures. Consider now some non-buffer sublayer $\Lambda_{j,\ell}$, for
 $\ell<L_{j}$. At the beginning of a $(j,\ell)$-phase, we construct
the graph $H_{j,\ell}$ by setting $E(H_{j,\ell})\gets\bigcup_{u\in\Lambda_{j,\ell}}\edge_{j,\ell}(u)$, using the incident-edge data structure. 
This takes $O(|E(H_{j,\ell})|)$ time. Note that, without the incident-edge data structure, it is not immediately clear how to construct the graph $H_{j,\ell}$ in this time. The resulting graph $H_{j,\ell}$ is precisely $G[\Lambda_{j,\ell}]$, 
as desired. Next, we perform the core decomposition in graph $H_{j,\ell}$
using the algorithm from \Cref{thm: core decomposition}. The running time of the algorithm is $\Ohat(|E(H_{j,\ell})|)$. Recall that, from \Cref{claim:bound edge}, the edge set $E_{G}(\Lambda_j)$ has an $(h_j\Delta)$-orientation. Moreover, from Invariant \ref{inv: size of every sublayer}, $|\Lambda_{j,\ell}|\leq  \frac{n_{\le j}}{2^{\ell-1}}$. Therefore, 
$|E(H_{j,\ell})|\le|\Lambda_{j,\ell}|\cdot h_{j}\Delta\le\frac{n_{\le j}}{2^{\ell-1}}\cdot h_{j}\Delta$.
By \Cref{cor:bound phase}, the total number of 
$(j,\ell)$-phases over the course of the algorithm is bounded by  $\Ohat(2^{\ell}\Delta)$. Therefore,  the total time that is needed to construct the graphs $H_{j,\ell}$ and to compute core decompositions of such graphs over the course of the entire algorithm is bounded by $\Ohat(n_{\leq \ell}h_j\Delta/2^{\ell}) \cdot \Ohat(2^{\ell}\Delta)=\Ohat(n_{\le j}h_{j}\Delta^{2})\leq \Ohat(m\Delta^{2})$.

Note that it is straightforward to check that invariants \ref{inv: size of every sublayer} and \ref{inv: high degree of non-core vertices} hold over the course of the algorithm:  For each $1\leq j\leq r$ and $1\leq \ell\leq L_j$ we need to ensure that $|\Lambda_{j,\ell}|\le n_{\le j}/2^{\ell-1}$ always holds.
This can be checked in constant time by keeping track of $|\Lambda_{j,\ell}|$.
For each vertex $u\in U_{j,\ell}$, we need to ensure the invariant that
$\deg_{\le(j,\ell)}(u)\ge\deg_{\le(j,\ell)}^{(0)}(u)/\gapU$ always holds. This
can be checked in $O(\log n)$ by maintaining prefix sums of $|\edge_{j'}(u)|$
and $|\edge_{j,\ell'}(u)|$ for all $j'<j$ and $\ell'\le\ell$. 
As there are $\Ohat(\log^{2}n)$ sublayers, the total cost for maintaining
all sublayers $\Lambda_{j,\ell}$ and their corresponding graphs $H_{j,\ell}=G[\Lambda_{j,\ell}]$, together with computing the initial core decompositions $\fset_{j,\ell}$
is at most $\Ohat(m\Delta^{2})$.

\paragraph*{Oracles for $\protect\shortkpath$ queries.}

Whenever a core $K$ is created, we maintain the data structure
from \Cref{thm:shortkquery} that allows us to maintain the core $K$ under the deletion of edges from $G$, and to support $\shortkpath(K,u,v)$
queries within the core $K$. Consider some index $1\leq j\leq r$ and a non-buffer sublayer $\Lambda_{j,\ell}$ of $\Lambda_j$.
At the beginning of each $(j,\ell)$-phase, let $\fset_{j,\ell}$
denote the collection of cores constructed by the algorithm from \Cref{thm: core decomposition}. 
The total update time needed to maintain the data structure for all cores $K\in\fset_{j,\ell}$ throughout a single $(j,\ell)$-phase
is $\sum_{K\in\fset_{j,\ell}}O(|E(K)|^{1+1/q}(\gamma(n))^{O(q)})\leq O(|E(H_{j,\ell})|^{1+1/q}(\gamma(n))^{O(q)})$.
As observed already, $|E(H_{j,\ell})|\le|\Lambda_{j,\ell}|\cdot h_{j}\Delta\le n_{\le j} h_{j}\Delta/2^{\ell-1}\leq O(m\Delta/2^{\ell})$
Since, from  \Cref{cor:bound phase}, the total number of 
$(j,\ell)$-phases over the course of the algorithm is bounded by  $\Ohat(2^{\ell}\Delta)$,
the total time needed to maintain all cores within the layer $(j,\ell)$ over the course of the algorithm is bounded by $O(m^{1+1/q}\Delta^{2+1/q}(\gamma(n))^{O(q)})$. Summing this up over all $O(\log n)$ non-buffer sublayers $\Lambda_{j,\ell}$, we get that the total time that is needed to maintain all cores that are ever present in $\fset_j$ is bounded by $O(m^{1+1/q}\Delta^{2+1/q}(\gamma(n))^{O(q)})$.

Note that the algorithm from \Cref{thm:shortkquery} directly supports
the $\shortkpath(K,u,v)$ queries: given any pair of vertices $u$
and $v$ that lie in the same core $K$, it return a simple $u$-$v$ path $P$
in $K$ connecting $u$ to $v$, of length at most $(\gamma(n))^{O(q)}$, in time $(\gamma(n))^{O(q)}$.

\paragraph*{ES-trees for $\protect\tocore$ queries.}

At the beginning of a $(j,\ell)$-phase of a non-buffer sublayer $\Lambda_{j,\ell}$,
we construct an auxiliary graph $C_{j,\ell}$ for maintaining short paths
from vertices of $U_{j,\ell}$ to vertices of $\Lambda_{<j}\cup\Lambda_{j,<\ell}\cup\hat{K}_{j,\ell}$, whose existence is guaranteed by \Cref{lem:up path non buffer}. The vertex set of the
graph $C_{j,\ell}$ is $V(C_{j,\ell})$ $=\Lambda_{j,\ell}\cup\{s_{j,\ell}\}$.
The edge set of $C_{j,\ell}$ contains the edges of $E(H_{j,\ell})$. Additionally,
for every vertex $u\in\hat{K}_{j,\ell}$, we add the edge $(s_{j,\ell},u)$ into $C_{j,\ell}$.
For every vertex $u\in U_{j,\ell}$, if $E_{G}(u,\Lambda_{<j}\cup\Lambda_{j,<\ell})\neq\emptyset$,
then we add the edge $(s_{j,\ell},u)$ to $C_{j,\ell}$ and we associate the
edge $(s_{j,\ell},u)$ with an edge $(w,u)$ where $w$ is some neighbor
of $u$ in $\Lambda_{<j}\cup\Lambda_{j,<\ell}$. Note that we can
maintain this association by using the incident-edge data structure.
It is easy to see that the time needed to construct the graph $C_{j,\ell}$ is
 subsumed by the time needed to construct the graph $H_{j,\ell}$. 

Observe that, throughout each $(j,\ell)$-phase, graph $C_{j,\ell}$ only undergoes
 edge- and vertex-deletions, just like graph $H_{j,\ell}$. Therefore, we can
maintain an ES-tree $T_{j,\ell}$ rooted at $s_{j,\ell}$, in the graph
$C_{j,\ell}$, up to distance $O(\log n)$. The total update time for
$T_{j,\ell}$ is $O(|E(C_{j,\ell})|\log n)$ for each $(j,\ell)$-phase. As $|E(C_{j,\ell})|\le|E(H_{j,\ell})|+|U_{j,\ell}|=O(n_{\le j}h_{j}\Delta/2^{\ell})$
and there are at most $\Ohat(2^{\ell}\Delta)$ $(j,\ell)$-phases, the
total time needed to maintain the ES-trees $T_{j,\ell}$ in graphs $C_{j,\ell}$
over the course of the algorithm is at most $\Ohat(n_{\le j}h_{j}\Delta^{2})=\Ohat(m\Delta^{2})$.
The total cost of maintaining such trees for all non-buffer sublayers of all layers $\Lambda_j$ is at most $\Ohat(m\Delta^{2})$.

\paragraph*{Supporting $\protect\tocore$ queries.}

For convenience, we will slightly abuse notation. For each non-buffer
sublayer $\Lambda_{j,\ell}$, the ES-tree $T_{j,\ell}$ is formally
a subgraph of $C_{j,\ell}$, but we will treat $T_{j,\ell}$ as a subgraph
of $G[\Lambda_{\le j}]$ as follows. Each edge $(s_{j,\ell},u)\in$$T_{j,\ell}$
where $u\in U_{j,\ell}$ corresponds to some edge $(u,w)\in E(U_{j,\ell},\Lambda_{<j}\cup\Lambda_{j,<\ell})\subseteq E(G[\Lambda_{\le j}])$.
Edges of the form $(s_{j,\ell},u)\in T_{j,\ell}$, where $u\in\hat{K}_{j,\ell}$
do not correspond to edges of $G[\Lambda_{\le j}]$. The remaining
edges of $T_{j,\ell}$ that are not incident to $s_{j,\ell}$ are edges of $G[\Lambda_{\le j}]$ by definition.
For each buffer sublayer $\Lambda_{j}^{-}=\Lambda_{j,L_{j}}$, we
will also define a subgraph $T_{j,\ell}$ as follows. For each $u\in\Lambda_{j,L_{j'}}$,
$T_{j,\ell}$ contains the edge $(u,v)$ where $v$ is the (lexicographically)
first neighbor of $u$ in $\Lambda_{<j}\cup\Lambda_{j,<L_{j}}$ ($v$
must exist by \Cref{lem:up path buffer}). Note that $T_{j,\ell}$
is a subgraph of $G[\Lambda_{\le j}]$. 

Now, in order to respond to $\tocore(u)$ query, where $u\in\Lambda_{j,\ell}$, we follow the
simple path of length $O(\log n)$ in $T_{j,\ell}$ starting from $u$ to a vertex of
$\hat{K}_{j,\ell}$, or a vertex of $\Lambda_{<j}\cup\Lambda_{j,<\ell}$. If we
reach a vertex of $\hat{K}_{j,\ell}$, then we are done. Otherwise, we reach a
vertex in some sublayer that lies above $\Lambda_{j,\ell}$, and we continue the same process in that sublayer. The total number of such iterations is then bounded by the total number of sublayers in the entire graph -- at most $O(\log^{2}n)$. Observe that the paths computed at different iterations may only intersect at their endpoints, because every vertex on such a path, except for possibly the last vertex, lies in a single sublayer. Therefore, the concatenation of these paths is a simple path.
To conclude, given a query $\tocore(u)$, we can return a simple path $P$ of
length $O(\log^{3}n)$ connecting $u$ to a vertex in some core, in time $O(|P|)$, such that $P$ is non-increasing with respect to the sublayers. In other words, if $P_{u}=\{u=u_{1},u_{2},\dots,u_{k}\}$, then, for each $i$,
if $u_{i}\in\Lambda_{j,\ell}$, then $u_{i+1}\in\Lambda_{<j}\cup\Lambda_{j,\le\ell}$. 

\paragraph*{Minimum Spanning Forest for $\protect\shortpath$ queries.}
For an index $1\leq j\leq r$, we denote  $T_{\le j}=\bigcup_{j'\le j,\ell\ge1}T_{j',\ell}$. Recall that
$\fset_{\le j}$ is the collection of all cores for layers $\Lambda_1,\ldots,\Lamda_j$.
Let $\hat{K}_{\le j}=\bigcup_{K\in\fset_{\le j}}K$ denote the union
of all the cores in $\fset_{\le j}$. Note that $\hat{K}_{\le j}$
and $T_{\le j}$ subgraphs of $G[\Lambda_{\le j}]$, and they do not share any edges.

We maintain a fully dynamic minimum spanning forest
$M_{j}$ for the graph $G[\Lambda_{\le j}]$, with the following edge lengths. We assign weight $0$
to all edges in $\hat{K}_{\le j}$, weight $1$ to all edges of $T_{\le j}$,
and weight $2$ to all remaining edges of $G[\Lambda_{\le j}]$. 
The spanning forest $M_j$ can be maintained using the algorithm of \cite{dynamic-connectivity}, that has
$O(\log^{4}n)$ amortized update time.

Additionally, we use the \emph{top tree} data structure due to \cite{AlstrupHLT05}, whose properties are summarized in the following lemma.

\begin{lem}
	[Top Tree Data Structure from \cite{AlstrupHLT05}]The top tree data structure $\T$ is given as input a
	forest $F$ with weights on edges, that undergoes edge insertions and edge deletions
	(but we are guaranteed that $F$ remains a forest throughout the algorithm),
	 and supports the following queries, in $O(\log n)$ time per query:
	\begin{itemize}
		\item $\mathtt{connect}(x,y)$: given two vertices $x$ and $y$, determine whether $x$ and $y$ are in the same
		connected component of $F$ (see Section 2.4 of \cite{AlstrupHLT05});
		\item $\mathtt{minedge}(x,y)$: given two vertices $x$ and $y$ lying in the same connected component of $F$, return a minimum-weight edge on the unique path connecting $x$ to $y$ in $F$
		(a variation of Theorem 4 of \cite{AlstrupHLT05});
		
		\item $\mathtt{weight}(x,y)$:  given two vertices $x$ and $y$ lying in the same connected component of $F$, return the total weight of all edges lying on the unique path connecting $x$ to $y$ in $F$ (Lemma
		5 of \cite{AlstrupHLT05}).
		\item $\mathtt{jump}(x,y,d)$: given two vertices $x$ and $y$ lying in the same connected component of $F$, return the $d$th vertex on the unique path connecting $x$ to $y$ in $F$; if the path connecting $x$ to $y$ contains fewr than $d$ vertices, return $\emptyset$
		(Theorem 15 of \cite{AlstrupHLT05}).
	\end{itemize}
The data structure has $O(\log n)$ worst-case update time.
\end{lem}

For all $1\leq j\leq r$, we maintain the top tree data structure
  $\tset_{M_{j}}$ for the forest $M_j$.

It is easy to see that the total time that is required for maintaining the minimum spanning forests $\set{M_{j}}_{j=1}^r$ and their corresponding top tree data structures $\tset_{M_{j}}$
is subsumed by other components of the algorithm.

To conclude, the total update time of the $\lcd$ data structure for
\Cref{thm:LCD} is $\Ohat(m^{1+1/q}\Delta^{2+1/q}(\gamma(n))^{O(q)})$.%

\subsection{Supporting $\protect\shortpath$ Queries}

\label{subsec: shortpath}

In this section, we fix an index $1\leq j\leq r$, and describe an algorithm for processing $\shortpath(j,\cdot,\cdot)$
queries. Recall that we have denoted $\hat{K}_{\le j}=\bigcup_{K\in\fset_{\le j}}K$
and $T_{\le j}=\bigcup_{j'\le j,\ell\ge1}T_{j',\ell}$. We start by analyzing the structure of the spanning forest $M_j$.

Recall that $T_{\leq j}$ is a forest, and every tree in this forest is rooted in a vertex of $\hat K_{\leq j}$. Moreover, if a vertex of $T_{\leq j}$ does not serve as a tree root, then it does not lie in $\hat K_{\leq j}$, and every vertex in $\Lambda_{\leq j}\setminus \hat K_{\leq j}$ must lie in some tree in $T_{\leq j}$. Recall also that, from \Cref{cor:up path}, the depth of every tree in $T_{\leq j}$ is bounded by $O(\log^{3}n)$.

Consider now some connected component $C$ of graph $G[\Lambda_{\le j}]$. Let
$\fset^C$  denote the collection of all cores $K\in \fset_{\leq j}$ with $K\subseteq C$, and let $k_C=|\fset^C|$. Recall that, from  \Cref{obs:bound core moment}, \thatchapholnew{$k_{C}\le O(|V(C)|/(\phi^2 h_{j}))=O(|V(C)|(\gamma(n))^2/h_{j})$}. 

The following two observations easily follow from the properties of a minimum spanning tree.

\begin{observation}\label{obs: spanning tree spans cores}
	Let $C$ be a connected component of  $G[\Lambda_{\le j}]$, and let $K\in \fset^C$ be a core that currently lies in $\fset_{\leq j}$ and is contained in $C$. Then there is some connected sub-tree $T^*$ of the forest $M_j$ that contains every vertex of $K$.
	\end{observation}

\begin{proof}
	Assume otherwise. Consider the sub-graph of $M_j$ induced by the edges of $E(K)$. Then this graph is not connected, and moreover, there must be two connected components $X$ and $Y$ of this graph, such that core $K$ contains some edge $e$ connecting a vertex of $X$ to a vertex of $Y$. Adding the edge $e$ to $M_j$ must create some cycle $R$. We claim that at least one edge on this cycle must have weight greater than $0$. Indeed, otherwise, every edge on cycle $R$ lies in the core $K$, and so $X$ and $Y$ cannot be two connected components of the subgraph of $M_j$ induced by $E(K)$. Since edge $e$ has weight $0$, we have reached a contradiction to the minimality of the forest $M_j$.
	\end{proof}

\begin{observation}\label{obs: spanning tree contains Tj}
	Every edge of the forest $T_{\leq j}$ belongs to $M_j$.
	\end{observation}

\begin{proof}
	Assume for contradiction that this is not the case, and let $T'$ be a tree of $T_{\leq j}$ with $E(T')\not\subseteq E(M_j)$ As before, consider the sub-graph of $M_j$ induced by the edges of $E(T')$. This graph is not connected, and, so there must be two connected components $X$ and $Y$ of this graph, such that tree $T'$ contains some edge $e$ connecting a vertex of $X$ to a vertex of $Y$. Adding the edge $e$ to $M_j$ must create some cycle $R$. We claim that at least one edge on this cycle must have weight $2$. Indeed, otherwise, every edge on cycle $R$ has weight $0$ or $1$. This is impossible because tree $T'$ contains exactly one vertex that lies in a core of $\fset_{\leq j}$, and the only edges whose weight is $0$ are edges that are contained in the cores. Therefore, there must be some edge $e'$ on the cycle $R$ that is incident to some vertex of $T'$, is not contained in $T'$, and is not contained in any core of $\fset_{\leq j}$. The weight of such an edge then must be $2$. But, since the weight of the edge $e$ is $1$, this contradicts the minimality of $M_j$.
\end{proof}

Consider again some connected component $C$ of the graph $G[\Lambda_{\leq j}]$, and recall that $M_j$ is a minimum spanning forest of  $G[\Lambda_{\leq j}]$. Let $M_j^C\subseteq M_j$ be the unique tree in the forest $M_j$ that is spanning $C$. From the above two observations it is easy to see that, if we delete all weight-2 edges from $M_j^C$, then we will obtain $k_C$ connected components. Therefore, we obtain the following immediate corollary.

\begin{corollary}\label{cor: few weight-2 edges}
	For every connected  component $C$ of $G[\Lambda_{\leq j}]$, tree  $M_j^C$ contains at most $k_C-1$ edges of weight $2$.
	\end{corollary}

Consider now any path $P$ in the forest $M_{j}$. Recall that all edges of $P$ have weights in $\set{0,1,2}$. For $x\in\{0,1,2\}$, an \emph{$x$-block} of the path $P$ is a maximal subpath of $P$ such that every edge on the subpath has weight
exactly $x$. We need the following observation on the structure of such paths.

\begin{observation}\label{obs: structure of paths}
	Let $P$ be any path in the spanning forest $M_{j}$, and let $C$ be the connected component of
	$G[\Lambda_{\le j}]$ containing $P$. Then: 
	\begin{enumerate}
		\item there are at most $k_{C}-1$ edges of weight $2$ in $P$;
		\item the number of $0$-blocks in $P$ is at most $k_C$; and
		\item the number of $1$-blocks in $P$ is at most  $2k_{C}$, with each $1$-block
		having length at most $O(\log^{3}n)$. 
	\end{enumerate}
\end{observation}

\begin{proof}
	The first assertion follows immediately from \Cref{cor: few weight-2 edges}, and the second assertion follows immediately from \Cref{obs: spanning tree spans cores} and the fact that at most $k_C$ cores of $\fset_{\leq j}$ are contained in $C$.
	
	In order to prove the third assertion, let $\Sigma$ be a collection of paths that is obtained by removing all weight-$2$ edges from $P$. Then, from \Cref{cor: few weight-2 edges}, $|\Sigma|\leq k_C-1$. Moreover, since every tree in $T_{\leq j}$ contains exactly one vertex of $\hat K_{\leq j}$, for each such path $\sigma\in \Sigma$, there is at most one core $K\in \fset_{\leq j}$, such that the edges of $K$ lie on $\sigma$, and if such a core exists, then, from \Cref{obs: spanning tree spans cores}, the edges of $K$ appear contiguously on $\sigma$. Therefore, every path $\sigma$ contains at most two $1$-blocks, and the total number of $1$-blocks on $P$ is at most $2k_C$. Since every tree in $T_{\leq j}$ has depth at most $O(\log^{3}n)$, the length of each such $1$-block is at most $O(\log^3n)$.
\end{proof}

The following corollary follows immediately from \Cref{obs: structure of paths} and the fact that for every connected component $C$ of $\Lambda_{\leq j}$, $k_{C}= O(|V(C)|(\gamma(n))^2/h_{j})$.

\begin{corollary}
	\label{lem:path length}Let $P$ be any path contained in the graph $M_{j}$, and let $C$ be the connected component of $\Lambda_{\leq j}$ containing $P$. Then the total number of edges of $P$ that have non-zero weight is at most \thatchapholnew{ $\tilde O(k_C)\leq \Otil(|V(C)|(\gamma(n))^2/h_{j})$. }
\end{corollary}

We are now ready to describe an algorithm for processing a query $\shortpath(j,u,v)$. 
Our first step is to check whether $u$ and $v$ are connected in
	$M_{j}$. This can be done in time $O(\log n)$, using the $\mathtt{connect}(u,v)$ query in the top tree $\tset_{M_{j}}$ data structure. If $u$ and $v$ are not connected in $M_{j}$, then we terminate the algorithm and report that $u$
	and $v$ are not connected in $G[\Lambda_{\le j}]$. We assume from now on that $u$ and $v$ are connected in $M_j$.

	We denote by $P$ the unique path connecting $u$ and $v$ in $M_j$. Note that our algorithm does not compute the path $P$ explicitly, as it may be too long. We think of the path $P$ as being oriented from $u$ to $v$. Let $B_1,\ldots,B_z$ be the sequence of all maximal $0$-blocks on path $P$; we assume that the blocks are indexed in the order of their appearance on $P$.
	For $1\leq i\leq z$, we denote by $b_i$ and by $b'_i$ the first and the last endpoint of $B_i$, respectively. For $1\leq i<z$, let $A_i$ be the sub-path of $P$ connecting $b'_i$ to $b_{i+1}$; let $A_0$ be the sub-path of $P$ connecting $u$ to $b_1$, and let $A_z$ be the sub-path of $P$ connecting $b'_z$ to $v$. The next step in our algorithm is to identify all endpoints of the $0$-blocks on path $P$, that is, the algorithm will find the parameter $z$ (the number of the maximal $0$-blocks on $P$), and, for all $1\leq i\leq z$, it will compute the endpoints $b_i,b'_i$ of block $B$. We do so using queries $\mathtt{minegdge}$,  $\mathtt{weight}$, and $\mathtt{jump}$ to the top tree $\tset_{M_j}$ data structure.
	
	Specifically, we start by running query  $\mathtt{minedge}(u,v)$ on the top tree $\tset_{M_j}$. Let $e=(x,y)$ be the returned edge. If the weight of the edge is greater than $0$, then there are no $0$-weight edges on path $P$, and so we can skip the current step. Assume therefore that the weight of the edge $e$ is $0$. Let $B_i$ be the $0$-block containing $e$ (note that we do not know the index $i$). In order to find the first endpoint $b_i$ of the $0$-block $B_i$, we perform a binary search using queries $\mathtt{jump}(x,u,d)$ for various values of $d$. If $a_d$ is the vertex returned in response to query $\mathtt{jump}(x,u,d)$, then we use query $\mathtt{weight}(x,a_d)$ to find the total weight of all edges on the sub-path of $P$ connecting $x$ to $a_d$. If the returned weight is $0$, then we know that $a_d\in B_i$, and we increase the value of $d$; otherwise, we know that the sub-path of $P$ between $b_i$ and $x$ contains fewer than $d$ edges, and we decrease $d$ accordingly. Therefore, using binary search, in $O(\log n)$ iterations, we can compute the endpoint $b_i$ of the block $B_i$, and we can compute the other endpoint $b'_i$ of the block similarly. The total time needed to compute both endpoints is therefore $O(\log^2n)$. Once we compute the endpoints $b_i,b'_i$, we recursively apply the same algorithm to the sub-path of $P$ connecting $u$ to $b_i$, and the sub-path of $P$ connecting $b'_i$ to $v$. We conclude that we can compute the number $z$ of the maximal $0$-blocks on path $P$, and the endpoints of these blocks, in time $O(z\log^2n)$.
	
	Once the endpoints of all $0$-blocks are computed, we compute the paths $A_1,\ldots,A_z$, using queries $\mathtt{jump}(a,a',1)$. 
Lastly, for all $1\leq i\leq z$, we run query 	$\shortkpath(K_{i},b_{i},b'_{i})$, where $K_i$ is the core of $\fset_{\leq j}$ containing $b_i$ and $b'_i$, to compute a path $B'_i$ in core $K_i$ that connects $b_i$ to $b'_i$ and has length at most $(\gamma(n))^{O(q)}$. We then return a $u$-$v$ path that is obtained by concatenating the paths $A_1,B'_1,A_2,\ldots,B'_z,A_z$.

We use the following lemma to bound the length of the resulting path.

\begin{lemma}
	Given query $\shortpath(j,u,v)$, the above algorithm either correctly reports that
	$u$ and $v$ are not connected in $G[\Lambda_{\le j}]$ in time $O(\log n)$, or it returns
	a simple $u$-$v$ path $P'$ of length at most $O(|V(C)|(\gamma(n))^{O(q)}/h_{j})$, 
	in time $O(|P'|\cdot \left (\gamma(n))^{O(q)}\right )$, where $C$ is the connected component of
	$G[\Lambda_{\le j}]$ containing $u$ and $v$.
\end{lemma}

\begin{proof}
	It is immediate to see that, if $u$ and $v$ are not connected in $G[\Lambda_{\le j}]$, then the algorithm reports this correctly in time $O(\log n)$. Therefore, we assume from now on that some connected component $C$ of $G[\Lambda_{\le j}]$ contains $u$ and $v$. As before, we denote by $P$ the unique $u$-$v$ path in the graph $M_j$. Note that, from \Cref{lem:path length}, the total number of edges	on $P$ with non-zero weight is at most \thatchapholnew{$\Otil(|V(C)|(\gamma(n))^2/h_{j})$. In particular, the number of maximal $0$-blocks on $P$ must be bounded by $z\leq \Otil(|V(C)|(\gamma(n))^2/h_{j})$. Since we are guaranteed that, for all $1\leq i\leq z$, the length of the path $B'_i$ is bounded by $(\gamma(n))^{O(q)}$, we conclude that the length of the returned path $P'$ is at most $O(|V(C)|(\gamma(n))^{O(q)}/h_{j})$.}
	
	In order to bound the running time, recall that detecting the endpoints of the $0$-blocks takes time $O(z\cdot \log^2n)$. Computing all vertices on paths $A_1,\ldots,A_z$ takes time $O(\log n)$ per vertex. Lastly, computing the paths $B'_1,\ldots,B'_z$ takes total time at most $z\cdot (\gamma(n))^{O(q)}$. Altogether, the running time is bounded by $O(|P'|\cdot \left (\gamma(n))^{O(q)}\right )$.
	\end{proof}

%% file: SSSP.tex
\section{SSSP}
\label{sec: SSSP}
This section is dedicated to the proof of \Cref{thm: main for SSSP}. The main idea is identical to that of \cite{fast-vertex-sparsest}, who use the framework
of \cite{Bernstein}, combined with a weaker version of the \lcd data structure. The improvements in the guarantees that we obtain follow immediately by plugging the new $\lcd$ data structure from \Cref{sec: LCD} into their algorithm. We still include a proof for completeness, since our \lcd data structure is defined somewhat differently.
As is the standard practice in such algorithms, we treat each distance scale separately. We prove the following theorem that allows us to handle a single distance scale.

\begin{theorem}
	\label{thm: main for SSSP w distance bound} There is a deterministic
	algorithm, that, given a simple undirected $n$-vertex graph $G$ with weights on edges that
	undergoes edge deletions, together with a source vertex $s\in V(G)$ and parameters $\eps\in(1/n,1)$
	and $D>0$,   supports the following queries: 
	\begin{itemize}
		\item $\distquery_{D}(s,v)$: in time $O(1)$, either correctly report that $\ensuremath{\dist_{G}(s,v)>2D}$,
		or return an estimate $\widetilde{\dist}(s,v)$. Moreover, if $D\leq\dist_{G}(s,v)\leq2D$,
		then $\dist_{G}(s,v)\le\widetilde{\dist}(s,v)\le(1+\eps)\dist_{G}(s,v)$ must hold.
		\item $\pquery_{D}(s,v)$: either correctly report that $\ensuremath{\dist_{G}(s,v)>2D}$
		in time $O(1)$, or return a $s$-$v$ path $P$ 
		in time $\Ohat(|P|)$. Moreover, if $D\leq\dist_{G}(s,v)\leq2D$, then the length of $P$ must be bounded by $(1+\eps)\dist_{G}(s,v)$. Path $P$ may not be simple, but an edge may appear at most once on $P$.
	\end{itemize}

The total update time of the algorithm is $\Ohat(n^{2}/\eps^{2})$.
\end{theorem}

We provide a proof of \Cref{thm: main for SSSP w distance bound} below, after we complete the proof of 
\Cref{thm: main for SSSP} using it, via standard arguments.

We will sometimes refer to edge weights as edge lengths, and we denote the length of an edge $e\in E(G)$ by $\ell(e)$.
We assume that the minimum edge weight is $1$ by scaling, so the
maximum edge weight is $L$. For all $0\le i\le\ceil{\log(Ln)}$, we maintain
a data structure from \Cref{thm: main for SSSP w distance bound}
with the distance parameter $D_{i}=2^{i}$. Therefore, the total update time of our algorithm is bounded by $\Ohat(n^{2}(\frac{\log L}{\eps^{2}}))$, as required.

In order to respond to a query $\distquery(s,v)$, we perform a binary search on the values $D_i$, and run queries $\distquery_{D_{i}}(s,v)$ in the corresponding data structure. Clearly, we only need to perform at most $O(\log\log (Ln))$ such queries, in order to respond to query $\distquery(s,v)$.

In order to respond to  $\pquery(s,v)$, we first run the algorithm for $\distquery(s,v)$ in order to identify a distance scale $D_i$, for which  $D_i\leq\dist_{G}(s,v)\leq2D_i$ holds. We then run query $\pquery_{D_i}(s,v)$ in the corresponding data structure. 

In order to complete the proof of \Cref{thm: main for SSSP}, it now remains to prove \Cref{thm: main for SSSP w distance bound}, which we do in the remainder of this section.

Recall that we have denoted by $\ell(e)$ the length/weight of the edge $e$ of $G$. We use standard edge-weight rounding to show that
we can assume that $D=\ceil{4n/\eps}$ and that all edge lengths are integers between $1$ and $4D$.
In order to
achieve this, we discard all edges whose length is greater than $2D$,
and we set the length of each remaining edge $e$ to be $\ell'(e)=\ceil{4n\ell(e)/(\eps D)}$.
For every pair $u,v$ of vertices, let $\dist'(u,v)$ denote the distance
between $u$ and $v$ with respect to the new edge length values.
Notice that for all $u,v$, $\frac{4n}{\eps D}\dist(u,v)\leq\dist'(u,v)\leq\frac{4n}{\eps D}\dist(u,v)+n$,
since the shortest $s$-$v$ path contains at most $n$ edges. Moreover,
if $\dist(u,v)\geq D$, then $n\leq\dist(u,v)\cdot\frac{n}{D}$, so
$\dist'(u,v)\leq\frac{4n}{\eps D}\dist(u,v)+\frac{n}{D}\dist(u,v)\leq\frac{4n}{\eps D}\dist(u,v)(1+\eps/4)$.
Notice also that, if $D\leq\dist(u,v)\leq2D$, then $\ceil{\frac{4n}{\eps}}\leq\dist'(u,v)\leq4\ceil{\frac{4n}{\eps}}$.
Therefore, from now on we can assume that $D=\ceil{4n/\eps}$, and
for simplicity, we will denote the new edge lengths by $\ell(e)$
and the corresponding distances between vertices by $\dist(u,v)$.
From the above discussion, all edge lengths are integers between $1$
and $4D$. It is now enough to prove \Cref{thm: main for SSSP w distance bound}
for this setting, provided that we ensure that, whenever $D\leq\dist(s,v)<4D$
holds, we return a path of length at most $(1+\eps/2)\dist(s,v)$
in response to query $\pathquery(v)$.

\paragraph*{The Algorithm.}

Let $m$ denote the initial number of edges in the input graph $G$. We partition
all edges of $G$ into $\lambda=\floor{\log(4D)}$ classes, where
for $0\leq i\leq \lambda$, edge $e$ belongs to \emph{class $i$} iff
$2^{i}\leq\ell(e)<2^{i+1}$. We denote the set of all edges of $G$
that belong to class $i$ by $E^{i}$. Fix an index $1\leq i\leq\lambda$,
and let $G_{i}$ be the sub-graph of $G$ induced by the edges in
$E^{i}$. We view $G_{i}$ as an unweighted graph and maintain the
$\lcd$ data structure from \Cref{thm:LCD} on $G_{i}$ with parameter
$\Delta=2$ and $q=\log^{1/8}n$ using total update time $\Ohat(m^{1+1/q}\Delta^{2+1/q})=\Ohat(m)$.
Recall that $\gamma(n)=\exp(O(\log^{3/4}n))$.

We let  $\alpha=(\gamma(n))^{O(q)}=\Ohat(1)$ be chosen such that, in response to query $\shortpath(j,u,v)$, the \lcd data structure must return a path of length at most $|V(C)|\cdot \alpha /h_j$, where $C$ denotes the connected component of graph $G[\Lambda_{\leq j}]$ containing $u$ and $v$. We use the parameter $\tau_{i}=\frac{8n\lambda\alpha}{\eps D}\cdot2^{i}$ that is associated with graph $G_i$. This parameter is used to partition the vertices of $G$ into a set of vertices that are \emph{heavy} with respect to class $i$, and vertices that are \emph{light} with respect of class $i$. Specifically, we let  $U_{i}=\set{v\in V(G_{i})\mid\td_{G_{i}}(v)\geq\tau_{i}}$
be the set of vertices that are heavy for class $i$, and we let $\overline{U}_{i}=V(G_{i})\setminus U_{i}$ be the set of vertices that are light for class $i$.

Next, we define the heavy and the light graph for class $i$. The \emph{heavy graph for class $i$}, that is denoted by $G_i^H$, is defined as  $G_{i}[U_{i}]$. In other words, its vertex set is the set of all vertices that are heavy for class $i$, and its edge set is the set of all class-$i$ vertices whose both endpoints are heavy for class $i$. The \emph{light graph for class $i$}, denoted by $G_{i}^{L}$, is defined as follows. Its vertex set is $V(G_i)$, and its edge set contains all edges $e\in E_i$, such that at least one endpoint of $e$ lies in $\overline U_i$. 
Notice that we can exploit the \lcd data structure to compute the initial graphs $G_i^H$ and $G_i^L$, and to maintain them, as edges are deleted from $G$.

Our algorithm exploits the $\lcd$ data structure in two ways. First, observe that, from 
\Cref{claim:bound edge}, for all $1\leq i\leq \lambda$, the total number of edges that ever belong to the light graph $G_{i}^{L}$ over
the course of the algorithm is bounded by $O(n\tau_{i})$. Additionally, we will exploit the $\shortpath$ queries that the \lcd data structure supports.

 Let $j_{i}$ be the largest integer, such that $h_{j_{i}}\ge\tau_{i}$ (recall that $h_j$ is the virtual degree of vertices in layer $\Lambda_j$).
Given a query $\shortpath(j_{i},u,v)$ to the $\lcd$
data structure on $G_{i}$, where $u$ and $v$ lie in the same connected component $C$ of $G_{i}^{H}$, the data structure must return a simple $u$-$v$ path in $C$, containing at most  $\frac{|V(C)|\alpha}{\tau_{i}}$ edges. Abusing the notation, we denote this query by $\shortpath(C,u,v)$ instead.

Let $G^{L}=\bigcup_{i=1}^{\lambda}G_{i}^{L}$ be the \emph{light graph} for the graph $G$.
Next, we define an \emph{extended light graph} $\hat{G}^{L}$, as follows.
We start with $\hat{G}^{L}=G^{L}$; the vertices of $G^{L}$ are called
\emph{regular vertices}. Next, for every $1\leq i\leq\lambda$, for
every connected component $C$ of $G_{i}^{H}$, we add a vertex $v_{C}$
to $\hat{G}^{L}$, that we call a \emph{special vertex}, or a \emph{supernode},
and connect it to every regular vertex $u\in V(C)$ with an edge of
length $1/4$. 

 For all $1\leq i\leq \lambda$, we use the $\CONNSF$ data structure on graph $G_i^H$, in order to maintain its connected components. The total update time of these connectivity data structures is bounded by $O(m\lambda)\leq O(m\log D)$.\footnote{We note that our setting is slightly different from that of~\cite{Bernstein},
	who used actual vertex degrees and not their virtual degrees in the definitions of the light and the heavy graphs. Our definition  is identical to
	that of~\cite{fast-vertex-sparsest}, though they did not define the
	virtual degrees explicitly. However, they used Procedure \DP in order to define  the heavy and the light graphs, and so their definition of both graphs
	is identical to ours, except for the specific choice of the thresholds
	$\tau_{i}$).} %
The following observation follows immediately from the assumption
that all edge lengths in $G$ are at least $1$.
\begin{observation}
	\label{obs: distances in G and light graph} Throughout the algorithm,
	for every vertex $v\in V(G)$, $\dist_{\hat{G}^{L}}(s,v)\leq\dist_{G}(s,v)$. 
\end{observation}

The following theorem was proved in~\cite{fast-vertex-sparsest};
the proof follows the arguments from~\cite{Bernstein} almost exactly.

\begin{theorem}[Theorem 4.4 in \cite{fast-vertex-sparsest}]
	\label{thm: main for maintaining a light graph} There is a deterministic
	algorithm, that maintains an approximate single-source shortest-path
	tree $T$ of graph $\hat{G}^{L}$ from the source $s$, up to distance $8D$.
	Tree $T$ is a sub-graph of $\hat{G}^{L}$, and for every vertex $v\in V(\hat{G}^{L})$,
	with $\dist_{\hat{G}^{L}}(s,v)\leq8D$, the distance from $s$ to
	$v$ in $T$ is at most $(1+\eps/4)\dist_{\hat{G}^{L}}(s,v)$. The
	total update time of the algorithm is $\tilde{O}\left(\frac{nD}{\eps}+|E(G)|+\sum_{e\in E}\frac{D}{\eps\ell(e)}\right)$,
	where $E(G)$ is the set of edges that belong to $G$ at the beginning
	of the algorithm, and $E$ is the set of all edges that are ever present
	in the graph $\hgl$. 
\end{theorem}

Recall that $D=\Theta(n/\eps)$. Since, for all $1\leq i\leq\lambda$,
the total number of edges of $E^{i}$ ever present in $\hgl$ is bounded
by $O(n\tau_{i})=O\left(n\cdot\frac{8n\lambda\alpha}{\eps D}\cdot2^{i}\right)=\Ohat(n\cdot2^{i})$ from 
\Cref{claim:bound edge},
and since the total number of edges incident to the special vertices
that are ever present in $\hgl$ is bounded by $O(n\lambda\log n)=\tilde{O}(n)$,
we get that the running time of the algorithm from \Cref{thm: main for maintaining a light graph}
is bounded by:

\[
\tilde{O}\left(\frac{n^{2}}{\eps^{2}}+\sum_{i=1}^{\lambda}\frac{|E^{i}|D}{\eps\cdot2^{i}}\right)=\Ohat\left(\frac{n^{2}}{\eps^{2}}\right).
\]
As other components take $\Ohat(m)$ time, the total update time of
the algorithm for \Cref{thm: main for SSSP w distance bound} is $\Ohat(n^{2}/\eps^{2})$,
as required. It remains to show how the algorithm responds to queries
$\pquery_{D}(s,v)$ and $\distquery_{D}(s,v)$.

\paragraph*{Responding to $\pquery_{D}(s,v)$.}
Given a query $\pquery_{D}(s,v)$, we start by computing the unique simple $s$-$v$ path
$P$ in the tree $T$ given by \Cref{thm: main for maintaining a light graph}.
If vertex $v$ is not in $T$, then clearly $\dist_{G}(s,v)>2D$ and so we
report $\dist_{G}(s,v)>2D$. From now, we assume $v\in T$. Next,
we transform the path $P$ in $\hgl$ into an $s$-$v$ path
$P^{*}$ in the original graph $G$ as follows.

Let $v_{C_{1}},\ldots,v_{C_{z}}$ be all special vertices that appear
on the path $P$. For $1\leq k\leq z$, let $u_{k}$ be the regular
vertex preceding $v_{C_{k}}$ on $P$, and let $u'_{k}$ be the regular
vertex following $v_{C_{k}}$ on $P$. If $C_{k}$ is a connected
component of a heavy graph $G_{i}^{H}$ of class $i$, we use the query $\shortpath(C_{k},u_{k},u'_{k})$
in the $\lcd$ data structure for graph $G_{i}$ in order to to obtain a simple $u_{k}$-$u'_{k}$
path $Q_{k}$ contained in $C_{k}$, that contains at most $\frac{|V(C_{k})|\alpha}{\tau_{i}}$
 (unweighted) edges. Then, we replace the vertex $v_{C_{k}}$
 with the path $Q_{k}$ on path $P$. As we can find the path $P$
 in time $O(|P|)$, by following the tree $T$, and since the query time to compute each path  $Q_{k}$ is bounded by $|Q_{k}|\cdot(\gamma(n))^{O(q)}=\Ohat(|Q_{k}|)$, the total time to compute path $P^{*}$ is bounded by $\Ohat(|E(P^{*})|)$.  %

 We now bound the length of the path $P^*$.
 Recall that,  by \Cref{obs: distances in G and light graph}, path $P$ has length $(1+\eps/4)\dist_{\hat{G}^{L}}(s,v)\leq(1+\eps/4)\dist_{G}(s,v)$. For each $1\leq i\leq\lambda$,
 let $\cset_{i}=\{C_{k}\mid v_{C_{j}}\in P$ and $C_{k}$ is a connected
 component of $G_{i}^{H}$$\}$. Let $\qset_{i}$ be the set of all
 corresponding paths $Q_{k}$ of $C_{k}\in\cset_{i}$. We can bound the
 total length of all path in $\qset_{i}$ as follows:
 \[
 \sum_{Q\in \qset_i}\ell(Q)\leq \sum_{C_{k}\in\cset_{i}}|Q_{k}|\cdot2^{i+1}\le\sum_{C_{k}\in\cset_{i}}\frac{|V(C_{k})|\alpha}{\tau_{i}}\cdot2^{i+1}\le\sum_{C_{k}\in\cset_{i}}|V(C_{k})|\cdot\frac{\eps D}{4n\lambda}\leq\frac{\eps D}{4\lambda}
 \]
 
 (we have used the fact that $\tau_{i}=\frac{8n\lambda\alpha}{\eps D}\cdot2^{i}$, and that all components in $\cset_{i}$ are vertex-disjoint). Summing up over all $\lambda$ classes, the total
 length of all paths $Q_{k}$ corresponding to the super-nodes on path $P$ is at most $\eps D/4$. We conclude that $\ell(P^{*})\le\ell(P)+\eps D/4$.
 If $\dist_{G}(s,v)\ge D$, we have that $\ell(P^{*})\le(1+\eps/4)\dist_{G}(s,v)+\eps\dist_{G}(s,v)/4=(1+\eps/2)\dist_{G}(s,v)$. Notice that path $P^*$ may not be simple, since a vertex may belong to several heavy graphs $G_i^H$. However, for every edge $e\in E(G)$, there is a unique index $i$ such that $e\in G_i$, and the sets of edges of the heavy graph $G_i^H$ and the light graph $G_i^L$ are disjoint from each other. In particular, if $e\in E(G_i^H)$, then $e\not\in \hat G^L$. Since path $P$ is simple, all graphs $C_1,\ldots,C_z$ are edge-disjoint from each other, and their edges are also disjoint from $E(\hat G^L)$. We conclude that an edge may appear at most once on $P^*$.
 
\paragraph{Responding to $\protect\distquery_{D}(s,v)$.}
Given a query $\distquery_{D}(s,v)$, we simply return $\dist'(s,v)=\dist_{T}(s,v)+\eps D/4$ in time $O(1)$.
Recall that $\dist'(s,v)=\dist_{T}(s,v)+\eps D/4\ge\ell(P^{*})\ge\dist_{G}(s,v)$ (here, $P^*$ is the path that we would have returned in response to query $\pquery_{D}(s,v)$, though we only use this path for the analysis and do not compute it expliclty). As before if $\dist_{G}(s,v)\ge D$, then, from \Cref{obs: distances in G and light graph}, $\dist'(s,v)\le(1+\eps/2)\dist_{G}(s,v)$.

%% file: APSP_new.tex
\section{APSP}
\label{sec: APSP}

In this section, we prove \Cref{thm: main for APSP} by combining two algorithms. 
We use the function $\gamma(n) = \exp(O(\log^{3/4}n))$ from \Cref{thm:LCD}.

The first algorithm, summarized in the next theorem, is faster in the large-distance regime:
\begin{theorem}
        [APSP for large distances]\label{thm:APSP long}There is a deterministic
        algorithm, that, given parameters $0<\eps<1/2$ and $D>0$, and a
        simple unweighted undirected $n$-vertex graph $G$ that undergoes
        edge deletions, maintains a data structure using total update
        time of $\Ohat\left(n^{3}/(\eps^{3}D)\right)$ and supports the following
        queries: 
        \begin{itemize}
                \item $\distquery_{D}(u,v)$: either correctly declare that $\dist_{G}(u,v)>2D$
                in $O(\log n)$ time, or return an estimate $\dist'(u,v)$ in $O(\log n)$ time.
                If $D\leq\dist_{G}(u,v)\leq2D$, then $\dist_{G}(u,v)\le\dist'(u,v)\le(1+\eps)\dist_{G}(u,v)$ must hold.
                \item $\pquery_{D}(u,v)$: either correctly declare that $\dist_{G}(u,v)>2D$
                in $O(\log n)$ time, or return a $u$-$v$ path $P$ of length at most
                $9D$ in $\Ohat(|P|)$ time. If $D\leq\dist_{G}(u,v)\leq2D$,
                then $|P|\le(1+\eps)\dist_{G}(u,v)$ must hold.
        \end{itemize}
\end{theorem}

The second algorithm is faster for the short-distance regime.

\begin{theorem}
	[APSP for small distances]\label{thm:APSP short}There is a deterministic
	algorithm, that, given parameters $1\leq k<o(\log^{1/8}n)$ and $D>0$, and a simple
	unweighted undirected $n$-vertex graph $G$ that undergoes edge deletions,
	maintains a data structure using total update time $\Ohat(n^{2+3/k}D)$
	and supports the following queries: 
	\begin{itemize}
		\item $\distquery_{D}(u,v)$: in time $O(1)$, either correctly establish that $\dist_{G}(u,v)>2D$, or correctly establish that $\dist(u,v)\le  2^{k}\cdot3D+(\gamma(n))^{O(k)}$. %
		\item $\pquery_{D}(u,v)$: either correctly establish that $\dist_{G}(u,v)>2D$
		in $O(1)$ time, or return a $u$-$v$ path $P$ of length at most
		$2^{k}\cdot3D+(\gamma(n))^{O(k)}$, in time $O(|P|)+(\gamma(n))^{O(k)}$. 
	\end{itemize}
\end{theorem}

We prove Theorems \ref{thm:APSP long} and  \ref{thm:APSP short} below, after we complete the proof of
 \Cref{thm: main for APSP}
using them. Let $\eps=1/4$, and $D^{*}=n^{0.5-1/k}$.
For $1\le i\le \ceil{\log_{1+\eps}n}$, let $D_{i}=(1+\eps)^{i}$. For all $1\le i\le \ceil{\log_{1+\eps}n}$, if $D_i\leq D^*$, then we maintain the data structure from \Cref{thm:APSP short}
with the value $D=D_{i}$, and the input parameter $k$, and otherwise we maintain the data structure from \Cref{thm:APSP long} with the bound $D=D_i$ and the parameter $\eps$.  Since, from the statement of \Cref{thm: main for APSP}, $k \leq o(\log^{1/8} n)$ holds, it is easy to verify that the total update time for maintaining these data structures is bounded by $\Ohat(n^{2.5+2/k})$.

Given a query $\distquery(u,v)$,  we
 perform a binary search on indices $i$, in order to find an index for which $\dist_{G}(u,v)>2D_{i}$ and $\dist_{G}(u,v)<2^{k}\cdot3D_{i+1}+(\gamma(n))^{O(k)}$ hold, by querying the data structures form Theorems \ref{thm:APSP short} and \ref{thm:APSP long}. 
We then return $\widetilde{\dist}(u,v) = 2^{k}\cdot3\cdot D_{i+1}+(\gamma(n))^{O(k)}$ as a response to the query. Notice that we are guaranteed that $\widetilde{\dist}(u,v)\leq 2^{k}\cdot3\cdot \dist_G(u,v)+\Ohat(1)$, as required. As there are $O(\log n)$ possible values of $D_{i}$, the
query time is $O(\log n \log\log n)$. 

Given a query $\pquery(u,v)$, we start by checking whether $u$ and $v$ are connected, for example by running $\distquery_{D}(u,v)$ query with $D=(1+\eps)n$ on the data structure from \Cref{thm:APSP long}. If $u$ and $v$ are not connected, then we can report this in time $O(\log n)$. Otherwise, we perform a binary search on indices $i$ exactly as before, to find an index for which $\dist_{G}(u,v)>2D_{i}$ and $\dist_{G}(u,v)<2^{k}\cdot3D_{i+1}+(\gamma(n))^{O(k)}$ hold. Then, we use query in the appropriate data structure,
$\pquery_{D_{i+1}}(u,v)$ and obtain a $u$-$v$ path $P$ of length
at most $2^{k}\cdot3D_{i+1}+(\gamma(n))^{O(k)}\leq 2^{k}\cdot3\cdot \dist_G(u,v)+\Ohat(1)$, in time $\Ohat(|P|)$.

\subsection{The Large-Distance Regime}

The goal of this section is to prove \Cref{thm:APSP long}. The algorithm easily follows by combining our algorithm for \SSSP with the algorithm of \cite{GutenbergW20} for \APSP (that simplifies the algorithm of  \cite{henzinger16} for the same problem).

\subsection*{Data Structures and Update Time}
Our starting point is an observation of \cite{GutenbergW20}, that we can assume w.l.o.g. that throughout the edge deletion sequence, the graph $G$ remains connected.
Specifically, we will maintain a graph $G^*$, starting with $G^*=G$. Whenever an edge $e$ is deleted from $G$, as part of the input update sequence, if the removal of $e$ does not disconnect the graph $G$, then we delete $e$ from $G^*$ as well. Otherwise, we ignore this edge deletion operation, and edge $e$ remains in $G^*$. It is easy to see that in the latter case, edge $e$ is a bridge in $G^*$, and will remain so until the end of the algorithm. It is also immediate to verify that, if $u,v$ are two vertices that lie in the same connected component of $G$, then $\dist_G(u,v)=\dist_{G^*}(u,v)$. Moreover, if $P$ is any (not necessarily simple) path connecting $u$ to $v$ in graph $G^*$, such that an edge may appear at most once on $P$, then $P$ is also a $u$-$v$ path in graph $G$.

Throughout the algorithm, we use two parameters: $R^c=\eps D/8$ and $R^d=4D$. We maintain the following data structures.

\begin{itemize}
	\item Data structure  $\CONNSF(G)$ for dynamic connectivity. Recall that the data structure has total update time $\Otil(m)$, and it supports connectivity queries $\conn(G,u,v)$: given a pair  $u,v$ of vertices of $G$, return ``yes'' if $u$ and $v$ are connected in $G$, and ``no'' otherwise. The running time to respond to each such query is  $O(\log n/\log\log n)$. 
	
	\item A collection $S\subseteq V(G)$ of \emph{source vertices}, with $|S|\leq O(n/R^c)\leq O(n/(\eps D))$;
	
	\item For every source vertex $s\in S$, the data structure from \Cref{thm: main for SSSP w distance bound}, in graph $G^*$, with source vertex $s$, distance bound $R^d$, and accuracy parameter $\eps=1/4$.
\end{itemize}

Recall that the data structure from \Cref{thm: main for SSSP w distance bound} has total update time $\Ohat(n^2/\eps^2)$. Since we will maintain $O(n/(\eps D))$ such data structures, the total update time for maintaining them is $\Ohat(n^3/(\eps^3 D))$.

Consider now some source vertex $s\in S$, and the data structure from \Cref{thm: main for SSSP w distance bound} that we maintain for it. Since graph $G$ is unweighted, all edges of $G$ belong to a single class, and so the algorithm will only maintain a single heavy graph (instead of maintaining a separate heavy graph for every edge class), and a single light graph. In particular, this ensures that at any time during the algorithm's execution, all cores in $\bigcup_j\fset_j$ are vertex-disjoint.
In order to simplify the notation, we denote the extended light graph that is associated with graph $G^*$ by $\hat G^L$; recall that this graph does not depend on the choice of the vertex $s$. Recall that, from \Cref{obs: distances in G and light graph}, throughout the algorithm,
for every vertex $v\in V(G^*)$, $\dist_{\hat{G}^{L}}(s,v)\leq\dist_{G^*}(s,v)$ holds.  Additionally,  the data structure maintains an \EST, that we denote by $\tau(s)$, in graph $\hat G^L$, that is rooted at the vertex $s$, and has depth $R^d$. We say that the source $s$ \emph{covers} a vertex $v\in V(G)$ iff the distance from $v$ to $s$ in the tree $\tau(s)$ is at most $R^c$.

Our algorithm will maintain, together with each vertex $v\in V(G)$, a list of all source vertices $s\in S$ that cover $v$, together with a pointer to the location of $v$ in the tree $\tau(s)$. We also maintain a list of all source vertices $s'\in S$ with $v\in \tau(s')$, together with a pointer to the location of $v$ in $\tau(s')$. These data structures can be easily maintained along with the trees $\tau(s)$ for $s\in S$. The total update time for maintaining the \ESTs subsumes the additional required update time.

We now describe an algorithm for maintaining the set $S$ of source vertices. We start with $S=\emptyset$. Throughout the algorithm, vertices may only be added to $S$, but they may never be deleted from $S$. At the beginning, before any edge is deleted from $G$, we initialize the data structure as follows. As long as some vertex $v\in V(G)$ is not covered by any source, we select any such vertex $v$, add it to the set $S$ of source vertices, and initialize the data structure $\tau(v)$ for the new source vertex $v$. This initialization algorithm terminates once every vertex of $G$ is covered by some source vertex in $S$. As edges are deleted from $G$ and distances between vertices increase, it is possible that some vertex $v\in V(G)$ stops being covered by vertices of $S$. Whenever this happens, we add such a vertex $v$ to the set $S$ of source vertices, and initialize the corresponding data structure $\tau(v)$. We need the following claim.

\begin{claim}
	\label{prop:bound center}Throughout the algorithm,  $|S|\leq O(n/R^{c})$ holds.
\end{claim}

\begin{proof}
	For a source vertex $s\in S$, let $C(s)$ be the set of all vertices at distance at most $R^c/2$ from vertex $s$ in graph $\hat G^L$. From the algorithm's description, and since the distances between regular vertices in the graph $\hat G^L$ may only grow over the course of the algorithm, for every pair $s,s'\in S$ of source vertices, $\dist_{\hat G^L}(s,s')\geq R^c$ holds throughout the algorithm, and so $C(s)\cap C(s')=\emptyset$. Since graph $G^*$ is a connected graph throughout the algorithm, so is graph $\hat G^L$. 
	It is then easy to verify that, if $|S|\geq 2$, then for every source vertex $s\in S$, $|C(s)|\geq \Omega(|R^c|)$ (we have used the fact that graph $G$ is unweighted, and so, in graph $\hat G^L$, all edges have lengths in $\set{1/4,1}$). It follows that $|S|\leq O(n/R^c)$.
\end{proof}

\paragraph{Responding to $\pquery_{D}(x,y)$ queries.}
Suppose we are given a query $\pquery_{D}(x,y)$, where $x,y$ are two vertices of $G$. Recall that our goal is
to either correctly establishes that $\dist_{G}(x,y)>2D$, or to return
an $x$-$y$ path $P$ in $G$, of length at most $9D$. We also need to ensure that, if $D\leq\dist_G(x,y)\leq2D$,
then $|P|\le(1+\eps)\dist_{G}(x,y)$.

Our first step is to use query $\conn(G,x,y)$ in data structure  $\CONNSF(G)$  in order to check whether $x$ and $y$ lie in the same connected component of $G$. If this is not the case then we report that $x$ and $y$ are not connected in $G$. Therefore, we assume from now on that $x$ and $y$ are connected in $G$. 
Recall that the running time for query $\conn(G,x,y)$ is  $O(\log n/\log\log n)$. 

Recall that our algorithm ensures that there is some source vertex $s\in S$ that covers $x$. Therefore, $\dist_{\hat G_L}(s,x)\leq R^c$. 
It is also easy to verify that $\dist_{\hat{G}^{L}}(x,y)\leq\dist_{G^*}(x,y)$ must hold. Therefore, if $\dist_G(x,y)\leq 2D$, $y\in \tau(s)$ must hold. We can find the source vertex $s$ that covers $x$ and check whether $y\in \tau(s)$ in time $O(1)$ using the data structures that we maintain. If $y\not\in \tau(s)$, then we are guaranteed that $\dist_G(x,y)>2D$. We terminate the algorithm and report this fact.

Therefore, we assume from now on that $y\in \tau(s)$. We compute the unique simple $x$-$y$ path $P$ in the tree $\tau(s)$, by retracing the tree from $x$ and $y$ until we find their lowest common ancestor; this can be done in time $O(|P|)$. The remainder of the algorithm is similar to that for responding to queries for the \SSSP data structure.
We denote by $v_{C_{1}},\ldots,v_{C_{z}}$ the sequence of all special vertices that appear
on the path $P$. For $1\leq k\leq z$, let $u_{k}$ be the regular
vertex preceding $v_{C_{k}}$ on $P$, and let $u'_{k}$ be the regular
vertex following $v_{C_{k}}$ on $P$. We then use queries $\shortpath(C_{k},u_{k},u'_{k})$
to the $\lcd$ data structure in order to obtain a simple $u_{k}$-$u'_{k}$
path $Q_{k}$ contained in $C_{k}$. Then, we replace the vertex $v_{C_{k}}$
with the path $Q_{k}$ on path $P$. As in the analysis of the algorithm for \SSSP, the running time of this algorithm is bounded by $\Ohat(|E(P^{*})|)$, and the length of the path $P^*$ is bounded by $|P|+\eps R^d\leq \dist_G(x,y)+4\eps D$. Since $|P|\leq 2R^d\leq 8D$, this is bounded by $9D$. Moreover, if $D\leq \dist_G(x,y)\leq 2D$, then we are guaranteed that the length of $P^*$ is at most $(1+4\eps)\dist_G(x,y)$. The running time of the algorithm is $O(\log n)$ if it declares that $\dist_G(x,y)>2D$, and it is bounded by $\Ohat(|P^*|)$ if a path $P^*$ is returned.
We note that every edge may appear at most once on path $P^*$. Indeed, an edge of $G^*$ may belong to the heavy graph, or to the extended light graph $\hat G^L$, but not both of them. Therefore, an edge of $P$ may not lie on any of the paths in $\set{Q_1,\ldots,Q_z}$. Moreover, since path $P$ is simple, the connected components $C_1,\ldots,C_k$ of the heavy graph are all disjoint, and so the paths $Q_1,\ldots,Q_z$ must he disjoint from each other. Therefore, every edge may appear at most once on path $P^*$. As observed before, this means that $P^*$ is contained in the graph $G$.

\paragraph*{Responding to $\protect\distquery_{D}(x,y)$.}

The algorithm is similar to that for $\pquery_{D}(x,y)$.
As before,
our first step is to use query $\conn(G,x,y)$ in data structure  $\CONNSF(G)$  in order to check whether $x$ and $y$ lie in the same connected component of $G$. If this is not the case then we report that $x$ and $y$ are not connected in $G$. Therefore, we assume from now on that $x$ and $y$ are connected in $G$. 
Recall that the running time for query $\conn(G,x,y)$ is  $O(\log n/\log\log n)$. 

As before, we find a source $s$ that covers vertex $x$, and check whether $y\in \tau(s)$, in time $O(1)$. If this is not the case, then we correctly report that $\dist_G(x,y)>2D$, and terminate the algorithm. Otherwise, we return an estimate $\dist'(x,y)=\dist_{\hgl}(x,s)+\dist_{\hgl}(y,s)+4\eps D$. This can be done in time $O(1)$, by  reading the distance labels of $x$ and $y$ in tree $T(s)$. From
the above arguments, we are guaranteed that there is an $x$-$y$ path $P^{*}$
in $G$, whose length is at most $\dist'(x,y)$, so $\dist_{G}(x,y)\le\dist'(x,y)$ must hold.
Notice that $\dist_{\hgl}(y,s)\leq \dist_{\hgl}(x,s)+\dist_{\hgl}(x,y)\leq R^c+\dist_G(x,y)$. Therefore, 
$\dist'(x,y)\le  2R^c+4\eps D+\dist_G(x,y)\leq 8\eps D+\dist_G(x,y)$. Therefore, if $\dist_{G}(x,y)\ge D$, then  $\dist'(x,y)\le (1+8\eps)\dist_{G}(x,y)$
must hold.

In order to obtain the guarantees required in \Cref{thm:APSP long}, we use the parameter $\eps'=\eps/8$, and run the algorithm described above while using $\eps'$ instead of $\eps$. It is easy to verify that the resulting algorithm provides the desired guarantees.

\subsection{The Small-Distance Regime}

In this section, we prove \Cref{thm:APSP short}. Recall that we are given a simple
unweighted graph $G$ undergoing edge deletions, a parameter $k\ge1$ and
a distance scale $D$. We set $\Delta=n^{1/k}$ and $q=10k$. 

Our data structure is based on the $\lcd$ data structure from \Cref{thm:LCD}.
We invoke the algorithm from \Cref{thm:LCD} on the input graph $G$, with parameters $\Delta$ and $q$. Recall that the algorithm maintains a partition of the vertices of $G$ into layers $\Lambda_{1},\dots,\Lambda_{r+1}$, and notice that $r\le k+1$.
Let $\alpha=(\gamma(n))^{O(q)}$ be chosen such that, in response to the $\shortkpath$ and $\tocore$ queries, the length of the path returned by the $\lcd$
data structure is guaranteed to be at most $\alpha$. For every index $1<j\le r$, we define two distance parameters: $R_{j}^{d}$
called a \emph{distance radius} and $R_{j}^{c}$ called a \emph{covering radius} as follows: 
\[
R_{j}^{d}=2^{r-j}(3D+2\alpha k)\text{ and }R_{j}^{c}=R_{j}^{d}-2D.
\]
Note that $R_{j}^{d}\leq2^{k-1}\cdot3D+2^{k}\alpha k=O(D\cdot(\gamma(n))^{O(k)})$
for all $j>1$. (As $\Lambda_{1}=\emptyset$, we only give
the bound for all $j>1$). 
Recall that the $\lcd$ data structure maintains a collection $\fset_{j}$ of cores for each level $j > 1$. We need the following key concept:
\begin{definition}
	A vertex $v\in\Lambda_{j}$ is a \emph{far vertex} iff $\dist_{G}(v,\Lambda_{<j})>R_{j}^{d}$.
	A core $K\in\fset_{j}$ is a \emph{far core} iff all vertices in $K$ are
	far vertices, that is, $\dist_{G}(V(K),\Lambda_{<j})>R_{j}^{d}$.
\end{definition}

Observe that once a core $K$ becomes a far core, it remains a far core, until it is
destroyed. This is because distances in $G$ are non-decreasing, and both $\Lambda_{<j}$ and $V(K)$ are decremental vertex sets by \Cref{thm:LCD}.
At a high level, our algorithm can be described in one sentence: 

\begin{center} Maintain a collection of $\ESTs$
	of depth $R_{j}^{d}$ rooted at every far core in $\bigcup_j\fset_{j}$. 
\end{center}

Below, we describe the data structure in more detail and analyze its correctness.

\subsubsection{Maintaining Far Vertices and Far Cores}

\label{subsec:maintain far}

In this subsection, we show an algorithm that maintains, for every vertex of $G$, whether it is a far vertex. It also maintains, for every core of $\bigcup_j\fset_j$, whether it is a far core. Fix a layer $1<j\leq r$. Let $Z_{j}$ be a graph, whose vertex set is $V(G)$, and edge set contains all edges that have at least one endpoint in set $\Lambda_{\geq j}$. Equivalently, $E(Z_j)$ contains 
all edges incident to vertices with virtual degree at most $h_{j}$.
We construct another graph $Z_{j}'$ by adding a source vertex $s_{j}$
to $Z_{j}$, and adding, for every vertex $v\in\Lambda_{<j}$, an edge $(s,v)$ to this graph.
We maintain an $\EST$ $\hat{T}_{j}$ in graph $Z'_j$, with root $s_j$, and distance bound $(R_{j}^{d}+1)$.
Observe that $v\in\Lambda_{j}$ is a far vertex iff $v\notin V(\hat{T}_{j})$. 

Notice that graph $Z'_j$, in addition to undergoing edge deletions, may also undergo edge insertions.
Specifically, when a vertex $x$ is moved from  from $\Lambda_{<j}$
to $\Lambda_{\ge j}$ (that is, its virtual degree decreases from above $h_j$ to at most $h_j$), then we may need to insert all edges that are incident to $x$ into $Z'_j$. Note that edges connecting $x$ to vertices in $\Lambda_{\geq j}$ already belong to $Z'_j$, so we only need to insert  edges connecting $x$ to vertices of $\Lambda_{<j}$. We insert all such edges
$Z'_{j}$ first, and only then delete the edge $(s_{j},x)$ from $Z'_{j}$. Observe
that, for each such  edge $e=(x,y)\in E(x,\Lambda_{<j})$, inserting $e$ into $Z'_j$ may not decrease the distance from $s_j$ to $x$, or the distance from $s_j$ to $y$, as both these distances are currently $1$ and cannot be further decreased. It then follows that the insertion of the edge $e$ does not decrease the distance of any vertex from $s_j$. Therefore, the edge insertions satisfy the conditions of the \EST data structure.

As the total number of edges that ever appear
in $Z_{j}'$ is $O(n h_{j}\Delta)$ by \Cref{claim:bound edge},
the total update time for maintaining the data structure $\hat{T}_{j}$ is bounded by $O(n h_{j}\Delta R_{j}^{d})=O(n^{2+1/k}D(\gamma(n))^{O(k)})\leq \Ohat(n^{2+1/k}D)$ (we have used the fact that $h_j=\Delta^{r-j}$, $\Delta=n^{1/k}$, and $r\leq k+1$).

The above data structure allows us to maintain, for every vertex of $G$, whether it is a far vertex. For every core  $K\in \bigcup_j\fset_j$,  we simply maintain the number of vertices of $K$ that are far vertices. 
This allows us to maintain, for every core $K\in \bigcup_j\fset_j$, whether it is a far core. The
time that is required for tracking this information is clearly subsumed by the time for maintaining $\hat{T}_{j}$.
Therefore, the total time that is needed to maintain the information about far vertices and far cores, over all layers $j$, is bounded by $ \Ohat(n^{2+1/k}D)$.

\subsubsection{Maintaining $\protect\ESTs$ Rooted at Far Cores}

\label{subsec:maintain tree at far}

In this section, we define additional data structures that maintain $\ESTs$ that are rooted at the far cores, and analyze their total update time. 
Fix a layer $1<j\leq r$. Let $K\in\fset_{j}$ be a core in layer $j$, that is a far core. Let
$Z_{j}^{K}$ be the graph obtained from $Z_{j}$ by adding a source
vertex $s_{K}$, and adding, for every vertex  $v\in V(K)$, an edge $(s_{K},v)$. Whenever
a core $K$ is created in layer $j$, we check if $K$ is a far core. If
this is the case, then we initialize an $\EST$ $T_{K}$ in graph $Z_{j}^{K}$, with source $s_{K}$, and distance bound $(R_{j}^{d}+1)$. We maintain this data structure until core $K$ is destroyed. 
Additionally, whenever an existing core $K$ becomes a far core for the first time, we initialize the data structure $T_K$, and maintain it until $K$ is destroyed.

Observe that graph $Z_{j}^{K}$ may undergo both edge insertions and deletions.
As before, an edge may be inserted into $Z_{j}^{K}$ only when some vertex
$x$ is moved from $\Lambda_{<j}$ to $\Lambda_{\ge j}$ (recall that vertices may only be removed from a core $K$ after it is created). When vertex $x$
moves from $\Lambda_{<j}$  to $\Lambda_{\ge j}$, we insert all edges connecting $x$ to vertices of $\Lambda_{<j}$
into the graph $Z_{j}^{K}$. 
We claim that the insertion of such edges may 
not decrease the distance from $s_{K}$ to any vertex $v\in V(T_{K})$.
In order to see this, observe that, since vertex $x$ initially belonged to $\Lambda_{<j}$, and core $K$ was a far core,
$\dist_{G}(V(K),x)>R_{j}^{d}$. As edges are deleted
from $G$ and $K$, $\dist_{G}(V(K),x)$ may only grow. Therefore, when
vertex $x$ is moved to $\Lambda_{\ge j}$, its distance from the vertices of $K$ remains greater than  $R_{j}^{d}$, 
and so $\dist_{Z_{j}^{K}}(s_{K},x)>R_{j}^{d}+1$. As the depth of $T_K$ is $R^d_j+1$, inserting the edges of
$E(x,\Lambda_{<j})$ does not affect the distances of the vertices that belong to the tree $T_K$ from its root $s_K$.

Since, from by \Cref{claim:bound edge}, the total number of edges that may ever appear in $Z_{j}^{K}$ is $O(nh_{j}\Delta)$, the total time required for maintaining the $\EST$ $T_{K}$
is $O(nh_{j}\Delta)\cdot (R_{j}^{d}+1)$. By \Cref{thm:LCD}, the total
number of cores that are ever created in set $\fset_{j}$ over the course of the entire algorithm
the algorithm is at most $\Ohat(n\Delta/h_{j})$. Therefore, the total
update time that is needed in order to maintain trees $T_K$ for cores $K\in \fset_j$ is bounded by: 

$$O(nh_{j}\Delta R_{j}^{d})\cdot \Ohat(n\Delta/h_{j})=\Ohat (n^{2+2/k}D(\gamma(n))^{O(k)})=\Ohat (n^{2+2/k}D).$$

Summing this bound over all layers increases it by only factor $O(\log n)$.

\subsubsection{Total update time}

We now bound the total update time of the algorithm. Recall that the total update tiem of the $\lcd$ data structure is bounded by $\Ohat(m^{1+1/q}\Delta^{2+1/q}\leq \Ohat (mn^{3/k})$,
as $q=10k$ and $\Delta=n^{1/k}$. Each of the remaining data structures takes total update time at most $\Ohat (n^{2+2/k}D)$. Therefore, the total update time of the algorithm is bounded by  $\Ohat (n^{2+3/k}D)$. 

\subsubsection{Responding to Queries}

For any vertex $v\in\Lambda_{\ge j}$, we say that $v$ is \emph{covered}
by an $\EST$ $T_{K}$ iff $\dist_{Z_{j}}(V(K),v)\le R_{j}^{c}$ (i.e.~$\dist_{Z_{j}^{K}}(s_{K},v)\le R_{j}^{c}+1$).
For each $v\in\Lambda_{\ge j}$, we maintain a list of all $\ESTs$
$T_{K}$ that covers it. Within the list of $v$, we maintain the
core $K\in\fset_{j_{v}}$ from the smallest layer index $j_{v}$ such that
$T_{K}$ covers $v$. These indices can be explicitly maintained using the standard dictionary
data structure such as balanced binary search trees. The time for
maintaining such lists for all vertices is clearly subsumed by the
time for maintaining the $\ESTs$.

\paragraph*{Responding to $\protect\pquery_{D}(u,v)$.}

Given a pair of vertices $u$ and $v$, let $K_{u}$ be the core from
smallest level $j_{u}$ such that $T_{K_{u}}$ covers $u$ and $K_{v}$
be the core from smallest level $j_{v}$ such that $T_{K_{v}}$covers
$v$. Assume w.l.o.g.~that $j_{u}\le j_{v}$. If $v\notin T_{K_{u}}$,
then we report that $\dist_{G}(u,v)>2D$. 
Otherwise, compute the unique $u$-$v$ path $P$ in the tree $T_{K_{u}}$. This can be done in time
in time $O(|P|\log n)$, as follows. We maintain two current vertices $u',v'$, starting with $u'=u$ and $v'=v$. In every iteration, if the distance of $u'$ from the root of $T_{K_{u}}$ in tree $T_{K_{u}}$ is less than the distance of $v'$ from the root, we move $v'$ to its parent in the tree; otherwise, we move $u'$ to its parent. We continue this process, until we reach a vertex $z$ that is a common ancestor of both $u$ and $v'$. We denote the resulting $u$-$v$ path by $P$. Notice that so far the running time of the algorithm is $O(|E(P)|)$. Next, we consider two cases. First, if $z$ is not the root of the tree $T_{K_{u}}$, then $P$ is a path in graph $G$, and we return $P$.
Otherwise, the root of the tree $s_{K_{u}}$ lies on path $P$. We let $a$ and $b$ be the
vertices lying immediately before and immediately after $s_{K_{u}}$ in $P$. We compute $Q=\shortkpath(K_{u},a,b)$
in time $(\gamma(n))^{O(q)}$. Finally, we modify the path $P$ by replacing vertex $s_{K_{u}}$
with the path $Q$, and merging the endpoints $a$, $b$ of $Q$ with the copies of these vertices on path $P$. The resulting path, that we denote by $P'$, is a $u$-$v$ path in graph $G$. We return this path as the response to the query. It is immediate to verify that the
query time is $O(|E(P)|\log n)+(\gamma(n))^{O(q)} = \Ohat(|P|)$. %

We now argue that the response of the algorithm to the query is correct. 

Let $P^{*}$ be the shortest path between $u$ and $v$ in graph $G$. Let
$x$ be a vertex of $P^{*}$ that minimizes the index $j^*$ for which $x\in \Lambda_{j^*}$; therefore, $V(P^{*})\subseteq\Lambda_{\ge j^{*}}$. We start with the
following crucial observation.

\begin{lemma}
	\label{lem:near core}There is a far core $K'$ in some level $\Lambda_{j'}$,
	with $1< j'\le j^{*}$, such that $\dist_{Z_{j'}}(V(K'),x)\le R_{j'}^{c}-D$.
\end{lemma}
\begin{proof}
	Let $x_{1}=x$.  We gradually construct a path connecting $x_1$ to a vertex in a far core $K'$,
	as follows. First, using query $\tocore(x)$ of the \LCD data structure, we can obtain a path of length at most $\alpha$, connecting $x_{1}$
	to a vertex $a_{1}$ lying in some core $K_{1}$, such that, if $K_1\in \fset_{j_1}$, then $j_{1}\le j^{*}$. If $K_{1}$ is a far core, then we are done. Otherwise,
	there is a vertex $b_{1}$ in $K_{1}$ which is not a far vertex. By using  a query $\shortkpath(K_{1},a_{1},b_{1})$ of the \LCD data structure, we obtain  a path of length at most $\alpha$ connecting $a_{1}$ to $b_{1}$ inside the core $K_1$.
	As $b_{1}$ is not a far vertex, there must be some vertex $x_{2}\in\Lambda_{<j_{1}}$, for which $\dist_{Z_{j_{1}}}(b_{1},x_{2})\le R_{j_{1}}^{d}$. We repeat
	the argument for $x_2$ and subsequent vertices $x_i$, until we reach a vertex that lies in some far core $K'$. Note that, if $K'\in \fset_{j'}$, then $j'>1$ must hold, as $\Lambda_{1}=\emptyset$.
	Observe that, for each $i$, the constructed paths that connect $x_{i}$ and $a_{i}$,
	or connect $a_{i}$ to $b_{i}$, or connect $b_{i}$ to $x_{i+1}$, all
	lie inside $Z_{j'}$. By concatenating all these paths, we obtain a path in $Z_{j'}$, connecting $x$ to a core of $K'$. The length of the path is bounded by:

	\begin{align*}
	(2\alpha+R_{j^{*}}^{d})+(2\alpha+R_{j^{*}-1}^{d})+\dots+(2\alpha+R_{j'+1}^{d})+\alpha & \le R_{j^*}^{d}+R_{j^*-1}^{d}+\dots+R_{j'+1}^{d}+2\alpha k\\
	& =(3D+2\alpha k)(1+2+\dots+2^{r-(j'+1)})+2\alpha k\\
	& =(3D+2\alpha k)(2^{r-j'}-1)+2\alpha k\\
	& = R_{j'}^{d}-3D\\
	& = R_{j'}^{c}-D
	\end{align*}
	We conclude that that $\dist_{Z_{j'}}(V(K'),x)\le R_{j'}^{c}-D$.
\end{proof}

We assume w.l.o.g. that $x$ is closer to $u$ than $v$, that is,~$\dist_{G}(u,x)\le\dist_{G}(v,x)$.
Assume that $P^{*}$ has length at most $2D$. As $x$ lies in $P^{*}$
and $V(P^{*})\subseteq\Lambda_{\ge j^{*}}$, we get that $\dist_{Z_{j^{*}}}(u,x)\le\frac{2D}{2}=D$.
As $Z_{j^{*}}$ is a subgraph of $Z_{j'}$, we conclude that $\dist_{Z_{j'}}(u,x)\le\dist_{Z_{j^{*}}}(u,x)\le D$. Using the triangle inequality together with \Cref{lem:near core}, we get that $\dist_{Z_{j'}}(u,V(K'))\le\dist_{Z_{j'}}(u,x)+\dist_{Z_{j'}}(x,V(K'))\le R_{j'}^{c}$. In other words, tree $T_{K'}$ must cover $u$. Recall that we have let
$K_{u}$ be the core lying in smallest level $j_{u}$, such that $T_{K_{u}}$
covers $u$. Therefore, $j_u \le j'$ which implies that $V(P^{*})\subseteq\Lambda_{\ge j_{u}}$.
Therefore, path $P^*$ is contained in $Z_{j_u}$. Moreover, as $R_{j_{u}}^{d}=R_{j_{u}}^{c}+2D$ and $|P^{*}|\le2D$,
vertex $v$ must be contained in $T_{K_{u}}$ as well. If this is not the case, then we can conclude
that $|P^{*}|>2D$. The same argument applies if the index $j_v$ of the layer $\Lambda_{j_v}$ to which the core $K_v$ belongs is smaller than $j_u$.

Let $P$ be the unique $u$-$v$ path in the tree $T_{K_{u}}$. Clearly, $|P|\le\dist_{T_{K_{u}}}(s_{K_{u}},u)+\dist_{T_{K_{u}}}(s_{K_{u}},v)\le2R_{j_{u}}^{d}\le2^{k}\cdot3D+(\gamma(n))^{O(k)}$.
If the root vertex $s_{K_{u}}$ of the tree does not lie on the path $P$, then path $P$ is a $u$-$v$
path in graph $G$, whose  length is bounded by $2^{k}\cdot3D+(\gamma(n))^{O(k)}$; the algorithm then returns $P$.
Otherwise, the algorithm replaces the vertex $s_{K_{u}}$with the path $Q$ returned by the query $\shortkpath(K_{u},a,b)$ to the \LCD data structure,
where $a$ and $b$ are the vertices of $P$ appearing immediately before and after $s_{K_{u}}$
on it. As $|Q|\le\alpha$, the length of returned path is bounded by $2R_{j_{u}}^{d}+\alpha\le2^{k}\cdot3D+(\gamma(n))^{O(k)}$.

\paragraph*{Responding to $\protect\distquery_{D}(u,v)$.}

The algorithm for responding to $\protect\distquery_{D}(u,v)$ is similar. 
As before, we let $K_{u}$ be the core from
smallest level $j_{u}$ such that $T_{K_{u}}$ covers $u$, and we let $K_{v}$
be the core from smallest level $j_{v}$ such that $T_{K_{v}}$covers
$v$. Assume w.l.o.g.~that $j_{u}\le j_{v}$. If $v\notin T_{K_{u}}$,
then we report that $\dist_{G}(u,v)>2D$.
Otherwise, we declare that $\dist(u,v)\leq 2^{k}\cdot3D+(\gamma(n))^{O(k)}$. The correctness of this algorithm follows immediately from the analysis of the algorithm for responding to $\pquery_{D}(u,v)$. The algorithm can be implemented to run in time $O(1)$ if we store, together with every vertex $v\in V(G)$, the list of the cores that cover $v$, sorted by the index $j$ of the set $\fset_j$ to which the core belongs. It is easy to see that time that is required to maintain this data structure is subsumed by the total update time of the algorithm that was analyzed previously.

%% file: Appendix.tex
\section{Proofs Omitted from \Cref{sec: prelims}}
\label{sec:omit}

\subsection{Proof of \Cref{obs:DSP}: Degree Pruning}

	It is immediate that the degree of every vertex in graph $H[A]$ is
	at least $d$. We now prove that $A$ is the unique maximal set with
	this property at any time. Assume for contradiction that at some time
	there is a subset $A'\subseteq V(H)$ where every vertex in $H[A']$
	has degree at least $d$ but $A'\not\subseteq A$. Denote $\{v_{1},\ldots,v_{r}\}=V(H)\setminus A$
	where the vertices are indexed in the order in which they were removed
	from $A$. Then there must be some vertex $v\in A'\setminus A$. Let
	$v_{i}$ be such a vertex with the smallest index $i$. But then $v_{1},\ldots,v_{i-1}\not\in A'$,
	so $v_{i}$ must have fewer than $d$ neighbors in $A'$, a contradiction.

\section{Proofs Omitted from \Cref{sec: LCD}}
\label{sec: omit lcd}
\subsection{Proof of \Cref{claim:bound edge}: Bounding Number of Edges Incident to Layers}

Fix some index $1\leq j\leq r$.
In order to define an $(h_{j}\Delta)$-orientation of $E_{\ge j}$,
we first define an ordering $\rho$ of the vertices of $V(G)$. Consider
the following experiment. We run $\algPS(G,h_{j-1})$ in order to
maintain the vertex set $A_{j-1}$, as $G$ undergoes edge deletions.
For a vertex $v\in V(G)$, we define its \emph{drop time} to be the
first time in the execution of this algorithm when $v$ did not belong
to set $A_{j-1}$; if no such time exists, then the drop time of $v$
is infinite. Recall that, from \Cref{obs: maintain DSP}, if the drop
time of $v$ is finite and equal to $t$, then at time $t$, $v$
had fewer than $h_{j-1}=\Delta h_{j}$ neighbors in $A_{j-1}$. We
let $\rho$ be the ordering of the vertices of $V(G)$ by their drop
time, from smallest to largest, breaking ties arbitrarily. Notice
that every edge in $E_{\ge j}$ must have an endpoint with a finite
drop time. Consider now some edge $e=(u,v)\in E_{\ge j}$. If $u$
appears before $v$ in the ordering $\rho$, then we assign the direction
of the edge $e$ to be from $u$ to $v$; note that, from the definition
of $E_{\ge j}$, the drop time of $u$ must be finite. This gives 
a $(h_{j}\Delta)$-orientation for $E_{\ge j}$. It now follows immediately
that $|E_{\ge j}|\le\Delta h_{j}n$. 

Next, let $S_{j}$ be the set of vertices that join the layer $\Lambda_{j}$
at any time of the algorithm's execution. Observe that $|S_{j}|\le n_{\le j}$ must hold because
virtual degrees may only decrease, and so $|E_{\ge j}(S_{j})|\le n_{\le j}\cdot h_{j}\Delta$.
As the edges whose both endpoints are contained in $\Lambda_{j}$
at any point of time must belong to $E_{\ge j}(S_{j})$, the number
of such edges is at most $n_{\le j}h_{j}\Delta$. We conclude that  the number of edges $e$, such that, at any time during the algorithm's execution, both endpoints of $e$ are contained
in $\Lambda_{j}$ is at most $n_{\le j}h_{j}\Delta$.

\subsection{Existence of Expanding Core Decomposition}
\label{sec: expanding core decomposition existence}

The goal of this section is to prove the following theorem about the existence of a core decomposition in a high-degree graph. We note that a theorem that is very similar in spirit (but different in the exact definitions and parameters) was shown in \cite{fast-vertex-sparsest}, and the proof that we provide uses similar ideas.

\begin{theorem}
	[Expanding Core Decomposition]\label{fact:xcore decomp}Let $H$
	be an $n$-vertex simple graph with minimum degree at least $h$.
	There exists a collection $\fset=\{K_{1},\dots,K_{t}\}$ of vertex-disjoint
	induced subgraphs, called \emph{expanding cores} or just \emph{cores},
	where $t=O((n\log n)/h)$ such that
	\begin{itemize}
		\item Each core $K\in\fset$ is a $\phi$-expander and $\deg_{K}(u)\ge\phi h/3$
		for all $u\in V(K)$ where $\phi=\Omega(1/\log n)$. Moreover,
		$K$ has diameter $O((\log n)/\phi)$ and is $(\phi h/3)$-edge-connected. 
		\item For each vertex $u\notin\bigcup_{K\in\fset}V(K)$, there are at least
		$2h/3$ edge-disjoint paths of length $O(\log n)$ from $u$ to vertices
		in $\bigcup_{K\in\fset}V(K)$.
	\end{itemize}
\end{theorem}

\begin{proof}
We start with the following two propositions. 
\begin{prop}
	\label{prop:ex decomp opt} Let $G=(V,E)$ be an $n$-vertex $m$-edge graph. Then
	there is a partition $V_{1},\dots,V_{k}$ of $V$ into disjoint sets,
	such that $\sum_{i=1}^{k}\delta(V_{i})\le m/2$, and for all $1\leq i\leq k$, $G[V_{i}]$
	is strong $\phi$-expander w.r.t.~$G$ where $\phi=\Omega(1/\log n)$.
\end{prop}
\begin{proof}
	The well-known $\phi$-expander decomposition (e.g.~Observation 1.1. of \cite{expander-pruning}) says that, given any
	graph $G=(V,E)$ with $m$ edges (possibly with self-loops and multi-edges)
	and a parameter $\phi$, there exists a partition $V_{1},\dots,V_{k}$
	of $V$ such that $\sum_{i=1}^{k}\delta_{G}(V_{i})\le O(\phi m\log m)$
	and $G[V_{i}]$ is a $\phi$-expander. 
	
	Let $G'$ be obtained from $G$ by adding, for each vertex $v$, $\deg_{G}(v)$
	 self-loops at $v$. We claim that a $\phi$-expander decomposition
	$V'_{1},\dots,V'_{k}$ of $G'$ where $\phi=\Omega(1/\log m)$ is
	indeed the desired strong expander decomposition for $G$. This is
	because, for any set $\emptyset\neq S\subset V_{i}$, we have $\vol_{G'[V'_{i}]}(S)\ge\vol_{G}(S)$
	because of the self-loops and $\delta_{G'[V_{i}]}(S)=\delta_{G[V_{i}]}(S)$.
	So we have that $\frac{\delta_{G[V'_{i}]}(S)}{\min\{\vol_{G}(S),\vol_{G}(V'_{i}\setminus S)\}}\ge\frac{\delta_{G'[V'_{i}]}(S)}{\min\{\vol_{G'[V'_{i}]}(S),\vol_{G'[V'_{i}]}(V_{i}\setminus S)\}}\ge\phi$.
	That is, $G[V'_{i}]$ is indeed a strong $\phi$-expander with respect
	to $G$. Also, for each $i$, $\delta_{G}(V'_{i})=\delta_{G'}(V'_{i})$.
	So we have $\sum_{i=1}^{k}\delta_{G}(V'_{i})=\sum_{i=1}^{k}\delta_{G}(V'_{i})\le O(\phi\cdot(2m)\log(2m))\le m/2$
	by choosing an appropriate constant in $\phi=\Omega(1/\log m)$.
\end{proof}

\begin{prop}
	\label{prop:ex decomp opt big}Let $H'$ be an $n$-vertex graph with
	minimum degree $h'$. Then there is a collection $\fset'$ of vertex-disjoint
	induced subgraphs of $H'$ that we call cores, such that:
	\begin{itemize}
		\item Each core $K\in\fset'$ is a $\phi$-expander and for all $u\in V(K)$, $\deg_{K}(u)\ge\phi h'$, where $\phi=\Omega(1/\log n)$; and
		\item $\sum_{K\in\fset'}|E(K)|\ge 3|E(H')|/4$.
	\end{itemize}
\end{prop}

\begin{proof}
	We apply \Cref{prop:ex decomp opt} to graph $H'$ to obtain a partition $(V_{1},\dots,V_{k})$ of $V(H')$.
	We then let $\fset$ contain all graphs  $H'[V_{i}]$ with $|V_{i}|\ge2$. Notice that, from  \Cref{prop:ex decomp opt}, each such graph $H'[V_{i}]$ is a $\phi$-expander. Moreover, from by \Cref{obs: high degrees in strong expander}, for all $u\in V_i$, $\deg_{G[V_{i}]}(u)\geq \phi h'$
	Lastly, observe that $\sum_{K\in\fset'}|E(K)|=|E(H')|-(\sum_{i=1}^{k}\delta_{H'}(V_{i}))/2\ge3|E(H')|/4$. 
\end{proof}

We are now ready to provide the algorithm for constructing the core decomposition, that will be used in the proof of \Cref{fact:xcore decomp}.

\paragraph{The algorithm.}

We start with $\fset\gets\emptyset$, $H'\gets H$, and $h'\gets h/3$.
Let $A=\DP(H',h')$. We set $H'\gets H'[A]$, so  that $H'$ has minimum degree at least
$h'$. Then, we apply \Cref{prop:ex decomp opt big} to $H'$ and
obtain the collection $\fset'$ of cores. We set $\fset\gets\fset\cup\fset'$ and delete all
vertices in $\bigcup_{K\in\fset'}V(K)$ from $H'$. Then, we again
set $A=\DP(H',h')$ and repeat this process until $H'=\emptyset$. Let $\fset$ be the final collection of cores that the algorithm computes. We now prove that it has all required properties.

\paragraph{The first guarantee.}

\Cref{prop:ex decomp opt big} directly guarantees that each core $K\in\fset$
is a $\phi$-expander, and moreover, for all $u\in V(K)$, $\deg_{K}(u)\ge\phi h'$. By the standard ball-growing argument, any $\phi$-expander
has diameter at most $O(\log(n)/\phi)=O(\log^{2}n)$. Next, to prove
that $K$ is $(\phi h')$-edge connected, it is enough show that, for any vertex set
$S\subseteq V(K)$ with $\vol_{K}(S)\le\vol(K)/2$, $\delta_{K}(S)\ge\phi h'$ holds.
Observe that, since $K$ is a $\phi$-expander, $\delta_{K}(S)\ge\phi\vol_{K}(S)\ge\phi^{2}h'|S|$ must hold. At the same time, since the minimum degree in $K$ is at least $\phi h'$ and $K$
is a simple graph, $\delta_{K}(S)\ge\phi h'|S|-\binom{|S|}{2}$ must hold. We now consider two cases. First, if $|S|\ge1/\phi$, then $\phi^{2}h'|S|\ge\phi h'$.
Otherwise, it can be verified that $\phi h'|S|-\binom{|S|}{2}\ge\phi h'$
for all $1\le|S|<1/\phi$. In any case, $\delta_{K}(S)\ge\phi h'$. 

\paragraph{The second guarantee.}

We denote $U=V(H)\setminus\bigcup_{K\in\fset}V(K)$. Note that $v\in U$ only if, for some graph $H'$ that arose over the course of the algorithm, $v\not \in A$, where $A=\DP(H',h')$.
We say that vertex $v$ was \emph{removed} when procedure $\DP$ was applied to that graph $H'$.
 By orienting edges incident to $v$ towards $v$ whenever $v$ is
removed, we can orient all edges of $H$ incident to the vertex set $U$ such that $\indeg_{H}(v)\le h'$
for each $v\in U$. Let $\overrightarrow{H}$ be a directed graph obtained from $H$
by contracting all vertices in $\bigcup_{K\in\fset}V(K)$ into
a single vertex $t$, while keeping the orientation of edges incident
to $U$. Observe that $V(\overrightarrow{H})=U\cup\{t\}$ and $\overrightarrow{H}$
is a DAG with $t$ as a single sink. It is now enough to show that, for every vertex $u\in U$, there
are $2h/3$ edge-disjoint directed paths of length $O(\log n)$ in
$\overrightarrow{H}$ from $u$ to $t$.

For any $S\subseteq V(\overrightarrow{H})$, let $\invol_{\overrightarrow{H}}(S)=\sum_{u\in S}\indeg_{\overrightarrow{H}}(u)$,
$\outvol_{\overrightarrow{H}}(S)=\sum_{u\in S}\outdeg_{\overrightarrow{H}}(u)$,
and $\vol_{\overrightarrow{H}}(S)=\invol_{\overrightarrow{H}}(S)+\outvol_{\overrightarrow{H}}(S)$.
Observe that, for $v\in U$, $\outdeg_{\overrightarrow{H}}(v)\ge2\indeg_{\overrightarrow{H}}(v)$
because $\indeg_{\overrightarrow{H}}(v)\le h'=h/3$ but $\deg_{\overrightarrow{H}}(v)\ge h$.
So, for any $S\subseteq U$, $\outvol_{\overrightarrow{H}}(S)\ge2\invol_{\overrightarrow{H}}(S)$. 

Fix a vertex $u\in U$. Let $B_{d}=\{v\mid\dist_{\overrightarrow{H}}(u,v)\le d\}$.
Suppose that $B_{d}\subseteq U$, then we have
\begin{align*}
\vol_{\overrightarrow{H}}(B_{d+1}) & =\vol_{\overrightarrow{H}}(B_{d})+\vol_{\overrightarrow{H}}(B_{d+1}\setminus B_{d})\\
& \ge\vol_{\overrightarrow{H}}(B_{d})+|E_{\overrightarrow{H}}(B_{d},B_{d+1}\setminus B_{d})|\\
& =\vol_{\overrightarrow{H}}(B_{d})+|E_{\overrightarrow{H}}(B_{d},B_{d+1})|-|E_{\overrightarrow{H}}(B_{d},B_{d})|\\
& \ge\vol_{\overrightarrow{H}}(B_{d})+\outvol_{\overrightarrow{H}}(B_{d})-\invol_{\overrightarrow{H}}(B_{d})\\
& \ge\vol_{\overrightarrow{H}}(B_{d})+\vol_{\overrightarrow{H}}(B_{d})/3=(4/3)\vol_{\overrightarrow{H}}(B_{d})
\end{align*}
where the last inequality is because $\outvol_{\overrightarrow{H}}(B_{d})\ge2\invol_{\overrightarrow{H}}(B_{d})$.
This proves that $t\in B_{3\log_{4/3}n}$, otherwise $\vol_{\overrightarrow{H}}(B_{3\log_{4/3}n})\ge(4/3)^{3\log_{4/3}n}\ge n^{3}$
which is a contradiction. This implies that there is a  directed $u$-$t$
path $P$ of length $O(\log n)$ in $\overrightarrow{H}$, but we
want to show that there are many such edge-disjoint paths.

Observe that the argument above only exploits the fact that $\outdeg_{\overrightarrow{H}}(v)\ge2\indeg_{\overrightarrow{H}}(v)$
for all $v\in U$. So even if we remove edges of a $u$-$t$ path
$P$ from $\overrightarrow{H}$, this inequality still holds for all
$v\in U\setminus\{u\}$. As we can assume that $\indeg_{\overrightarrow{H}}(u)=0$
because in-coming edges to $u$ do not play a role for finding $u$-$t$
paths, we also have $\outdeg_{\overrightarrow{H}}(u)\ge2\indeg_{\overrightarrow{H}}(u)=0$.
Therefore, we can repeat the argument $\outdeg_{\overrightarrow{H}}(u)\ge h-h'\ge2h/3$
times, and obtains $2h/3$ edge-disjoint $u$-$t$ paths in $\overrightarrow{H}$.
So we conclude that, for each vertex $u\in U=V(H)\setminus\bigcup_{K\in\fset}V(K)$,
there are $2h/3$ edge-disjoint paths of length $O(\log n)$ from
$u$ to vertices in $\bigcup_{K\in\fset}V(K)$.

\end{proof}

\input{embedding.tex}

%% file: embedding.tex
\subsection{Proof of \Cref{thm: expander decomp}: Strong Expander Decomposition}
\label{sec: strong expander decomp}

We will use the recent almost-linear time determinstic algorithm for
computing a (standard) expander decomposition by Chuzhoy et al.~\cite{ChuzhoyGLNPS19}.
\begin{thm}
	[Restatement of Corollary 7.7 from \cite{ChuzhoyGLNPS19}]\label{thm:standard exp decomp}There
	is a deterministic algorithm that, given a graph $G=(V,E)$ with $m$
	edges (possibly with self-loops and parallel edges), a parameter $\phi\in(0,1)$,
	and a number $r\ge1$, computes a partition of $V$ into disjoint
	subsets $V_{1},\dots,V_{k}$ such that $\sum_{i=1}^{k}\delta_{G}(V_{i})\le\phi m\cdot(\log m)^{O(r^{2})}$,
	and for all $1\le i\le k$, $G[V_{i}]$ is a $\phi$-expander. The
	running time of the algorithm is $O(m^{1+O(1/r)+o(1)}\cdot(\log m)^{O(r^{2})})$.
\end{thm}

We can now complete the proof of \Cref{thm: expander decomp} using \Cref{thm:standard exp decomp}.
Given an input graph $G=(V,E)$ for \Cref{thm: expander decomp},
we construct a graph $G'$ as follows. We start by setting $G' \gets G$ and, for
each vertex $v\in V$, we add $\deg_{G}(v)$ self-loops to it 
in $G'$. We then apply \Cref{thm:standard exp decomp} to graph  $G'$, with parameters $\phi$ and $r=\log^{1/4}m$, to obtain a partition of
$V(G')$ into disjoint subsets $V_{1},\dots,V_{k}$ such that $\sum_{i=1}^{k}\delta_{G'}(V_{i})\le(\log|E(G')|)^{O(r^{2})}\cdot\phi\cdot|E(G')|$,
and for all $1\le i\le k$, $G'[V_{i}]$ is a $\phi$-expander. 

First, observe that for each $i$, $\delta_{G}(V_{i})=\delta_{G'}(V_{i})$
because $G$ and $G'$ differ only by self-loops. So we have $\sum_{i=1}^{k}\delta_{G}(V_{i})=(\log|E(G')|)^{O(r^{2})}\cdot\phi\cdot|E(G')|\le\gamma(m)\phi m$. 

Second, observe that, for any set $\emptyset\neq S\subsetneq V_{i}$,
$\vol_{G'[V_{i}]}(S)\ge\vol_{G}(S)$ because of the self-loops in
$G'$, and $\delta_{G'[V_{i}]}(S)=\delta_{G[V_{i}]}(S)$. So we have
that $\frac{\delta_{G[V_{i}]}(S)}{\min\{\vol_{G}(S),\vol_{G}(V_{i}\setminus S)\}}\ge\frac{\delta_{G'[V_{i}]}(S)}{\min\{\vol_{G'[V_{i}]}(S),\vol_{G'[V_{i}]}(V_{i}\setminus S)\}}\ge\phi$.
That is, for each $i$, $G[V_{i}]$ is indeed a strong $\phi$-expander with respect
to $G$. Therefore, we can simply return the partition $\{V_{1},\dots,V_{k}\}$
as an output for \Cref{thm: expander decomp}. The running time is
 $O(m^{1+O(1/r)+o(1)}\cdot(\log m)^{O(r^{2})})=\Ohat(m)$
by \Cref{thm:standard exp decomp}.

\subsection{Proof of \Cref{lem:embed}: Embedding Small Expanders}
\label{sec:embedding}

In this section we prove \Cref{lem:embed}.
The proof uses the \emph{cut-matching game}, that was introduced by
Khandekar, Rao, and Vazirani \cite{KRV} as part of their fast randomized
algorithm for the Sparsest Cut and Balanced Cut problems.  
Chuzhoy et al. \cite{ChuzhoyGLNPS19} provided an efficient deterministic implementation of this game (albeit with weaker parameters), based
on a variation of this game due to Khandekar et al.~\cite{KhandekarKOV2007cut}.
We start by describing the variant of the Cut-Matching game that we use, that is based
on the results of \cite{ChuzhoyGLNPS19}.

\subsubsection{Deterministic Cut-matching Game }

\label{sec:CMG}
 The \emph{cut-matching game} is a game that is  played between two players, called the \emph{cut
	player} and the \emph{matching player}. The game starts with a graph $W$ whose vertex set $V$ has cardinality $n$, and $E(W)=\emptyset$. The game is played in rounds; in each round $i$, the cut player chooses a partition $(A_{i},B_{i})$
of $V$ with $|A_{i}|\le|B_{i}|$. The matching player then
chooses an arbitrary matchings $M_{i}$ that matches {\bf every} vertex of $A_i$ to some vertex of $B_i$. The edges of $M_i$ are then added to $W_i$, completing the current round.
Intuitively, the game terminates once graph $W$ becomes a $\psi$-expander, for some given parameter $\psi$. It is convenient to think of the cut player's goal as minimizing the number of rounds, and of the matching player's goal as making the number of rounds as large as possible.
 We prove the following theorem, that easily follows from 
 \cite{ChuzhoyGLNPS19}.

\begin{thm}
	[Deterministic Algorithm for Cut Player]\label{thm:CMG}There
	is a deterministic algorithm, that, for every round $i\geq 1$, given the graph $W$ that serves as input to the $i$th round of the cut-matching game, produces, in time $O(n\gamma(n))$, a partition  $(A_{i},B_{i})$ of $V$ with $|A_i|\leq |B_i|$, such
	that, no matter how the matching player plays, after $R=O(\log n)$ rounds, the resulting graph $W$ is a $1/\gamma(n)$-expander. 
\end{thm}

\begin{proof}
For the sake of the proof, it is more convenient to work with the notion of
\emph{sparsity} instead of \emph{conductance}. %
\begin{defn}
	[Sparsity] The sparsity $\Psi(G)$ of a graph $G=(V,E)$ is the minimum, over all vertex sets $S\subseteq V$ with $1\leq |S|\le|V\setminus S|$, of $\delta(S)/|S|$.
\end{defn}

From the definition, it is immediate to see that, if a graph $G$ has maximum degree $d$, then $\Phi(G)\le\Psi(G)\le d\cdot\Phi(G)$. In particular, if $\Psi(G)\geq \phi'$ for any parameter $\phi'$, then $G$ is a $(\phi'/d)$-expander.
Clearly, for any subgraph $H\subseteq G$ with $V(H)=V(G)$, $\Psi(H)\le\Psi(G)$ must hold. We need the following observation:
\begin{observation}
	[Observation 2.3 of \cite{ChuzhoyGLNPS19}]\label{obs: exp plus matching is exp}Let
	$G=(V,E)$ be an $n$-vertex graph with $\Psi(G)\ge\psi$, and let
	$G'$ be another graph that is obtained from $G$ by adding to it
	a new set $V'$ of at most $n$ vertices, and a matching $M$ connecting
	every vertex of $V'$ to a distinct vertex of $G$. Then $\Psi(G')=\Omega(\psi)$.
\end{observation}

In order to implement the algorithm of the cut player, we will employ the
following algorithm by \cite{ChuzhoyGLNPS19}:
\begin{thm}
	[Theorem 1.6 of \cite{ChuzhoyGLNPS19}]There is a deterministic algorithm,
	called $\cutorcert$, that, given an $n$-vertex graph $G=(V,E)$ with
	maximum vertex degree $O(\log n)$, and a parameter $r\ge 1$, returns
	one of the following:
	\begin{itemize}
		\item either a cut $(A,B)$ in $G$ with $n/4\le|A|\le|B|$ and $|E_{G}(A,B)|\le n/100$;
		or
		\item a subset $S\subseteq V$ of at least $n/2$ vertices, such that $\Psi(G[S])\ge1/\log^{O(r)}n$.
	\end{itemize}
	The running time of the algorithm is $O(n^{1+O(1/r)}\cdot(\log n)^{O(r^{2})})$. 
\end{thm}

The following lemma (first proved by \cite{KhandekarKOV2007cut}) shows 
that it is impossible for the cut player to return a balanced sparse
cut for more than $O(\log n)$ iterations of the cut-matching game.
\begin{lem}
	[Restatement of Theorem 2.5 of \cite{ChuzhoyGLNPS19}]\label{lem:cmg bound round} There is a constant $c$, for which the following holds. Consider the cut-matching game where in each iteration $1\leq i\leq c\log n$, we use Algorithm $\cutorcert$ in order to implement the cut player;  specifically, if the algorithm returns a partition $(A_{i},B_{i})$ of $V$ with $|E_{W}(A_{i},B_{i})|\le n/100$ and $n/4\le|A_{i}|\le|B_{i}|$, then we let $(A'_i,B'_i)$ be any partition of $V$ with $A'_i\subseteq A_i$ and $|A'_i|=|B'_i|$, and we use the partition $(A'_i,B'_i)$ as the response of the cut player in round $i$. Otherwise, we terminate the cut-matching game. Then no matter how the matching player plays in each iteration, the game will be terminated before reaching iteration $\floor{c\log n}$.
\end{lem}

We are now ready to complete the proof of \Cref{thm:CMG}. In each iteration $1\leq i\leq O(\log n)$ of the cut-matching game, we apply algorithm $\cutorcert$ to graph $W_{i-1}$
(where $W_{0}$ is the initial graph with $E(W)=\emptyset$), with the parameter $r=O(\log^{1/4}n)$. Since the edge set $E(W_{i-1})$ is the union of $i-1$ matchings, graph $W_{i-1}$ it has maximum degree at most $i-1$. We will ensure that the number of rounds is bounded by $O(\log n)$, so graph $W_{i-1}$ is a valid input
for $\cutorcert$.

If $\cutorcert$ returns a cut $(A_{i},B_{i})$ with $E_{W_{i-1}}(A_{i},B_{i})\le n/100$
and $n/4\le|A_{i}|\le|B_{i}|$, then we 
output an arbitrary partition $(A'_i,B'_i)$ of $V$ with $A'_i\subseteq A_i$ and $|A'_i|=|B'_i|$.
By \Cref{lem:cmg bound round}, this can happen for at most $O(\log n)$
iterations, regardless of the responses of the matching player. Otherwise, if $\cutorcert$ returns a subset $S\subseteq V$ of
at least $n/2$ vertices, with $\Psi(W_{i-1}[S])\geq1/\log^{O(r)}n$, we output the partition $(A_i,B_i)$, where $A_i=V\setminus S$ and $B_i=T$.
Let $M_{i}$ be the matching returned by the matching player,
that matches every vertex of $V\setminus S$ to a distinct vertex of $S$.
We are then guaranteed that the graph $W_{i}= W_{i-1}\cup M_{i}$ is a $\phi'$-expander, where $\phi'=1/\log^{O(r)}n=1/(\log n)^{O(\log^{1/4}n)}\geq \gamma(n)$ (recall that $\gamma(n)=\exp(\log^{3/4}n)$; we have also used the fact that the maximum vertex degree in $W_i$ is at most $O(\log n)$).
We conclude that $W_i$ is a $(1/\gamma(n))$-expander.

The running time of the algorithm is dominated by Algorithm $\cutorcert$, whose running time is $O(n^{1+O(1/r)}\cdot(\log n)^{O(r^{2})})=O(n\gamma(n))$. This completes the proof of \Cref{thm:CMG}.
\end{proof}

In order to complete the proof of \Cref{lem:embed}, we next provide an efficient deterministic algorithm for the cut player. The idea (that is quite standard) is that, in addition to producing the required matching $M_i$ in each round of the game, the cut player will also embed the edges of $M_i$ into the graph $G$, where the embedding paths have a relatively short length and cause a relatively small congestion.

\subsubsection{Implementing the Matching Player}

\label{subsec:termmatching}

There are well-known efficient algorithms, that, given a $\phi$-expander $G=(V,E)$, and any two vertex subsets
$A,B\subseteq V$, compute a large collection of paths between vertices of $A$ and vertices of $B$, that cause congestion $\Otil(1/\poly(\phi))$, such that every path has length $\Otil(1/\poly(\phi))$. We will use such an algorithm in order to implement the matching player.
The algorithm is summarized in the following theorem, that uses the approach of \cite{fast-vertex-sparsest,ChuzhoyGLNPS19}. 
We include the proof for completeness.
\begin{thm}
	\label{thm:term matching}There
	is a deterministic algorithm, that we call $\termMatching(G,A,B,\phi)$, that receives as input
	a parameter $\phi\in(0,1)$, a $\phi$-expander
	$G=(V,E)$ with $m$ edges, and two subsets $A,B\subseteq V$ of vertices of $G$ called terminals,
	with $|A|\leq |B|$. The algorithm returns
	a matching $M$ between vertices of $A$ and vertices of $B$ of cardinality $|A|$, and a set
	$\pset$ of paths of length at most $O(\log n/\phi)$ each, embedding the matching $M$ into $G$ with edge-congestion $O(\log^2n/\phi^2)$. The running time of the algorithm is $\Otil(m/\phi^3)$.
\end{thm}

The remainder of this subsection is dedicated to proving \Cref{thm:term matching}. The proof is almost identical to that in \cite{ChuzhoyGLNPS19}.
The key subroutine that is used in the proof is the following lemma.
\begin{lem}
	\label{lem: matching player 1 it} There is a deterministic algorithm,
	that, given an $m$-edge graph $G=(V,E)$, two disjoint subsets $A',B'$
	of its vertices with $|A'|\leq|B'|$, and an integer $\ell\geq 32\log m$,
	computes one of the following:
	\begin{itemize}
		\item either a matching $M'$ between vertices of $A'$ and vertices of $B'$ with $|M'|\ge|A'|\cdot\frac{8\log m}{\ell^{2}}$, together with a collection $\pset'$ paths of length at most $\ell$ each that embed $M'$ into $G$ with edge-congestion $1$; or
		\item a cut $(X,Y)$ in $G$ with $\Phi_{G}(X,Y)\leq 24\log m/\ell$. 
	\end{itemize}
	The running time of the algorithm is $\tilde{O}(\ell |E(G)|)$.
\end{lem}

\begin{proof}
We can assume w.l.o.g. that the graph $G$ is connected, as otherwise we can compute a cut $(X,Y)$ with conductance zero in time  $O(m)$.
Next, we create an auxiliary graph $G_{st}$ as follows. We start with graph $G$, and then add a source vertex
$s$ that connects with an edge to every vertex of $A'$, and a sink vertex $t$, that connects with an edge to every vertex of $B'$. We then initialize an $\EST$ data structure on graph $G_{s,t}$, with source vertex $s$, and distance threshold $\ell+2$. We denote by $\T$ the shortest-path tree rooted at $s$ that the data structure maintains. We also initialize $\pset'=\emptyset$ and $M'=\emptyset$.

The algorithm performs iterations, as long as 
 $\dist_{G_{st}}(s,t)\le\ell+2$ and $|\pset'|<\frac{8|A'|\log m}{\ell^{2}}$ hold.

In order to execute an iteration,
let $P_{st}$ be the shortest $s$-$t$ path in $\T$. Observe that path $P'=P_{st}\setminus\{s,t\}$ is a simple path in graph $G$, of length at most $\ell$, connecting some vertex $a'\in A'$ to some vertex $b'\in B'$. We delete the edges of $P_{st}$ from $G_{st}$
and update the $\EST$ data structure accordingly. Also, we add the path $P_{st}\setminus\{s,t\}$
to set $\pset'$ and set $A'\gets A'\setminus\{$$a'\}$ and $B'\gets B'\setminus\{$$b'\}$. Lastly, we add $(a',b')$ to $M$, and continue to the next iteration. 

Notice that the total running time of the algorithm so far is $O(m\ell)$, by the guarantees of the $\EST$ data structure.

We now consider two cases. First, if, when the above algorithm terminates, $|\pset'|=\frac{8|A'|\log m}{\ell^{2}}$ holds, then we return the matching $M'$ and the paths set $\pset'$. Clearly,  $|M'|\ge|A'|\cdot\frac{8\log m}{\ell^{2}}$ holds, the paths in $\pset'$ are edge-disjoint (so they cause edge-congestion $1$), and the length of every path is at most $\ell$.

Therefore, we assume from now on that, when the above algorithm terminates, $|\pset'|<\frac{8|A'|\log m}{\ell^{2}}$ holds, and so, from the algorithm description, $\dist_{G_{st}}(s,t)>\ell+2$.
In particular, $\dist_{G_{st}}(A',B')>\ell$ now holds, where $\dist_{G_{st}}(A',B')=\min_{a'\in A',b'\in B'}\dist_{G_{st}}(a',b')$.
We use the following standard claim:

\begin{claim}
	\label{claim: cut through ball growing} There is a deterministic algorithm, that, given an $m$-edge graph $H$
	with two sets $S,T$ of its vertices, such that $\dist_{H}(S,T)>\ell$ for some parameter $\ell$, computes a vertex set $Z$ with $\Phi_{H}(Z)<\frac{8\log m}{\ell}$,
and	$\vol_{H}(Z)\le\vol(H)/2$, such that either $S\subseteq Z$ or $T\subseteq Z$ hold. The running time of the algorithm is $O(m)$.
\end{claim}

\begin{proof}
	For any vertex set $X\subseteq V(H)$ and a parameter $d$, let $\ball_{H}(X,d)=\{u\mid\dist_{H}(X,u)\le d\}$.
	Note that we are guaranteed that $\ball_{H}(S,\ell/3)\cap\ball_{H}(T,\ell/3)=\emptyset$.
	Therefore, we can assume w.l.o.g. that $\vol(\ball_{H}(S,\ell/3))\le\vol(H)/2$. We claim
	that there must be an index $0\le i\le\ell/3$ such that $\delta(\ball_{H}(S,i))\le\frac{8\log m}{\ell}\cdot\vol(\ball_{H}(S,i))$.
	Indeed, assume otherwise. Then $\vol(\ball_{H}(S,\ell/3))\ge(1+\frac{8\log m}{\ell})^{\ell/3}>\vol(H)/2$ which is a contradiction.
We can compute the index $i$ and the vertex set $Z=\ball_{H}(S,i)$  by performing breadth-first search from vertices of $S$ and vertices of $T$, in time $O(m)$. It is now easy to verify that $\Phi_{H}(Z)<\frac{8\log m}{\ell}$,
$\vol_{H}(Z)\le\vol(H)/2$, and either $S\subseteq Z$ or $T\subseteq Z$ hold. 
\end{proof}

We are now ready to complete the proof of \Cref{lem: matching player 1 it}. 
Let $H=G_{st}\setminus\{s,t\}$.
We invoke \Cref{claim: cut through ball growing} on graph $H$, with the sets $A'$ and $B'$ of vertices, and
obtain a cut $Z$. We claim that $\Phi_{G}(Z)<\frac{24\log m}{\ell}$.
Indeed, let $E_1$ denote the set of all edges lying on the paths in $\pset'$, and let $E_2$ denote the set of all edges in $\delta_G(Z)\setminus E_1$. From the guarantees of \Cref{lem: matching player 1 it}, $|E_2|\leq \frac{8\log m}{\ell}\cdot \vol_H(Z)\leq \frac{8\log m}{\ell}\cdot\vol_G(Z)$.
Let $k$ denote the cardinality of the set $A'$ at the beginning of the algorithm. Since $|\pset'|<\frac{8k\log m}{\ell^{2}}<\frac k 2$ (as we have assumed that $\ell\geq 32\log m$), we get that $|A'|\geq k/2$, and in particular, $\vol_G(Z)\geq |A'|\geq k/2$.

In order to complete the proof that $\Phi_{G}(Z)<\frac{24\log m}{\ell}$, it is now enough to show that $|E_1|\leq  \frac{16\log m}{\ell}\cdot\vol_G(Z)$.
Indeed, recall that the length of every path in $\pset'$ is bounded by $\ell$, and $|\pset'|<\frac{8k\log m}{\ell^{2}}$. Therefore:

\[|E_1|\leq \ell\cdot |\pset'| <\frac{8k\log m}{\ell}. \]

Since we have established that $\vol_G(Z)\geq k/2$, we get that $|E_1|<\frac{16\log m}{\ell}\vol_G(Z)$. We conclude that $E_G(Z)\leq \frac{24\log m}{\ell}\vol_G(Z)$, as required.

As the total running time of the algorithm is bounded by $O(m\ell)$, this concludes the
proof of \Cref{lem: matching player 1 it}. 
\end{proof}

We obtain the following corollary of \Cref{lem: matching player 1 it}.

\begin{cor}
	\label{thm: matching player} There is a deterministic algorithm,
	that, given an $m$-edge graph $G=(V,E)$, two disjoint subsets $A,B$
	of its vertices with $|A|\leq|B|$, and a parameter $\ell\geq32\log m$, computes
	one of the following:
	\begin{itemize}
		\item either a matching $M$ between vertices of $A$ and vertices of $B$ with $|M|=|A|$, and a collection
		$\pset$ of paths of length at most $\ell$ each, that embeds $M$ into $G$ with congestion at most $\ell^{2}$; or
		\item a cut $(X,Y)$ in $G$ where $\Phi_{G}(X,Y)\leq24\log m/\ell$. 
	\end{itemize}
	The running time of the algorithm is $\tilde{O}(\ell^{3}m)$. 
\end{cor}

\begin{proof}
We start with $M=\emptyset$ and $\pset=\emptyset$, and then iterate. In every iteration, we let $A'\subseteq A$ and $B'\subseteq B$
	be the subsets of vertices that do not participate in the matching
	$M$; since $|A|\leq|B|$, we are guaranteed that $|A'|\leq|B'|$.
	If $A'=\emptyset$, then we terminate the algorithm, and return the
	current matchings $M$ and its embedding $\pset$ that we have computed. Otherwise,
	we apply the algorithm from \Cref{lem: matching player 1 it} to graph $G$ and vertex
	sets $A',B'$. If the outcome is a cut $(X,Y)$ in $G$ with $\Phi_{G}(X,Y)\leq24\log m/\ell$,
	then we terminate the algorithm, and return the cut $(X,Y)$. Therefore,
	we assume from now on that, whenever \Cref{lem: matching player 1 it}
	is called, it returns a matching $M'$ between $A'$ and $B'$ with
	$|M'|\ge\frac{8|A'|\log m}{\ell^{2}}$, and its corresponding embedding
	$\pset'$ with congestion $1$ and length $\ell$. We then add the
	paths in $\pset'$ to $\pset$, and we add the matching $M'$ to $M$,
	and continue to the next iteration. 
	
	As $|M'|\ge\frac{8|A'|\log m}{\ell^{2}}$ in every iteration, after
	$\ell^{2}$ iterations, we must have $A'=\emptyset$, and the algorithm
	will terminate. Notice that the congestion of the final path set $\pset$ is
	bounded by the number of iterations, $\ell^{2}$. Moreover, since
	the running time of each iteration is $\tilde{O}(\ell m)$, the total
	running time of the algorithm is $\tilde{O}(\ell^{3}m)$. 
\end{proof}

We are now ready to complete the proof of \Cref{thm:term matching}.
	We set $\ell=32\log m/\phi$, and then run the algorithm from \Cref{thm: matching player} on graph $G$, with vertex sets $A$ and $B$. As graph $G$ is
	a $\phi$-expander, the algorithm for \Cref{thm:term matching} may never return a cut $(X,Y)$ with $\Phi_{G}(X,Y)\leq24\log m/\ell <\phi$. Therefore, the algorithm must return a 
	matching $M$ between the vertices of $A$ and the vertices of $B$ of cardinality $|A|$, together with its embedding $\pset$, whose congestion is $\ell^{2}=O(\log^2 m/\phi^{2})$, such that the length of every path in $\pset$ is bounded by $\ell=O(\log m/\phi)$. The total running time of the algorithm is  $\Otil(\ell^{3}m)=\Otil(m/\phi^{3})$.

\subsubsection{Completing the Proof of \Cref{lem:embed}.}

We are now ready to complete the proof of \Cref{lem:embed}. 
We run the cut-matching game on a graph $W$ whose vertex set is the set $T$ of terminals, and the edge set is initially empty. In every round $i$ of the game, we use the algorithm from  \Cref{thm:CMG} to compute a partition $(A_{i},B_{i})$ of $T$ with $|A_{i}|\le|B_{i}|$, that we treat as the move of the cut player. Then we apply Algorithm \termMatching\xspace from \Cref{thm:term matching} to graph $G$, and the sets $A_i$ and $B_i$ of vertices. We denote by $M_i$ the matching returned by the algorithm, and by $\pset_i$ its embedding. We then add the edges of $M_i$ to graph $W$ and continue to the next iteration. From \Cref{thm:CMG}, after $O(\log |T|)$ iterations, graph $W$ is guaranteed to be a $(1/\gamma(|T|))$-expander. Since the edge set $E(W)$ is a union of $O(\log |T|)$ matchings, every vertex of $W$ has degree at most $O(\log |T|)$. We also compute an embedding $\pset=\bigcup_i\pset_i$ of $W$ into $G$, where every path in $\pset$ has length $O(\log(n)/\phi)$. Moreover, since each path set $\pset_i$ causes edge congestion at most $O(\log^2(n)/\phi^2)$, the congestion of the embedding $\pset$ is at most $O(\log^3(n)/\phi^2)$.
Lastly, it remains to bound the running time of the algorithm. Recall that the algorithm consists of $O(\log n)$ rounds. 
In every round we apply the algorithm from \Cref{thm:CMG}, whose running time is $O(n\gamma(n))$, and Algorithm \termMatching from \Cref{thm:term matching}, whose running time is $\Otil(m/\phi^3)$. Therefore, the total running time of the algorithm is bounded by
$\Otil\left (m/\phi^3+|T|\gamma(|T|)\right )=\tilde{O}\left (m\gamma(|T|)/\phi^3\right )$. This concludes the proof of \Cref{lem:embed}.

%% file: mincost.tex
\section{Application: Maximum Bounded-Cost Flow}
\label{sec:mincost}

In this section, we provide an algorithm for the Maximum Bounded-Cost Flow problem, as the main
application of our algorithm for decremental $\SSSP$  from \Cref{thm: main for SSSP}.
The technique is a standard application of the multiplicative weight
update framework \cite{GK98,Fleischer00,AroraHK12}. We provide the proofs
for completeness. In \cite{fast-vertex-sparsest}, the same technique
was used to provide algorithms for Maximum $s$-$t$ Flow in vertex-capacitated graphs. 
We note that the Maximum Bounded-Cost Flow  problem is somewhat more general, and it is a useful subroutine for a large number of applications; we discuss some of these applications in \Cref{sec:more_app}. We start with some basic definitions.

\paragraph*{Definitions.}
 Given a directed graph $G=(V,E)$ and a pair $s,t\in V$ of its vertices, an $s$-$t$ flow is a function $f\in\mathbb{R}_{\ge0}^{E}$, such that, for every vertex $v\in V-\{s,t\}$, the amount of flow entering $v$
equals the amount of flow leaving $v$, that is, $\sum_{(u,v)\in E}f(u,v)=\sum_{(v,u)\in E}f(v,u)$.
Let $f(v)=\sum_{(u,v)\in E}f(u,v)$ be the \emph{amount of flow at
$v$}. The \emph{value} of the flow $f$ is $\sum_{(s,v)\in E}f(s,v)-\sum_{(v,s)\in E}f(v,s)$.
Assume further that we are given capacities $c\in(\mathbb{R}_{>0}\cup\{\infty\})^{E}$ and costs $b\in\mathbb{R}_{\ge0}^{E}$ on edges. A flow $f$ is \emph{edge-capacity-feasible} if $f(e)\le c(e)$ for all $e\in E$.
The \emph{cost} of the flow $f$ is $\sum_{e\in E}b(e)f(e)$. If we are given a
cost bound $\overline{b}\ge0$, then we say that $f$ is \emph{edge-cost-feasible}
iff $\sum_{e\in E}b(e)f(e)\le\overline{b}$. We can define capacities
and costs on the vertices of $G$ similarly. Let $c\in(\mathbb{R}_{>0}\cup\{\infty\})^{V}$
be vertex capacities and let $b\in\mathbb{R}_{\ge0}^{V}$ be vertex costs.
As before, we say that $f$ is \emph{vertex-capacity-feasible} if $f(v)\le c(v)$
for all $v\in V-\{s,t\}$, and it is  \emph{vertex-cost-feasible} if $\sum_{v\in V-\{s,t\}}b(v)f(v)\le\overline{b}$.
We may just write \emph{capacity-feasible }and \emph{cost-feasible }when
clear from context. If $G$ is undirected, one way to define an $s$-$t$
flow is by treating $G$ as a directed graph, where we replace each undirected edge $\set{u,v}$ with a pair $(u,v)$ and $(v,u)$ of bi-directed edges. We
will assume w.l.o.g. that the flow only traverses the edges in one direction,
that is, for each edge $\{u,v\}\in E$, either $f(u,v)=0$ or $f(v,u)=0$. 

Next, we define the \emph{Maximum Bounded-Cost Flow problem} ($\MBCF$). In the edge-capacitated version, we are given a graph $G=(V,E)$ with edge capacities
$c\in(\mathbb{R}_{>0}\cup\{\infty\})^{E}$ and edge costs $b\in\mathbb{R}_{\ge0}^{E}$,
together with two special vertices $s$ and $t$, and a cost bound
$\overline{b}$. The goal is to compute an $s$-$t$ flow $f$ of
maximum value, such that $f$ is both capacity-feasible and cost-feasible. The
vertex-capacitated version is defined similarly except that we are given vertex capacities $c\in(\mathbb{R}_{>0}\cup\{\infty\})^{V}$ and
vertex cost $b\in\mathbb{R}_{\ge0}^{V}$ instead. A $(1+\epsilon)$-approximate
solution for this problem is a flow $f$ which is both capacity-feasible
and cost-feasible, such that the value of the flow is at least $\opt(\overline{b})/(1+\epsilon)$
where $\opt(\overline{b})$ is the maximum value of a capacity-feasible
flow of cost at most $\overline{b}$.

\paragraph*{Connection to the Min-Cost Flow Problem.}
The classical Min-Cost Flow problem is defined exactly like \MBCF, except that, instead of the
cost bound, we are given a target flow value $\tau$. The goal is
to either (i) compute an $s$-$t$ flow $f$ of value at least $\tau$, such that $f$ is capacity-feasible and has the smallest cost among all flows satisfying these requirements, or (ii) to certify that there
is no capacity-feasible flow of value at least $\tau$. Let $\optcost(\tau)$
be the minimum cost of any capacity-feasible flow of value at least
$\tau$. Observe that an exact algorithm for \MBCF implies an exact algorithm for the min-cost flow problem and vice versa, via binary search.
Moreover, a $(1+\epsilon)$-approximation algorithm for \MBCF gives a $(1+\eps)$-factor pseudo-approximation for the Min-Cost Flow problem, that is: we either find a flow of cost at most $\optcost(\tau)$
and value at least $\tau/(1+\epsilon)$, or certify that there is
no capacity-feasible flow of value at least $\tau$. Note that, if 
we insist that the value of the flow that we obtain in the min-cost flow is at least $\tau$, then the problem is at least as difficult as the {\bf exact} maximum $s$-$t$ flow.
From now on we focus on the \MBCF problem.

\paragraph*{Our results.}
We show approximation algorithms for the \MBCF problem in undirected graphs for both edge-capacitated and vertex-capacitated settings, though in the former scenario we only consider unit edge-capacities. 

\begin{theorem}
[Unit-edge capacities]\label{thm:mincost edge}There is a deterministic
algorithm that, given a simple undirected $n$-vertex $m$-edge
graph $G=(V,E)$ with unit edge capacities $c(e)=1$
and edge costs $b(e)>0$ for $e\in E$, together with a source $s$, a sink $t$,
a cost bound $\overline{b}$, and an accuracy parameter $0<\epsilon<0.1$,
computes a $(1+\epsilon)$-approximate solution for \MBCF in time $\Ohat\left (n^{2}\cdot \frac{\log B}{\epsilon^{O(1)}}\right )$, where $B$ is the ratio of largest to smallest edge cost.
\end{theorem}

Previously, Cohen et al. \cite{CohenMSV17} gave an exact algorithm for \MBCF
with running time $\tilde{O}(m^{10/7}\cdot \log B)$ when the input graph
$G$ has unit edge-capacities as well, but $G$ can be directed and
is not necessarily simple. Lee and Sidford \cite{LeeS14} showed an exact algorithm
with running time $\tilde{O}(m\sqrt{n}\cdot \log^{O(1)}B)$ on directed
graphs with general edge-capacities. To the best of our knowledge,
 no faster algorithms for \MBCF are currently known, even when approximation is allowed. Our algorithm provides a $(1+\eps)$-approximate  solution, and
only works in simple, undirected graphs with unit edge-capacities. It is faster than these previously known algorithms when $m=\omega(n^{1.5+o(1)})$, 
and it implies a number of applications, as shown in \Cref{sec:more_app}.

We also show a deterministic algorithm  for graphs with (arbitrary) vertex capacities
and costs.

\begin{theorem}
[Vertex capacities]\label{thm:mincost vertex}There is a deterministic
algorithm that, given an undirected $n$-vertex graph $G=(V,E)$ with
vertex capacities $c(v)>0$ and vertex costs $b(v)>0$ for all $v\in V$,
a source $s$, a sink $t$, a cost bound $\overline{b}$, and an accuracy
parameter $0<\epsilon<0.1$, computes a $(1+\epsilon)$-approximate
\MBCF in time $\Ohat\left (n^{2}\cdot \frac{\log (BC)}{\epsilon^{O(1)}}\right )$, where $B$ is the ratio of largest to smallest vertex cost,
and $C$ is the ratio largest to smallest vertex capacity.
\end{theorem}

We note that a randomized algorithm with similar guarantees can be obtained from the algorithm of
 Chuzhoy and Khanna \cite{fast-vertex-sparsest} for $\SSSP$, though this was not explicitly noted in their paper (they only explicitly provide an algorithm for approximate max flow). 
We obtain a deterministic algorithm by using our deterministic algorithm for \SSSP instead of the randomized algorithm of \cite{fast-vertex-sparsest}.
The best previous algorithm for vertex-capacitated \MBCF, with running time $\tilde{O}(m\sqrt{n}\log^{O(1)}(BC))$,  follows from the work of Lee and Sidford
\cite{LeeS14}; the algorithm solves the problem exactly.
Our algorithm has a faster running time when
$m=\omega(n^{1.5+o(1)})$. 

The remainder of this section is dedicated to proving \Cref{thm:mincost edge} and
\Cref{thm:mincost vertex}. We start by describing, in \Cref{sec:mincost_MWU}, 
an algorithm for \MBCF in general edge-capacitated graphs, based on the multiplicative
weight update (MWU) framework, and we bound the number of ``augmentations''
in the algorithm. Then, in \Cref{sec:mincost_dynamic},
we show how to perform the ``augmentations'' efficiently when the
input graph is as in \Cref{thm:mincost edge} and \Cref{thm:mincost vertex},
using our algorithm for decremental $\SSSP$. We will use the following observation:

\begin{remark}
\label{rem:zero weight} It is easy to extend \Cref{thm: main for SSSP}
so that the algorithm can handle edges of length $0$, if we have
a promise that in every query  $\distquery(s,v)$ or $\pquery(s,v)$, the distance from the source $s$
to the query vertex $v$ is non-zero. To do this, let $\ell_{\min}$ and $\ell_{\max}$
be the minimum and the maximum \emph{non-zero }edge lengths in the graph
respectively. For each edge of length $0$, we set the length to be
$\epsilon\ell_{\min}/n$ instead. This will not increase the length
of any non-zero length path by more than a factor $(1+\eps)$. Let $L'=\frac{\ell_{\max}}{\ell_{\min}}$
be the ratio between the original largest to smallest \emph{non-zero} length. We can now use \Cref{thm: main for SSSP} with the new bound $L=L'n/\eps$.
\end{remark}

%% file: mincost_MWU.tex
\subsection{A Multiplicative Weight Update-Based Flow Algorithm\label{sec:mincost_MWU}}

We describe an algorithm for computing approximate \MBCF in edge-capacitated graphs in \Cref{alg:fleischer}. The algorithm is based on the multiplicative weight
update (MWU) framework, and it is a straightforward adaptation of the algorithms
of Garg and K\"{o}nemann \cite{GK98}, Fleischer \cite{Fleischer00},
and Madry \cite{Madry10_stoc}.

\Cref{alg:fleischer} is stated for both undirected and directed graphs.
Let $G=(V,E)$ be the input graph and let $\P_{s,t}$ be the set of
all $s$-$t$ paths. If $G$ is (un)directed, then $\P_{s,t}$ contains
all (un)directed $s$-$t$ paths. We always augment a flow along some
$s$-$t$ path $P\in\P_{s,t}$. Let $f(P)$ denote the amount of flow
through path $P$. We use the shorthand $f(P)\gets f(P)+c$ to indicate that we increase the flow value $f(e)$ for all $e\in E(P)$ by $c$, that is,  $f(e)\gets f(e)+c$. The algorithm maintains lengths $\ell\in\mathbb{R}_{\ge0}^{E}$
on edges and a parameter $\phi\ge0$. For any path $P$, let $\ell(P)=\sum_{e\in P}\ell(e)$
be the length of $P$ and similarly let $b(P)=\sum_{e\in P}b(e)$ be the
cost of $P$. In general, for any function $d\in\mathbb{R}_{\ge0}^{E}$, we
let $d(P)=\sum_{e\in P}d(e)$. We use the shorthand $d=\ell+\phi b$ to indicate a new edge-length function $d(e)=\ell(e)+\phi b(e)$ for all $e\in E$. A $d$-shortest
$s$-$t$ path is an $s$-$t$ path $P^{*}$ that minimizes $d(P^*)$
among all paths in $\P_{s,t}$. An $\alpha$-approximate $d$-shortest path
$\tilde{P}$ is a path $P\in \P_{s,t}$ with $d(\tilde{P})\le\alpha\cdot d(P^{*})$.

\begin{algorithm}
\textbf{Input: }An undirected or a directed graph $G=(V,E)$ with edge
capacities $c\in(\mathbb{R}_{>0}\cup\{\infty\})^{E}$ and edge costs
$b\in\mathbb{R}_{\ge0}^{E}$, a source $s$, a sink $t$, a cost bound
$\overline{b}$, and an accuracy parameter $0<\epsilon<1$.

\textbf{Output: }An $s$-$t$ flow which is capacity-feasible and
cost-feasible.
\begin{enumerate}
\item Set $\delta=(2m)^{-1/\epsilon}$
\item Set $\ell(e)=\delta/c(e)$ if $c(e)$ is finite; otherwise $\ell(e)=0$.
Set $\phi=\delta/\overline{b}$. Set $f\equiv0$.
\item \textbf{while} $\sum_{e\in E}c(e)\ell(e)+\overline{b}\phi<1$ \textbf{do}
\begin{enumerate}
\item $P\gets$ a $(1+\epsilon)$-approximate $(\ell+\phi b)$-shortest
$s$-$t$ path
\item $c\gets\min\{\min_{e\in P}c(e),\overline{b}/b(P)\}$
\item $f(P)\gets f(P)+c$
\item for every edge $e\in E(P)$, set  $\ell(e)\gets\ell(e)(1+\frac{\epsilon c}{c(e)})$ 
\item $\phi\gets\phi(1+\frac{\epsilon c\cdot b(P)}{\overline{b}})$
\end{enumerate}
\item \textbf{return} the scaled down flow $f/\log_{1+\epsilon}(\frac{1+\epsilon}{\delta})$
\end{enumerate}
\caption{An approximate algorithm for max bounded-cost $s$-$t$ flow in edge-capacitated
graphs\label{alg:fleischer}}
\end{algorithm}

\begin{lemma}
The flow $f/\log_{1+\epsilon}(\frac{1+\epsilon}{\delta})$ computed by \Cref{alg:fleischer}
is capacity-feasible and cost-feasible.\label{lem:feasible}
\end{lemma}

\begin{proof}
When the flow on an edge $e$ is increased by an additive amount of $a\cdot c(e)$,
where $0\le a\le1$, then $\ell(e)$ is multiplicatively increased
by factor $(1+a\epsilon)\ge(1+\epsilon)^{a}$. As $\ell(e)=\delta/c(e)$
initially and $\ell(e)<(1+\epsilon)/c(e)$ at the end, it grows by the multiplicative factor at most $(1+\eps)/\delta=(1+\eps)^{\log_{1+\eps}((1+\eps)/\delta)}$ over the course of the algorithm. Therefore, the flow on $e$ is at
most $c(e)\cdot\log_{1+\epsilon}(\frac{1+\epsilon}{\delta})$ before scaling
down\footnote{If $G$ is undirected, it can be the case that $f(e)$ is even decreased
while $\ell(e)$ is increased. This gives even more slack to the analysis.}, and so $f/\log_{1+\epsilon}(\frac{1+\epsilon}{\delta})$ is capacity-feasible.
Similarly, every time the cost of the flow increases by additive amount $a\overline{b}$, where
$0\le a\le1$, the value of $\phi$ is multiplicatively increased
by factor $(1+a\epsilon)\ge(1+\epsilon)^{a}$. As $\phi=\delta/\overline{b}$
initially and $\phi<(1+\epsilon)/\overline{b}$ at the end, the value of $\phi$ grows by at most factor $(1+\eps)/\delta$ over the course of the algorithm. Therefore, the cost of the final flow $f$ before the scaling down
is at most $\overline{b}\cdot\log_{1+\epsilon}(\frac{1+\epsilon}{\delta})$, and so flow $f/\log_{1+\epsilon}(\frac{1+\epsilon}{\delta})$ is cost-feasible.
\end{proof}

Next, we bound the number of augmentation in \Cref{alg:fleischer},
that is, the number of times that the ``while'' loop is executed.
\begin{lemma}
\label{lem:aug bound}
 \ 
\begin{enumerate}
\item If graph $G$ has unit edge capacities, and there is an $s$-$t$ cut of capacity
$k$, then there are at most $\tilde{O}(k/\epsilon^{2})$ augmentations.
In particular, if $G$ is a simple graph with unit edge-capacities, then there
are at most $\tilde{O}(n/\epsilon^{2})$ augmentations.
\item If graph $G$ is has at most $k$ edges with finite capacity, then are at
most $\tilde{O}(k/\epsilon^{2})$ augmentations. 
\end{enumerate}
\end{lemma}

\begin{proof}
(1) By assumption, there is an $s$-$t$ cut $(S,\overline{S})$ with $|E(S,\overline{S})|=k$.
In each augmentation, either $\phi$ is increased by factor $(1+\epsilon)$
or there is some edge $e\in E(S,\overline{S})$, for which $\ell(e)$ is
increased by factor $(1+\epsilon)$. Again, we have $\ell(e)=\delta/c(e)$
initially and $\ell(e)<(1+\epsilon)/c(e)$ at the end. Also, $\phi=\delta/\overline{b}$
initially and $\phi<(1+\epsilon)/\overline{b}$ at the end. So there can be at
most $(k+1)\log_{1+\epsilon}(\frac{1+\epsilon}{\delta})=\tilde{O}(k/\epsilon^{2})$
augmentations. 

(2) Let $E'$ be the set of edges with finite capacity. In each augmentation,
either $\phi$ is increased by factor $(1+\epsilon)$ or there is
some edge $e\in E'$, for which $\ell(e)$ is increased by factor $(1+\epsilon)$.
As before, we start with $\ell(e)=\delta/c(e)$, and at the end, $\ell(e)<(1+\epsilon)/c(e)$ holds. Similarly, $\phi=\delta/\overline{b}$
initially and $\phi<(1+\epsilon)/\overline{b}$ holds at the end.
Since the total number of edges with finite capacity is at most $k$, the total number of augmentations is bounded by  $(k+1)\log_{1+\epsilon}(\frac{1+\epsilon}{\delta})=\tilde{O}(k/\epsilon^{2})$. 
\end{proof}
\begin{lemma}
	Flow
$f/\log_{1+\epsilon}(\frac{1+\epsilon}{\delta})$ is a $(1+O(\epsilon))$-approximate
solution to the \MBCF problem.
\end{lemma}

\begin{proof}
For each edge $e$, let $\P_{e}\subseteq \pset_{s,t}$ be the set of all paths containing $e$. We first write
the standard LP relaxation for \MBCF and its dual LP (we can use the same relaxation for both undirected and directed graphs, except that the set $\P_{s,t}$ of paths is defined differently)

\begin{center}
\begin{tabular}{|c|c|}
\hline 
$\begin{array}{crccc}
(\boldsymbol{P}_{LP})\\
\max & \sum_{P\in\P_{st}}f(P)\\
\text{s.t.} & \sum_{P\in\P_{e}}f(P) & \le & c(e) & \forall e\in E\\
 & \sum_{P\in\P_{s,t}}b(P)\cdot f(P) & \le & \overline{b}\\
 & x & \ge & 0
\end{array}$ & $\begin{array}{crccc}
(\boldsymbol{D}_{LP})\\
\min & \sum_{e\in E}c(e)\ell(e)+\overline{b}\phi\\
\text{s.t.} & \ell(P)+\phi b(P) & \ge & 1 & \forall P\in\P_{s,t}\\
 & \ell,\phi & \ge & 0
\end{array}$\tabularnewline
\hline 
\end{tabular}
\par\end{center}

Denote $D(\ell,\phi)=\sum_{e\in E}c(e)\ell(e)+\overline{b}\phi$, and let
$\alpha(\ell,\phi)$ be the length of the $(\ell+\phi b)$-shortest
$s$-$t$ path. Let $\ell_{i}$ be the edge-length function $\ell$ after the
$i$-th execution of the while loop, and let $\phi_{i}$ be defined similarly for $\phi$.
We denote $D(i)=D(\ell_{i},\phi_{i})$ and $\alpha(i)=\alpha(\ell_{i},\phi_{i})$
for convenience. We also denote by $P_{i}$ the path found in the $i$-th iteration and by $c_{i}$
 the amount by which the flow on $P_i$ is augmented.
Observe that:

\begin{align*}
D(i) & =\sum_{e\in E}c(e)\ell_{i-1}(e)+\overline{b}\phi_{i-1}+\sum_{e\in P_{i}}c(e)\cdot\left(\frac{\epsilon c_{i}}{c(e)}\cdot\ell_{i-1}(e)\right )+\overline{b}\phi_{i-1}\cdot\frac{\epsilon c_{i}\cdot b(P_{i})}{\overline{b}}\\
 & =D(i-1)+\epsilon c_{i}(\ell_{i-1}(P_{i})+\phi_{i-1}b(P_{i})).
\end{align*}

Since $P_i$ is a $(1+\eps)$-approximate shortest path with respect to $\alpha(i-1)$, we get that:

\begin{align*}
D(i) & \le D(i-1)+\epsilon(1+\epsilon)c_{i}\alpha(i-1).
\end{align*}
Let $\beta=\min_{\ell,\phi}D(\ell,\phi)/\alpha(\ell,\phi)$ be the
optimal value of the dual LP $\boldsymbol{D}_{LP}$. Then:
\begin{align*}
D(i) & \le D(i-1)+\epsilon(1+\epsilon)c_{i}D(i-1)/\beta\\
 & \le D(i-1)\cdot e^{\eps(1+\eps)c_{i}/\beta}.
\end{align*}

If $t$ is the index of the last iteration, then $D(t)\ge1$. Since $D(0)\le2\delta m$:

\[
1\le D(t)\le2\delta m\cdot e^{\epsilon(1+\epsilon)\sum_{i=1}^{t}c_{i}/\beta}.
\]
 Taking a $\ln$ from both sides, we get:
 
\begin{equation}\label{eq: calc}
0\le\ln(2\delta m)+\epsilon(1+\epsilon)\sum_{i=1}^{t}c_{i}/\beta.
\end{equation}

Let $F=\sum_{i=1}^{t}c_{i}$. Note that $F$ is exactly the total
amount of flow by which we augment over all iterations. Therefore, inequality~(\ref{eq: calc}) can be rewritten as:

\[\frac{F\eps(1+\eps)}{\beta}\geq \ln(1/(2\delta m)).\]

From \Cref{lem:feasible}, since the scaled-down flow is a feasible solution for  $\boldsymbol{P}_{LP}$, 
 $F/\log_{1+\epsilon}(\frac{1+\epsilon}{\delta})\le\beta$ must hold.
It remains to show that $F/\log_{1+\epsilon}(\frac{1+\epsilon}{\delta})\ge(1-O(\epsilon))\beta$:

\begin{align*}
\frac{F/\log_{1+\epsilon}(\frac{1+\epsilon}{\delta})}{\beta} & \ge\frac{\ln(1/(2\delta m))}{\epsilon(1+\epsilon)}\cdot\frac{1}{\log_{1+\epsilon}(\frac{1+\epsilon}{\delta})}\\
 & =\frac{\ln (1/\delta)-\ln(2m)}{\epsilon(1+\epsilon)}\cdot\frac{\ln(1+\epsilon)}{\ln(\frac{1+\epsilon}{\delta})}\\
 & \ge\frac{(1-\epsilon)\ln(1/\delta)}{\epsilon(1+\epsilon)}\cdot\frac{\ln(1+\epsilon)}{\ln(\frac{1+\epsilon}{\delta})}\\
 & \geq 1-O(\epsilon),
\end{align*}
where the penultimate inequality uses the fact that $\delta=(2m)^{-1/\epsilon}$, so $2m=(1/\delta)^{\eps}$, and $\ln(2m)=\eps\ln(1/\delta)$, and the last inequality holds because $\ln(1+\epsilon) \ge \epsilon - \epsilon^2/2$ 
and $\ln(\frac{1+\epsilon}{\delta}) \le (1+\eps)\ln(1/\delta)$.

\end{proof}

%% file: mincost_dynamic.tex
\subsection{Efficient Implementation Using Decremental $\protect\SSSP$\label{sec:mincost_dynamic}}

In this section, we complete the proofs of \Cref{thm:mincost edge}
and \Cref{thm:mincost vertex} by providing an efficient implementation of \Cref{alg:fleischer}
from \Cref{sec:mincost_MWU}. The algorithm exploits the algorithm for decremental $\SSSP$ from \Cref{thm: main for SSSP}, that we denote by $\aset$. A similar technique was used in~\cite{Madry10_stoc} and 
in \cite{fast-vertex-sparsest}. Our proofs are almost the same as the ones in \cite{fast-vertex-sparsest}, except that we need to take
care of the cost function $b$.

\subsubsection{Simple Graphs with Unit Edge Capacities}

We start with the proof of \Cref{thm:mincost edge}. Let $G=(V,E)$ be the input
undirected simple graph with $n$ nodes and $m$ edges. We assume that all edge capacities are unit. Let $b\in\mathbb{R}_{>0}^{E}$
be the edge costs, with $b_{\max}=\max_{e}b(e)$, $b_{\min}=\min_{e}b(e)$, and
$B=b_{\max}/b_{\min}$. Let $\overline{b}$ be the cost bound.
Let $\delta=(2m)^{-1/\epsilon}$ be the same as in \Cref{alg:fleischer}. For every edge $e\in E$, we let its weight be $w(e)=\ell(e)+\phi b(e)$.  
We run \Cref{alg:fleischer}, but we will employ the algorithm $\aset$ in order to compute $(1+\eps)$-approximate shortest $s$-$t$ paths in $G$, with respect to the edge weights $w(e)$. In order to do so, we construct another simple undirected graph $G'=(V',E')$ as follows. Let $K=\log_{(1+\epsilon/3)}\frac{1+\epsilon}{\delta}=O\left(\frac{\log m}{\epsilon^{2}}\right )$.
Recall that, at the beginning of the algorithm, for every edge $e\in E$, we set $\ell(e)=\delta$, and we set $\phi=\delta/\overline b$. Therefore, the initial weight of edge $e$ is $w(e)=\ell(e)+\phi b(e)=\delta(1+b(e)/\overline b)$. As long as the algorithm does not terminate, $\ell(e)<1$ and $\phi<1/\overline b$ must hold, so $w(e)<1+b(e)/\overline b$. Therefore, over the course of the algorithm, $w(e)$ may grow from $\delta(1+b(e)/\overline b)$ to at most $(1+b(e)/\overline b)$. The idea is to discretize all possible values that edge $w(e)$ may attain by powers of $(1+\eps/3)$.

We now define the new graph $G'=(V',E')$. For every vertex $v\in V$, we add $(K+1)$ vertices $v_{0},\dots,v_{K}$ to $V'$.  For
every edge $e=(u,v)\in E$, we add $(K+1)$ edges $e_{0},\dots,e_{K}$ to $E'$, where for each $0\leq i\leq K$,
$e_{i}=(u_{i},v_{i})$, and the weight $w'(e_{i})=\delta(1+b(e)/\overline{b})(1+\epsilon/3)^{i}$. Additionally, for each original vertex $v\in V$ and index $i\in\{1,\dots,K\}$,
we add an edge $(v_{0},v_{i})$ of weight $w'(v_{0},v_{i})=0$ to $E'$. Note that $|V'|=O(nK)=O\left(\frac{n\log m}{\epsilon^{2}}\right)$.

We run the algorithm $\aset$ from Theorem~\ref{thm: main for SSSP} on graph $G'$, where the length of each edge $e'$ is $w'(e')$, and the error parameter is $\eps/3$. Note that the ratio $L$ of largest to smallest non-zero edge length is $\frac{(1+b_{\max}/\overline{b})}{\delta(1+b_{\min}/\overline{b})}\le B/\delta$.  Note that some edges of $G'$ have length $0$. However, as we show later, we will never ask a query between a pair of vertices that lie within distance $0$ from each other, and so, using \Cref{rem:zero weight}, we can use algorithm $\aset$, except that we need to replace the $\log L$ factor in its running time by factor $\log(\frac{Ln}{\epsilon})=\log(\frac{Bn}{\eps\delta})=O(\log(Bn)/\eps^{O(1)})$. Therefore, the total update time of the algorithm $\aset$ is $\Ohat((|V'|^2 \log B)/\eps^{O(1)})=\Ohat((n^2\log B)/\eps^{O(1)})$.

Next, we need to describe the sequence of edge deletions in graph $G'$. The edges are deleted 
according to the following rule.  For every edge $e=(u,v)\in E$,
we delete an edge $e_{i}=(u_{i},v_{i})\in E'$ when $w'(e_{i})$ becomes smaller than $\ell(e)+\phi b(e)$. These are the only edge deletions in $G'$. 
Lastly, we show that, given  a $(1+\eps/3)$-approximate shortest $s_0$-$t_0$ path in $G'$ (with respect to edge lengths $w'(e')$), we can efficiently obtain a $(1+\eps)$-approximate shortest $s$-$t$ path in $G$ (with respect to edge lengths $w(e)$).

\begin{claim}\label{claim: getting approximate paths}
At any time before \Cref{alg:fleischer} terminates, given any $(1+\epsilon/3)$-approximate
 $w'$-shortest $s_{0}$-$t_{0}$ path $P'$ in $G'$, we can construct,
in time $O(|P'|)$, a $(1+\epsilon)$-approximate  $w$-shortest $s$-$t$
path $P$ in $G$.
\end{claim}

\begin{proof}
Since we assume that  \Cref{alg:fleischer} did not yet terminate, $\sum_{e\in E}\ell(e)+\overline{b}\phi<1$, and so for every edge
$e\in E$, $\delta(1+b(e)/\overline{b})\le \ell(e)+\phi b(e)\le1+b(e)/\overline{b}$. From the definition of the edge deletion sequence, if $i'$ is the smallest index for which the edge $e_{i'}$ lies in $E'$, then $\ell(e)+\phi b(e)\le w(e_{i'})<(\ell(e)+\phi b(e))(1+\epsilon/3)$.

Let $\dist$ denote the distance from $s$ to $t$ in $G$ with respect to edge lengths $w(e)$, and let $\dist'$ denote the distance from $s_{0}$ to $t_{0}$ with respect to edge lengths $w'(e')$. Then $\dist \le\dist'<\dist\cdot (1+\epsilon/3)$ must hold.

Assume now that we are given a $(1+\epsilon/3)$-approximate shortest $s_{0}$-$t_{0}$
in $G'$ with respect to edge lengths $w'(e')$. Then, by contracting every subpath $(v_{i},v_{0},v_{j})\subseteq P$
of length $0$ which corresponds to the same node $v$ in $G$, we
obtain an $s$-$t$ path $P$ in $G$ whose length is $w(P')\le(1+\epsilon/3)\dist'<(1+\epsilon/3)^{2}\dist \le(1+\epsilon)\dist$.
\end{proof}

Our algorithm only employs query $\pquery$ for the vertex $t_0$. In particular, it is easy to see that the distance from $s_0$ to $t_0$ is always non-zero. Therefore, we obtain a correct implementation of  \Cref{alg:fleischer}. We now analyze its running time. As already observed, the total running time needed to maintain the data structure from Theorem~\ref{thm: main for SSSP} is $\Ohat((n^2\log B)/\eps^{O(1)})$.
Observe that in every iteration of \Cref{alg:fleischer}, we employ a single call to $\pquery(t_0)$ in graph $G'$.  
Each such query takes $\Ohat(|P|\log\log(Ln/\epsilon)) = \Ohat(n\log\log(B/\epsilon))$ time to return a path $P$
and, by \Cref{lem:aug bound}, the number of queries is bounded by $\tilde{O}(n/\epsilon^{2})$. Therefore, the total time needed to respond to all queries is bounded by  $\Ohat \left((n^2\log B )/\eps^{O(1)}\right )$.
The running time of other steps for implementing \Cref{alg:fleischer}, such as maintaining
$\ell$ and $\phi$, are subsumed by these bounds. Altogether, the total running
time is $\Ohat(n^{2}\cdot \frac{\log B}{\epsilon^{O(1)}})$. 

\subsubsection{Vertex-Capacitated Graphs}

We now complete the proof \Cref{thm:mincost vertex}. Our proof is almost identical to that of \cite{fast-vertex-sparsest}. 
Let $G=(V,E)$ be the input undirected simple graph with $n$ nodes
and $m$ edges. Let $b\in\mathbb{R}_{>0}^{V}$ be the vertex costs,
with $b_{\max}=\max_{v}b(v)$, $b_{\min}=\min_{v}b(v)$, and $B=b_{\max}/b_{\min}$. Let $\overline{b}$ be the cost bound. Additionally, let $c\in\mathbb{R}_{>0}^{V}$ be the vertex capacities, with $c_{\max}=\max_{v}c(v)$, $c_{\min}=\min_{v}c(v)$,
and $C=c_{\max}/c_{\min}$.

We proceed as follows. First, we use a standard reduction from vertex-capacitated flow problems in undirected graphs to edge-capacitated flow problems in directed graphs, constructing a directed graph $G''$ with capacities on edges. We will run \Cref{alg:fleischer} on $G''$. In order to compute approximate $(\ell+\phi b)$-shortest $s$-$t$ paths in $G''$, we will employ the algorithm $\aset$ for decremental \SSSP from Theorem~\ref{thm: main for SSSP} in another graph $G'$ -- a simple undirected edge-weighted graph that we will construct. We now describe the construction of both graphs.

\paragraph*{Graph $G''$.}
We construct  a directed graph $G''=(V'',E'')$ with edge capacities
$c''(e)$ and edge costs $b''(e)$ for $e\in E''$, using a standard reduction from the input graph $G=(V,E)$. The set $V''$ of vertices contains, for every vertex $v\in V$ of the original graph, a pair $v_{in},v_{out}$ of vertices. Additionally, we add a directed edge $(v_{in},v_{out})$
of capacity $c''(v_{in},v_{out})=c(v)$ and cost $b''(v_{in},v_{out})=b(v)$ to $E''$.
For each undirected edge $(u,v)\in E$, we add a pair of new edges $(v_{out},u_{in}),(u_{out},v_{in})$ to $E''$,
both with capacity $\infty$ and cost $0$. This completes the definition of the graph $G''$, that we view as a flow network with source $s_{out}$ and destination $t_{in}$. Observe that for any $s$-$t$ flow
$f$ in $G$, there is a corresponding $s_{out}$-$t_{in}$ flow $f''$
in $G''$, of the same value and cost, such that $f$ is capacity-feasible
iff $f''$ is capacity-feasible. Similarly, any $s_{out}$-$t_{in}$ flow $f''$
in $G''$ can be transformed, in time $O(m)$, into an $s$-$t$ flow of the same value and the same cost in $G$, such that $f$ is capacity-feasible iff $f''$ is capacity-feasible.  We run \Cref{alg:fleischer}
on $G''$, and maintain edge lengths $\ell''\in\mathbb{R}_{\ge0}^{E''}$ and a value $\phi\ge0$. It now remains to show how we compute a $(1+\eps)$-approximate $(\ell''+\phi b'')$-shortest $s_{out}$-$t_{in}$ path in graph $G''$. In order to do so, we define a new graph $G'$, on which we will run the algorithm $\aset$ for decremental \SSSP.

As before, for every edge $e\in E''$, we maintain a weight $w''(e)=\ell''(e)+\phi b''(e)$. Recall that, at the beginning of the algorithm, we set $\ell''(e)=\delta / c''(e)$ if $c''(e)$ is finite, and we set $\ell''(e)=0$ otherwise. We also set $\phi=\delta/\overline b$. 
Therefore, initially, $w''(e)=\delta(1/c''(e)+b''(e)/\overline b)$ (if $c''(e)=\infty$, then $w''(e)=\delta b''(e)/\overline b=0$, and it remains $0$ throughout the algorithm).
As long as the algorithm does not terminate, $\ell''(e)< 1/c''(e)$ and $\phi<1/\overline b$ must hold, so $w''(e)<(1/c''(e)+b''(e)/\overline b)$. %
Therefore, over the course of the algorithm, $w''(e)$ may increase from $\delta(1/c''(e)+b''(e)/\overline b)$ to $(1/c''(e)+b''(e)/\overline b)$. %

\paragraph{Graph $G'$.} 
We construct an undirected simple graph $G'=(V',E')$, from the original input graph $G=(V,E)$. We first place weights on the vertices of $G'$, and later turn it into an edge-weighted graph. As before, we let $K=\log_{(1+\epsilon/3)}\frac{1+\epsilon}{\delta}=O\left (\frac{\log m}{\epsilon^{2}}\right )$. For every vertex $v\in V-\{s,t\}$, we add $K+1$ new vertices $v_{0},\dots,v_{K}$ to $V'$, and for all $0\leq i\leq K$, we set the weight $w(v_{i})=\delta\left (\frac{1}{c(v)}+\frac{b(v)}{\overline{b}}\right )\cdot (1+\epsilon/3)^{i}$.
For each edge $e=(u,v)\in E$ in the original graph, we  add $(K+1)^2$ new edges $e_{i,j}=(u_{i},v_{j})$
for all $i,j\in\{0,\dots,K\}$ to $E'$. We also add two new vertices $s$ and $t$ to $V'$, with weight
$0$. For each edge $e=(s,u)\in E$, for all $0\leq i\leq K$, we add an edge $e_{i}^{s}=(s,u_{i})$ to $E'$. Similarly, for each edge $e=(u,t)\in E$, for all $0\leq i\leq K$, we add an edge $e_{i}^{t}=(u_{i},t)$ to $E'$. Observe that $|V'|=O(nK)=O\left(\frac{n\log n}{\eps^2}\right )$.

We would like to run the algorithm $\aset$ for decremental 
$\SSSP$ on the graph $G'$. However, $G'$
has weights on vertices and not edges. This can be easily fixed as follows. For each edge
$e=(u,v)$ in $G'$, we let its weight be $w(e)=(w(u)+w(v))/2$. Since $w(s)=w(t)=0$,
for every $s$-$t$ path $P'$ in $G'$, the total weight of all edges on $P'$ equals to the total weight of all vertices on $P'$.
Note that all edge weights are non-zero.

We run the algorithm $\aset$ from Theorem~\ref{thm: main for SSSP} on graph $G'$, where the length of each edge $e$ is $w(e)$, and the error parameter is $\eps/3$.

Notice that the ratio $L$ of largest to smallest edge length in $G'$ is $L=\frac{1/c_{\min}+b_{\max}/\overline{b}}{\delta(1/c_{\max}+b_{\min}/\overline{b})}\le \frac{CB}{\delta}$.
By \Cref{thm: main for SSSP}, the total  running time of $\A$
is $\Ohat\left(\frac{(nK)^{2}\log(L)}{\epsilon^{2}}\right )=\Ohat\left(n^{2}\cdot \frac{\log(CB)}{\epsilon^{O(1)}}\right )$.

Next, we need to describe the sequence of edge deletions from the graph $G'$. The edges are deleted 
according to the following rule.  %
For every vertex $v\in V$ in the original graph, for every $0\leq i\leq K$, whenever $w''(v_{in},v_{out})>w(v_i)$,
we delete all edges incident to $v_i$ from $G'$. For convenience, we say that vertex $v_i$ becomes \emph{eliminated}.
We use the following analogue of Claim~\ref{claim: getting approximate paths}.

\begin{claim}
Throughout the execution of \Cref{alg:fleischer}, given any $(1+\epsilon/3)$-approximate $s$-$t$ path $P'$ in $G'$ with respect to edge lengths $w(e)$, we can construct, in time $O(|P'|)$, a $(1+\epsilon)$-approximate
$s_{out}$-$t_{in}$ path $P''$, with respect to edge lengths $w''(e)$, in $G''$.
\end{claim}

\begin{proof}
	Let $P''$ be the shortest $s_{out}$--$t_{in}$ path in graph $G''$, with respect to the edge lengths $w''$, and assume that $P''=(s_{out},v^1_{in},v^1_{out},\ldots,v^z_{in},v^z_{out},t_{in})$.  Let $W''$ denote the length of the path $P''$.
	For all $1\leq j\leq z$, let $e^j=(v^j_{in},v^j_{out})$. Since \Cref{alg:fleischer} did not yet terminate, $w''(e^j)< 1/c''(e^j)+b''(e^j)/\overline{b}$. Therefore, if we let $i_j$ be the smallest index, such that vertex $v^j_{i_j}$ of $G'$ is not yet eliminated, then $w(v^j_{i_j})\leq w''(e^j)(1+\eps/3)$. Consider now the following path in graph $G'$: $P'=(s,v^1_{j_1},v^2_{j_2},\ldots,v^z_{j_z})$. Since no vertex on this path is eliminated, the path is indeed still contained in $G'$. The total weight of the vertices on this path is bounded by $(1+\eps/3)W''$. From the above discussion, the total $w'$-weight of the edges on this path is then also bounded by $(1+\eps/3)W''$.
	
	We denote by $\dist''$ the distance from $s_{out}$ to $t_{in}$ in $G''$, with respect to the edge lengths $w''(e)$, and we denote by $\dist'$ the distance from $s$ to $t$ in graph $G'$, with respect to edge lengths $w(e)$. From the above discussion, $\dist'\leq (1+\eps/3)\dist''$.
	
	Consider now a $(1+\eps/3)$-approximate $s$-$t$ path $P'$ in graph $G'$, with respect to the edge lengths $w'(e)$, so the total weight $w'(e)$ of all edges on $P'$ is at most $(1+\eps/3)\dist'\leq (1+\eps/3)^2\dist''\leq (1+\eps)\dist''$. Assume that $P'=(s,v^1_{j_1},v^2_{j_2},\ldots,v^z_{j_z})$. Consider the following path in graph $G''$: $P=(s_{out},v^1_{in},v^1_{out},\ldots,v^z_{in},v^z_{out},t_{in})$. Note that for all $1\leq j\leq z$, the weight $w''(v^j_{in},v^j_{out})\leq w'(v^j)$ must hold (or vertex $v^j$ would have been eliminated). Therefore, the total weight $w''(e)$ of the edges of $P''$ is bounded by the total weight $w(v)$ of the vertices of $P'$, which in turn is equal to the total weight $w(e')$ of the edges of $P'$, that is bounded by $(1+\eps)\dist''$.
\end{proof}

From the above claim, in every iteration of \Cref{alg:fleischer}, we can use $\pquery(t)$ in graph $G'$ in order to compute the $(1+\eps)$-approximate shortest $s_{out}$-$t_{in}$ path in $G''$, with respect to edge lengths $w''=\ell''+\phi b''$. 
It now remains to analyze the running time of the algorithm.
Each query to the decremental \SSSP data structure takes time $\Ohat(|P|\log\log L) = 
\Ohat((n \log(BC)/\eps^{O(1)})$ to return a path $P$ and, from \Cref{lem:aug bound}(2), there are at most $\tilde{O}(n/\epsilon^{2})$ such queries. Therefore, the total  time spent on responding to the queries is $\Ohat\left(\frac{n^2 \log(BC)}{\eps^{O(1)}} \right)$.
As observed above, the total expected running time for maintaining the decremental \SSSP data structure is $\Ohat\left(n^{2}\cdot \frac{\log(CB)}{\epsilon^{O(1)}}\right )$.
The running time of other steps for implementing \Cref{alg:fleischer}, such as maintaining
$\ell''$ and $\phi$, is subsumed by these running times. Overall, the total running
time is $\Ohat\left(n^{2}\cdot \frac{\log(CB)}{\epsilon^{O(1)}}\right )$.

%% file: more_app.tex
\section{Additional Applications}
\label{sec:more_app}
In this section, we describe additional applications of decremental $\SSSP$, and some new results that follow from our algorithm from \Cref{thm: main for SSSP}, as well as additional results that could be obtained from the algorithm of \cite{fast-vertex-sparsest}.

\subsection{Concurrent $k$-commodity Bounded-Cost Flow }

In the concurrent $k$-commodity  bounded-cost flow problem, we are
given a graph $G$ with capacities and costs on either edges or nodes,
and a cost bound $\overline{b}$. We are also given $k$ \emph{demands} represented by tuples $(s_{1},t_{1},d_{1}),\dots,(s_{k},t_{k},d_{k})$, where for all $1\leq i\leq k$, $s_i$ and $t_i$ are vertices of $G$, that we refer to  as a \emph{demand pair}, and $d_i$ is a non-negative real number.
The goal is to find a largest value $\lambda>0$, and to compute, for all $i\in\{1,\dots,k\}$, an $s_{i}$-$t_{i}$
flow $f_{i}$ of value $\lambda d_{i}$ (that is, flow $f_{i}$ routes $\lambda d_{i}$
units of flow from $s_{i}$ to $t_{i}$). We say that the resulting flow $f=\bigcup_if_i$ is edge-capacity-feasible, iff for all $e\in E$, $\sum_{i=1}^{k}f_{i}(e)\le c(e)$. We say that the flow $f$ is edge-cost-feasible, if $\sum_{e\in E}(\sum_{i\le k}f_{i}(e))b(e)\le\overline{b}$.
The goal is to maximize $\lambda$, while ensuring that the resulting flow $f$ is both
capacity-feasible and cost-feasible. The problem with vertex capacities
and cost is defined analogously.

The concurrent $k$-commodity  flow problem is defined in the same way, except that we no longer have costs on edges or vertices, and we do not require that the flow is cost-feasible.

We denote by $T_{\MBCF}(n,m,\epsilon,B,C)$
the time needed for computing  a $(1+\epsilon)$-approximate solution for an \MBCF problem instance,  in a graph with $n$ nodes and $m$ edges, where $B$ is the ratio of largest to smallest (edge or vertex) costs, $C$ is the ratio of largest
to smallest (edge or vertex) capacities. We use the following result.

\begin{lemma}
[Concurrent $k$-commodity bounded-cost flow \cite{GK,Fleischer00}]\label{lem:k concurrent flow reduc}There
is an algorithm that, given a graph $G$ with $n$ nodes, $m$ edges,
(edge or vertex) capacities $c$,  where $C$ is the ratio of largest to smallest capacity,
(edge or vertex) costs $b$, and $B$ is the ratio of largest to smallest cost, and
a set of $k$ demands, computes a $(1+\epsilon)$-approximate
concurrent $k$-commodity bounded-cost flow in time $\tilde{O}\left (\frac{k}{\epsilon^{2}}\cdot T_{\MBCF}(n,m,\epsilon,BC/\delta,C)\cdot \log(BC)\right )$
where $\delta=(2m)^{-1/\epsilon}$. 
\end{lemma}

\begin{proof}
[Sketch] A similar lemma was shown by Garg and K\"{o}nemann
in Section 6 of \cite{GK}. However, the term $T_{\MBCF}(n,m,\epsilon,BC/\delta,C)$
in \cite{GK} was the time for computing \emph{exact} min-cost
flow. We sketch here why only $(1+\epsilon)$-approximate solution for \MBCF is sufficient.

For any commodity $1\leq i\leq k$  and edge lengths $\ell\in\mathbb{R}_{\ge0}^{E}$, let $\mincost_{i}(\ell)$
be minimum cost for sending a flow of $d_{i}$ units from $s_{i}$
to $t_{i}$ in $G=(V,E,c)$, where the edge costs are defined by $\ell$.
It was shown in \cite{GK}, that, in order to solve concurrent $k$-commodity bounded-cost flow, it is enough to solve
the following problem $O(\frac{1}{\epsilon^{2}}k\log k\log m)$
times: given $i$ and $\ell$, compute an $s_{i}$-$t_{i}$ flow $f_{i,\ell}$
of value $d_{i}$ and cost $\mincost_{i}(\ell)$ w.r.t. $\ell$. 

Let $\A$ be the $(1+\epsilon)$-approximate algorithm for \MBCF. We claim that, by calling this algorithm $O(\log BC)$ times, we
can compute an $s_{i}$-$t_{i}$ flow $f'_{i,\ell}$ such that the
value of $f'_{i,\ell}$ is exactly $d_{i}/(1+\epsilon)$, and its cost
is at most $\mincost_{i}(\ell)$. Indeed, observe that, when given
$\mincost_{i}(\ell)$ as a cost bound to $\A$, algorithm $\A$ will return
a flow of value at least $d_{i}/(1+\epsilon)$. By scaling, we obtain
a flow of value exactly $d_{i}/(1+\epsilon)$ and cost at most
$\mincost_{i}(\ell)$.

By following the analysis of \cite{GK}, it is easy to see that that we can
use $f'_{i,\ell}$ instead of $f_{i,\ell}$, for any given $i$ and
$\ell$. Every step in the analysis works as it is except that the
size of the solution at the end is reduced by factor $(1+\epsilon)$.%

\end{proof}

By plugging \Cref{thm:mincost edge} and \Cref{thm:mincost vertex}
into the above lemma, we obtain the following corollary:

\begin{corollary}
\label{cor:k concurrent flow}There is a deterministic algorithm for
computing a $(1+\epsilon)$-approximate concurrent $k$-commodity 
bounded-cost flow in time ${\Ohat}(kn^{2}\frac{\log^{2}BC}{\epsilon^{O(1)}})$
on either:
\begin{itemize}
\item  undirected simple graphs with unit edge-capacities and arbitrary edge-costs;
or
\item  undirected simple graphs with arbitrary vertex-capacities and vertex-costs.
\end{itemize}
\end{corollary}

Our algorithm for concurrent $k$-commodity flow is slower than the $\tilde{O}(mk)$
algorithm of Sherman \cite{Sherman17}. However, in the bounded-cost
version, our algorithms are
faster than the previous best algorithms whenever $m=\omega(n^{1.5+o(1)})$
and $k=O((m/n)^{2})$; see \Cref{tab:edge flow}. %
We note that the algorithm for vertex-capacitated graphs can also be obtained from the results of \cite{fast-vertex-sparsest}, except that the resulting algorithm would be randomized.

\subsection{Maximum $k$-commodity Bounded-Cost Flow}

In the maximum $k$-commodity bounded-cost flow, we are given a graph
$G$ with capacities and costs on either edges or nodes, and a cost
bound $\overline{b}$. We are also given $k$ \emph{demand pairs}
 $(s_{1},t_{1}),\dots,(s_{k},t_{k})$. The goal is
to compute, for all $i\in\{1,\dots,k\}$, an $s_{i}$-$t_{i}$ flow $f_{i}$
with the following constraints. Let $f=\bigcup_i f_i$ be the resulting $k$-commodity flow. We say that the flow is edge-capacity-feasible,
if $\sum_{i=1}^{k}f_{i}(e)\le c(e)$ for all $e\in E$, and we say that it is edge-cost-feasible, if $\sum_{e\in E}(\sum_{i=1}^kf_{i}(e))b(e)\le\overline{b}$.
Let $|f_{i}|$ denote the value of the flow $f_{i}$ -- the amount of flow sent from $s_i$ to $t_i$. The goal is to find the flows $f_1,\ldots,f_k$ that maximize $\sum_{i}|f_{i}|$, such that the resulting flow $f=\bigcup_if_i$ is both edge-capacity-feasible and edge-cost feasible.
The problem with vertex capacities and cost is defined similarly.

The maximum $k$-commodity flow problem is defined similarly, except
that there are no costs on edges or vertices, and we do not require that the flow $f$ is cost-feasible.

We obtain the following corollary.

\begin{corollary}
\label{cor:k max flow}There is a deterministic algorithm, that, given a graph $G$
with $n$ nodes, $m$ edges, capacities $c$ where $C$ is the ratio
of largest to smallest capacity, costs $b$ where $B$ is the ratio
of largest to smallest cost, and a set of $k$ demand pairs, can compute
a $(1+\epsilon)$-approximate maximum $k$-commodity bounded-cost flow
in time ${\Ohat}(kn^{2}\frac{\log BC}{\epsilon^{O(1)}})$ on either:
\begin{itemize}
\item undirected simple graphs with unit edge-capacities and arbitrary edge-costs;
or
\item undirected simple graphs with arbitrary vertex-capacities and vertex-costs.
\end{itemize}
\end{corollary}

As before, the result for vertex-capacitated graphs could also be obtained from \cite{fast-vertex-sparsest}, except that the resulting algorithm would be randomized.
To our best knowledge, unlike the concurrent $k$-commodity flow, there
is no black-box reduction from maximum $k$-commodity flow or maximum $k$-commodity
 bounded-cost flow to \MBCF. However, \Cref{cor:k max flow}
can be proved by extending the MWU-based technique used in \Cref{sec:mincost}
to the maximum $k$-commodity bounded-cost flow,  and employing the algorithm for dynamic $\SSSP$
for executing each iteration efficiently; we omit the proof.

Our algorithm for maximum $k$-commodity flow is faster than the $O(k^{O(1)}m^{4/3}/\epsilon^{O(1)})$-time
algorithm by \cite{KelnerMP12} and the $\tilde{O}(m^{2}/\eps^{2})$-time
algorithm by \cite{Fleischer00} whenever $m=\omega(n^{1.5+o(1)})$
and $k=O((m/n)^{2})$. The same bounds hold in the bounded-cost
version and in the vertex-capacitated setting: our algorithms are
faster than the previous best algorithms whenever $m=\omega(n^{1.5+o(1)})$
and $k=O((m/n)^{2})$; see \Cref{tab:edge flow}. %

Lastly, we describe several additional applications of the \SSSP problem, that can be obtained from the algorithm of \cite{fast-vertex-sparsest} (as well as from our algorithm from \Cref{thm: main for SSSP})

\subsection{Most-Balanced Sparsest Vertex Cut}

Given a graph $G=(V,E)$, a vertex cut  is a partition $(A,B,C)$ of
the vertex set $V$, such that there are no edges between $A$ and $C$,
and $A,C\neq\emptyset$. The sparsity of the cut $(A,B,C)$
is $h_{G}(A,B,C)=\frac{|B|}{\min\{|A|,|C|\}+|B|}$. We say that a
cut $(A,B,C)$ is $\phi$-sparse if $h_{G}(A,B,C)<\phi$. The \emph{most
balanced} $\phi$-sparse cut is a $\phi$-sparse cut $(A,B,C)$ such
that $\min\{|A|,|C|\}$ is maximized. The vertex expansion of a graph
$G$ is $h_{G}=\min\{h_{G}(A,B,C)\mid(A,B,C)$ is a vertex cut of
$G\}$. 

In the $\alpha$-approximate most-balanced sparsest vertex cut problem, we are given
a graph $G=(V,E)$ and a parameter $h_G$. Let $(A',B',C')$ be a most-balanced $h_{G}$-sparsest
cut. The goal is to find a vertex cut $(A,B,C)$ with $h_{G}(A,B,C)\le\alpha\cdot h_{G}(A',B',C')$, such that $\min\{|A|,|C|\}\ge\min\{|A'|,|C'|\}/3$. The following result follows from the algorithm of \cite{fast-vertex-sparsest}.

\begin{lemma}
[Most-balanced sparsest vertex cut]\label{lem:sparse cut reduc}There
is a randomized algorithm, that, given a graph $G$ with $n$ nodes, computes
a $O(\log^{2}n)$-approximate most-balanced sparsest vertex cut in
time $O(T_{mf}(n,2,n)\log^{2}n)$ where and $T_{mf}(n,\epsilon,C)$
is the time required to compute a $(1+\epsilon)$-approximate maximum $s$-$t$ flow and a $(1+\epsilon)$-approximate minimum $s$-$t$ cut in an $n$-vertex graph with vertex capacities, where
$C$ is the ratio of largest to smallest vertex capacity.
\end{lemma}

The lemma follows from the cut-matching game framework of Khandekar, Rao,
and Vazirani \cite{KRV}. The algorithm of \cite{KRV} is designed to compute a sparsest cut or minimum balanced cut in edge-capacitated graphs, but this is
only because it relies on maximum flow / minimum cut computation in edge-capacitated graphs.
By computing approximate maximum flow / minimum cut in  vertex-capacitated graphs, 
one can immediately obtain \Cref{lem:sparse cut reduc}. By plugging \Cref{thm:mincost vertex}
into the above lemma, we obtain the following corollary:

\begin{corollary}
\label{cor:sparse cut}There is a randomized algorithm for computing
a $O(\log^{2}n)$-approximate most-balanced sparsest vertex cut in a given $n$-vertex graph $G$,
in expected time ${\Ohat}(n^{2})$.
\end{corollary}

\subsection{Treewidth and Tree Decompositions}

Given a graph $G=(V,E)$, a tree decomposition of $G$ consists of a tree $T$, and, for every vertex $a\in V(T)$, a subset $X_{a}\subseteq V$
of vertices of $G$, that satisfy the following conditions: (i) for each edge
$(u,v)\in E$ of $G$, there is a tree-node $a\in V(T)$ with $u,v\in X_{a}$; and
(ii) for each vertex $u\in V$ of $G$, all tree-nodes $a\in V(T)$ with $u\in X_{a}$
induce a non-empty connected subgraph of $T$. The width of the tree decomposition is
$\max_{a\in V(T)}|X_{a}|-1$. The treewidth of $G$ is the minimum width
of a tree decomposition of $G$. Treewidth and tree decomposition
are used extensively in algorithmic graph theory and in fixed parameter tractable (FPT)
algorithms. 

The following lemma reduces the problem of approximating treewidth
to the most-balanced sparsest vertex cut problem, using standard techniques. We omit the proof; see also \cite{BodlaenderGHK95}.

\begin{lemma}
\label{lem:treewidth reduc}There
is an algorithm that, given an $n$-vertex graph $G$ and a parameter $\alpha$, computes a tree decomposition of $G$ of width $O(k\alpha \log n)$, where $k$ is the treewidth of $G$, in time $\tilde{O}(T_{svc}(n,\alpha))$ where $T_{svc}(n,\alpha)$
is the time needed for computing an $\alpha$-approximate most-balanced sparsest
vertex cut in $G$. %
\end{lemma}

By plugging \Cref{cor:sparse cut} into the above lemma, we obtain the following corollary:
\begin{corollary}
\label{cor:treewidth}There is a deterministic algorithm that, given an $n$-vertex graph $G$, computes a tree decomposition of $G$, of width  $O(k\log^{3}n)$, where $k$ is the treewidth of $G$, in time ${\Ohat}(n^{2})$.
\end{corollary}
For comparison, given a graph with $n$ nodes and treewidth $k$,
previous algorithms either have running time exponentially depending
on $k$ \cite{flat-wall-RS,Amir01,Amir10,Bodlaender96,BodlaenderDDFLP16}
or have a large polynomial running time \cite{BodlaenderGHK95,Amir01,Amir10,FeigeHL05}.
One exception is the algorithm by Fomin et al. \cite{FominLSPW18}
which computes an $O(k)$-approximation of treewidth in time $O(k^{7}n\log n)$; see \Cref{tab:treewidth} for a summary. Our algorithm is faster than \cite{FominLSPW18}
when $k\ge n^{1/7+o(1)}$ and also gives a better approximation. 

Although most of fixed parameter tractable (FPT) algorithms only concern graphs with constant treewidth $k=O(1)$ or very small $k$, there is a recent line of work on \emph{fully-polynomial} FPT algorithms \cite{FominLSPW18,IwataOO18} for many fundamental graph problems
including maximum matching and max flow, and matrix problems including
determinant and solving linear system, which concern instances whose treewidth can be polynomial. In those settings, the approximation factor of $O(\log^3 n)$ from \Cref{cor:treewidth} is of interest.

%% file: table.tex
\section{Tables}
\label{sec: tables}

\begin{table}[H]
\footnotesize{

\begin{tabular}{|>{\raggedright}p{0.11\textwidth}|>{\centering}p{0.04\textwidth}|>{\centering}p{0.13\textwidth}|>{\centering}p{0.18\textwidth}|>{\centering}p{0.11\textwidth}|>{\centering}p{0.08\textwidth}|>{\centering}p{0.08\textwidth}|>{\centering}p{0.08\textwidth}|}
\hline 
 & Year & $(\alpha,\beta)$-approximation & Total update time & Query time for $\dquery$  & Weighted? & Directed? & Det?\tabularnewline
\hline 
\hline 
\cite{EvenS}{*} & 1981 & $(1,0)$ & $O(mn^{2})$ & $O(1)$ &  & Directed & Det\tabularnewline
\hline 
\cite{BaswanaHS07}{*} & 2002 & $(1,0)$ & $\tilde{O}(n^{3})$ & $O(1)$ &  & Directed & \tabularnewline
\hline 
\cite{BaswanaHS07} & 2002 & $(1+\epsilon,0)$ & $\tilde{O}(n^{2}\sqrt{m}/\epsilon)$ & $O(1)$ &  & Directed & \tabularnewline
\hline 
\cite{rodittyZ2} & 2004 & $(1+\epsilon,0)$ & $\tilde{O}(mn/\epsilon)$ & $O(1)$ &  &  & \tabularnewline
\hline 
\cite{henzinger16}{*} & 2013 & $(1+\epsilon,0)$ & $\tilde{O}(mn/\epsilon)$ & $O(\log\log n)$ &  &  & Det\tabularnewline
\hline 
\cite{bernstein16}{*} & 2013 & $(1+\epsilon,0)$ & $\tilde{O}(mn\log L/\epsilon)$ & $O(1)$ & Weighted & Directed & \tabularnewline
\hline 
\cite{RodittyZ11,HenzingerKNS15} & 2004 & $(\alpha,\beta)$: $2\alpha+\beta<4$ & $\Omega(n^{3-o(1)})$ & $\Omega(n^{1-o(1)})$ &  &  & \tabularnewline
\hline 
\hline 
\cite{BaswanaKS12} (cf. \cite{ForsterG19})*  & 2008 & $(2k-1,0)$ & $\tilde{O}(m)$ & $\tilde{O}(n^{1+1/k})$ & Weighted &  & \tabularnewline
\hline 
\cite{BernsteinBGNSS20}* & 2020 & $(\poly\log n,0)$ & $\tilde{O}(m)$ & $\tilde{O}(n)$ & Weighted &  & Random adaptive\tabularnewline
\hline 
\hline 
\cite{BernsteinR11} & 2011 & $(2k-1+\epsilon,0)$ & $O(n^{2+1/k+o(1)})$ & $O(k)$ &  &  & \tabularnewline
\hline 
\cite{henzinger16}{*} & 2013 & $(2+\epsilon,0)$ or $(1+\epsilon,2)$ & $\tilde{O}(n^{2.5}/\epsilon)$ & $O(1)$ &  &  & \tabularnewline
\hline 
\cite{abraham2014fully} & 2014 & $(2^{O(k\rho)},0)$ & $O(mn^{1/k})$ & $O(k\rho)$ &  &  & \tabularnewline
\hline 
\cite{HenzingerKN14_focs} & 2014 & $((2+\epsilon)^{k}-1,0)$ & $O(m^{1+1/k+o(1)}\log^{2}L)$ & $O(k^{k})$ & Weighted &  & \tabularnewline
\hline 
\cite{chechik}{*} & 2018 & $((2+\epsilon)k-1,0)$ & $O(mn^{1/k+o(1)}\log L)$ & $O(\log\log (nL))$ & Weighted &  & \tabularnewline

\hline 
\textbf{This paper{*}} &  & $(3\cdot 2^{k},\gamma^{O(k)})$  & $\Ohat(n^{2.5+2/k}\gamma^{O(k)})$  & $O(\log\log n)$  &  &  & Det\tabularnewline
\hline 

\end{tabular}

}

\caption{Upper and lower bounds for decremental \APSP. We denote by $n$ the number of graph vertices, by $m$  the initial number edges, $L$ is
the ratio of largest to smallest edge length, and $k$ is a positive integral parameter. We also use parameters
$\rho=(1+\left\lceil \frac{\log n^{1-1/k}}{\log(m/n^{1-1/k})}\right\rceil )\le k$,
 $0<\epsilon<1$, and $\gamma=\exp(\log^{3/4}n)$. In the ``Year'' column, the year is according
to the conference version of the paper. 
If the algorithm only works for unweighted graphs or only undirected graphs, 
then we left the columns ``Weighted?'' and ``Directed?'', respectively, blank. 
If the result assumes an oblivious adversary, then we left the column ``Det?'' blank. Otherwise, we write ``Det'' or ``Random adaptive'' which means that the result is deterministic or randomized but works against an adaptive adversary, respectively. 
The algorithms without the ``{*}'' mark are subsumed by other algorithms with the ``{*}'' mark, to within $n^{o(1)}$ factors. 
The fact
that the algorithm in \cite{EvenS} can be extended to work
in directed graphs was observed in \cite{HenzingerKing}. The algorithms by \cite{BaswanaKS12,ForsterG19,BernsteinBGNSS20} are actually fully dynamic algorithms for maintaining spanners, but they automatically imply decremental APSP with large query time for \distquery.
\label{tab:APSP}}
\end{table}

\begin{table}[H]
\footnotesize{

\begin{tabular}{|>{\centering}p{0.1\textwidth}|>{\centering}p{0.04\textwidth}|>{\centering}p{0.06\textwidth}|>{\centering}p{0.14\textwidth}|>{\centering}p{0.09\textwidth}|>{\centering}p{0.09\textwidth}|>{\centering}p{0.09\textwidth}|>{\centering}p{0.09\textwidth}|>{\centering}p{0.13\textwidth}|}
\hline 
 & Year & Approx. & Total update time & Handles $\pquery$? & Query time for $\pquery$ & Weighted? & Det? & Notes\tabularnewline
\hline 
\hline 
\cite{EvenS}{*} & 1981 & exact & $O(mn)$ & Yes & $O(|P|)$ &  & Det & \tabularnewline
\hline 
\cite{RodittyZ11,HenzingerKNS15} & 2004 & exact & $\Omega(mn^{1-o(1)})$ &  &  &  &  & \tabularnewline
\hline 
\hline 
\cite{BernsteinR11} & 2011 & $1+\epsilon$ & $O(n^{2+o(1)})$ & Yes & $O(|P|)$ &  &  & \tabularnewline
\hline 
\cite{HenzingerKN14_soda} & 2014 & $1+\epsilon$ & $O(n^{1.8+o(1)}+m^{1+o(1)})$ & Yes & $O(|P|)$ &  &  & \tabularnewline
\hline 
\cite{HenzingerKN14_focs}{*} & 2014 & $1+\epsilon$ & $O(m^{1+o(1)}\log L)$ & Yes & $\tilde{O}(|P|)$ & Weighted &  & \tabularnewline
\hline 
\hline 
\cite{BernsteinChechik} & 2016 & $1+\epsilon$ & $\tilde{O}(n^{2})$ &  &  &  & Det & \tabularnewline
\hline 
\cite{BernsteinChechikSparse} & 2017 & $1+\epsilon$ & $\tilde{O}(n^{5/4}\sqrt{m})$ &  &  &  & Det & \tabularnewline
\hline 
\cite{Bernstein} & 2017 & $1+\epsilon$ & $\tilde{O}(n^{2}\log L)$ &  &  & Weighted & Det & \tabularnewline
\hline 
\cite{fast-vertex-sparsest} & 2019 & $1+\epsilon$ & $\Ohat(n^{2}\log L)$ & Yes & $\tilde{O}(n\log L)$ & Weighted & Random adaptive & Vertex deletions only\tabularnewline
\hline 
\cite{GutenbergW20} & 2020 & $1+\epsilon$ & $\Ohat(m\sqrt{n})$ &  &  &  & Det & \tabularnewline
\hline 
\cite{BernsteinBGNSS20}{*} & 2020 & $1+\epsilon$ & $\Ohat(m\sqrt{n})$ & Yes & ${\Otil}(n)$ &  & Random adaptive & \tabularnewline
\hline 
\textbf{This paper{*}} &  & $1+\epsilon$ & $\Ohat(n^{2}\log L)$ & Yes & ${\Ohat}(|P|)$ & Weighted & Det & \tabularnewline
\hline 
\end{tabular}

}

\caption{Upper and lower bounds for decremental \SSSP. We denote by $n$
the number of graph vertices, by $m$ the initial number of edges, $L$
is the ratio of largest to smallest edge length, and $0<\epsilon<1$ is a given parameter.
The dependency of the running time on $\epsilon$ is omitted for simplicity. We denote by $P$ the
(approximate) shortest path returned in response to $\protect\pquery$, and
by $|P|$ the number of edges in $P$. In the ``Year'' column, the
year is according to the conference version of the paper. If a result
works only on unweighted graphs, then we left the column
``Weighted?'' blank. If a result assumes an oblivious adversary, then
we left the column ``Det?'' blank. 
Otherwise, we write ``Det'' or ``Random adaptive'' which means that the result is deterministic or randomized but works against an adaptive adversary, respectively. 
The algorithms without
the ``{*}'' mark are subsumed by other algorithms with the ``{*}'' mark, to within
$\protect\poly\log n$ factors. \label{tab:SSSP}}
\end{table}

\begin{table}[H]
\footnotesize{

\begin{tabular}{|>{\raggedright}p{0.25\textwidth}|>{\raggedright}p{0.33\textwidth}|>{\raggedright}p{0.15\textwidth}|>{\raggedright}p{0.2\textwidth}|}
\hline 
\textbf{Problems in unit edge-capacity setting} & \textbf{Previous best} & \textbf{This paper } & \textbf{Faster when}\tabularnewline
\hline 
\hline 
maximum $s$-$t$ flow  & $\tilde{O}(m/\epsilon)$ \cite{Sherman17} & $\Ohat(n^{2}/\epsilon^{O(1)})$  & -\tabularnewline
\hline 
$k$-commodity concurrent flow & $\tilde{O}(km/\epsilon)$ \cite{Sherman17} & $\Ohat(kn^{2}/\epsilon^{O(1)})$ & -\tabularnewline
\hline 
maximum $k$-commodity flow & $O(k^{O(1)}m^{4/3}/\epsilon^{O(1)})$ \cite{KelnerMP12} 

$\tilde{O}(mn/\eps^{2})$ \cite{Madry10_stoc} & $\Ohat(kn^{2}/\epsilon^{O(1)})$ & $m=\omega(n^{1.5+o(1)})$ and

$k=O((m/n)^{2})$\tabularnewline
\hline 
maximum bounded-cost $s$-$t$ flow  & $\tilde{O}(m\sqrt{n})$ \cite{LeeS14} (exact)

$\tilde{O}(m^{10/7})$ \cite{CohenMSV17}(exact) & $\Ohat(n^{2}/\epsilon^{O(1)})$ & $m=\omega(n^{1.5+o(1)})$ \tabularnewline
\hline 
 $k$-commodity concurrent bounded-cost flow & $\tilde{O}(km\sqrt{n}/\epsilon^{O(1)})$ \cite{LeeS14}+\Cref{lem:k concurrent flow reduc}

$\tilde{O}(m(m+k)/\eps^{2})$ \cite{Fleischer00} & $\Ohat(kn^{2}/\epsilon^{O(1)})$ & $m=\omega(n^{1.5+o(1)})$ and

$k=O((m/n)^{2})$\tabularnewline
\hline 
maximum $k$-commodity bounded-cost flow & $\tilde{O}(m(m+k)/\eps^{2})$ \cite{Fleischer00} & $\Ohat(kn^{2}/\epsilon^{O(1)})$ & $k=O((m/n)^{2})$\tabularnewline
\hline 
\end{tabular}

}

\caption{Best currently known running times of algorithms for flow and cut problems in undirected simple graphs
with unit edge capacities. We use the $\tilde{O}$ notation to hide factors that are polylogarithmic in $n$ and $B$ -- the  ratio of maximum to minimum edge cost. All algorithms obtain a $(1+\epsilon)$-approximation for the corresponding
problem, unless explicitly stated otherwise.\label{tab:edge flow}}
\end{table}

\begin{table}[H]
\footnotesize{

\begin{tabular}{|>{\raggedright}p{0.25\textwidth}|>{\raggedright}p{0.33\textwidth}|>{\raggedright}p{0.15\textwidth}|>{\raggedright}p{0.2\textwidth}|}
\hline 
\textbf{Problem in vertex-capacitated setting} & \textbf{Previous best} & \textbf{This paper  /  follows from \cite{fast-vertex-sparsest}} & \textbf{Faster when}\tabularnewline
\hline 
\hline 
maximum $s$-$t$ flow  & $\Ohat(n^{2}/\epsilon^{O(1)})$ \cite{fast-vertex-sparsest}

$\tilde{O}(m\sqrt{n})$ \cite{LeeS14} (exact) & $\Ohat(n^{2}/\epsilon^{O(1)})$ & - %
\tabularnewline
\hline 
$k$-commodity concurrent flow & $\tilde{O}(km\sqrt{n}/\epsilon^{O(1)})$ \cite{LeeS14}+\Cref{lem:k concurrent flow reduc}, 

$\tilde{O}(mn/\eps^{2})$ \cite{Madry10_stoc} & $\Ohat(kn^{2}/\epsilon^{O(1)})$ & $m=\omega(n^{1.5+o(1)})$ and

$k=O((m/n)^{2})$\tabularnewline
\hline 
maximum $k$-commodity flow & $\tilde{O}(mn/\eps^{2})$ \cite{Madry10_stoc} & $\Ohat(kn^{2}/\epsilon^{O(1)})$ & $k=O((m/n)^{2})$\tabularnewline
\hline 
maximum bounded-cost $s$-$t$ flow  & $\tilde{O}(m\sqrt{n})$ \cite{LeeS14} (exact) & $\Ohat(n^{2}/\epsilon^{O(1)})$ & $m=\omega(n^{1.5+o(1)})$\tabularnewline
\hline 
$k$-commodity concurrent bounded-cost flow & $\tilde{O}(km\sqrt{n}/\epsilon^{O(1)})$ \cite{LeeS14}+\Cref{lem:k concurrent flow reduc}

$\tilde{O}(m(m+k)/\eps^{2})$ \cite{Fleischer00} & $\Ohat(kn^{2}/\epsilon^{O(1)})$ & $m=\omega(n^{1.5+o(1)})$ and

$k=O((m/n)^{2})$\tabularnewline
\hline 
maximum $k$-commodity bounded-cost flow & $\tilde{O}(m(m+k)/\eps^{2})$ \cite{Fleischer00} & $\Ohat(kn^{2}/\epsilon^{O(1)})$ & $k=O((m/n)^{2})$\tabularnewline
\hline 
$O(\log^{2}n)$-approximate sparsest cut & $\Ohat(n^{2})$ \cite{fast-vertex-sparsest}

$\tilde{O}(m\sqrt{n})$ \cite{LeeS14}+\Cref{lem:sparse cut reduc} & $\Ohat(n^{2})$ & - %
\tabularnewline
\hline 
\end{tabular}

}

\caption{Best currently known algorithms for flow and cut problems in undirected graphs with
vertex capacities. We use $\tilde{O}$ notation to hide factors that are polylogarithmic in $n,C$ and $B$, where
$C$ is the ratio of maximum to minimum vertex capacity, and $B$ is the ratio
of maximum to minimum vertex cost. All algorithms are for obtaining a $(1+\epsilon)$-approximation for the corresponding problem, unless explicitly stated otherwise.\label{tab:vertex flow}}
\end{table}

\begin{table}[H]
\begin{centering}
\begin{tabular}{|c|c|c|}
\hline 
\textbf{Reference} & \textbf{Approximation} & \textbf{Time}\tabularnewline
\hline 
\hline 
\multicolumn{3}{|c|}{FPT time}\tabularnewline
\hline 
\cite{flat-wall-RS} & 4 & $2^{O(k)}n^{2}$\tabularnewline
\hline 
\cite{Amir01,Amir10} & $3\frac{2}{3}$ & $2^{O(k)}n^{3}$\tabularnewline
\hline 
\cite{Bodlaender96} & 1 & $k^{O(k^{3})}n$\tabularnewline
\hline 
\multirow{2}{*}{\cite{BodlaenderDDFLP16}} & 3 & $2^{O(k)}n\log n$\tabularnewline
\cline{2-3} 
 & 5 & $2^{O(k)}n$\tabularnewline
\hline 
\multicolumn{3}{|c|}{Polynomial time}\tabularnewline
\hline 
\cite{BodlaenderGHK95} & $O(\log n)$ & $\poly(n)$\tabularnewline
\hline 
\cite{Amir01,Amir10} & $O(\log k)$ & $k^{5}n^{3}\text{polylog}(nk)$\tabularnewline
\hline 
\cite{FeigeHL05} & $O(\sqrt{\log k})$ & $\poly(n)$\tabularnewline
\hline 
\cite{FominLSPW18} & $O(k)$ & $k^{7}n$\tabularnewline
\hline 
\textbf{This paper / follows from \cite{fast-vertex-sparsest}} & $O(\log^{3}n)$ & $n^{2+o(1)}$\tabularnewline
\hline 
\end{tabular}
\par\end{centering}
\caption{Algorithms for approximating treewidth of a graph with $n$ nodes
and treewidth $k$. \label{tab:treewidth}}
\end{table}

%% file: main-SSSP-and-APSP.bbl
\newcommand{\etalchar}[1]{$^{#1}$}
\def\cprime{$'$} \def\cprime{$'$}
\begin{thebibliography}{BvdBG{\etalchar{+}}20}

\bibitem[ACT14]{abraham2014fully}
Ittai Abraham, Shiri Chechik, and Kunal Talwar.
\newblock Fully dynamic all-pairs shortest paths: Breaking the {O (n)} barrier.
\newblock In {\em LIPIcs-Leibniz International Proceedings in Informatics},
  volume~28. Schloss Dagstuhl-Leibniz-Zentrum fuer Informatik, 2014.

\bibitem[AHdLT05]{AlstrupHLT05}
Stephen Alstrup, Jacob Holm, Kristian de~Lichtenberg, and Mikkel Thorup.
\newblock Maintaining information in fully dynamic trees with top trees.
\newblock {\em {ACM} Trans. Algorithms}, 1(2):243--264, 2005.

\bibitem[AHK12]{AroraHK12}
Sanjeev Arora, Elad Hazan, and Satyen Kale.
\newblock The multiplicative weights update method: a meta-algorithm and
  applications.
\newblock {\em Theory of Computing}, 8(1):121--164, 2012.

\bibitem[Ami01]{Amir01}
Eyal Amir.
\newblock Efficient approximation for triangulation of minimum treewidth.
\newblock In {\em {UAI} '01: Proceedings of the 17th Conference in Uncertainty
  in Artificial Intelligence, University of Washington, Seattle, Washington,
  USA, August 2-5, 2001}, pages 7--15, 2001.

\bibitem[Ami10]{Amir10}
Eyal Amir.
\newblock Approximation algorithms for treewidth.
\newblock {\em Algorithmica}, 56(4):448--479, 2010.

\bibitem[AMV20]{Axiotis2020circulation}
Kyriakos Axiotis, Aleksander Madry, and Adrian Vladu.
\newblock Circulation control for faster minimum cost flow in unit-capacity
  graphs.
\newblock {\em arXiv preprint arXiv:2003.04863}, 2020.

\bibitem[BC16]{BernsteinChechik}
Aaron Bernstein and Shiri Chechik.
\newblock Deterministic decremental single source shortest paths: beyond the
  {O(mn)} bound.
\newblock In {\em Proceedings of the forty-eighth annual ACM symposium on
  Theory of Computing}, pages 389--397. ACM, 2016.

\bibitem[BC17]{BernsteinChechikSparse}
Aaron Bernstein and Shiri Chechik.
\newblock Deterministic partially dynamic single source shortest paths for
  sparse graphs.
\newblock In {\em Proceedings of the Twenty-Eighth Annual ACM-SIAM Symposium on
  Discrete Algorithms}, pages 453--469. SIAM, 2017.

\bibitem[BDD{\etalchar{+}}16]{BodlaenderDDFLP16}
Hans~L. Bodlaender, P{\aa}l~Gr{\o}n{\aa}s Drange, Markus~S. Dregi, Fedor~V.
  Fomin, Daniel Lokshtanov, and Michal Pilipczuk.
\newblock A c\({}^{\mbox{k}}\) n 5-approximation algorithm for treewidth.
\newblock {\em {SIAM} J. Comput.}, 45(2):317--378, 2016.

\bibitem[Ber16]{bernstein16}
Aaron Bernstein.
\newblock Maintaining shortest paths under deletions in weighted directed
  graphs.
\newblock {\em SIAM Journal on Computing}, 45(2):548--574, 2016.

\bibitem[Ber17]{Bernstein}
Aaron Bernstein.
\newblock Deterministic partially dynamic single source shortest paths in
  weighted graphs.
\newblock In {\em LIPIcs-Leibniz International Proceedings in Informatics},
  volume~80. Schloss Dagstuhl-Leibniz-Center for Computer Science, 2017.

\bibitem[BGHK95]{BodlaenderGHK95}
Hans~L. Bodlaender, John~R. Gilbert, Hj{\'{a}}lmtyr Hafsteinsson, and Ton
  Kloks.
\newblock Approximating treewidth, pathwidth, frontsize, and shortest
  elimination tree.
\newblock {\em J. Algorithms}, 18(2):238--255, 1995.

\bibitem[BGS20]{BernsteinGS20scc}
Aaron Bernstein, Maximilian~Probst Gutenberg, and Thatchaphol Saranurak.
\newblock Deterministic decremental reachability, scc, and shortest paths via
  directed expanders and congestion balancing.
\newblock 2020.
\newblock To appear at FOCS'20.

\bibitem[BHI15]{BhattacharyaHI15}
Sayan Bhattacharya, Monika Henzinger, and Giuseppe~F. Italiano.
\newblock Deterministic fully dynamic data structures for vertex cover and
  matching.
\newblock In {\em Proceedings of the Twenty-Sixth Annual {ACM-SIAM} Symposium
  on Discrete Algorithms, {SODA} 2015, San Diego, CA, USA, January 4-6, 2015},
  pages 785--804, 2015.

\bibitem[BHN16]{BhattacharyaHN16}
Sayan Bhattacharya, Monika Henzinger, and Danupon Nanongkai.
\newblock New deterministic approximation algorithms for fully dynamic
  matching.
\newblock In {\em Proceedings of the 48th Annual {ACM} {SIGACT} Symposium on
  Theory of Computing, {STOC} 2016, Cambridge, MA, USA, June 18-21, 2016},
  pages 398--411, 2016.

\bibitem[BHS07]{BaswanaHS07}
Surender Baswana, Ramesh Hariharan, and Sandeep Sen.
\newblock Improved decremental algorithms for maintaining transitive closure
  and all-pairs shortest paths.
\newblock {\em J. Algorithms}, 62(2):74--92, 2007.

\bibitem[BK19]{BhattacharyaK19}
Sayan Bhattacharya and Janardhan Kulkarni.
\newblock Deterministically maintaining a ($2 + \epsilon$)-approximate minimum
  vertex cover in $o(1/\epsilon^2)$ amortized update time.
\newblock In {\em Proceedings of the Thirtieth Annual {ACM-SIAM} Symposium on
  Discrete Algorithms, {SODA} 2019, San Diego, California, USA, January 6-9,
  2019}, pages 1872--1885, 2019.

\bibitem[BKS12]{BaswanaKS12}
Surender Baswana, Sumeet Khurana, and Soumojit Sarkar.
\newblock Fully dynamic randomized algorithms for graph spanners.
\newblock {\em {ACM} Trans. Algorithms}, 8(4):35:1--35:51, 2012.

\bibitem[Bod96]{Bodlaender96}
Hans~L. Bodlaender.
\newblock A linear-time algorithm for finding tree-decompositions of small
  treewidth.
\newblock {\em {SIAM} J. Comput.}, 25(6):1305--1317, 1996.

\bibitem[BR11]{BernsteinR11}
Aaron Bernstein and Liam Roditty.
\newblock Improved dynamic algorithms for maintaining approximate shortest
  paths under deletions.
\newblock In {\em Proceedings of the Twenty-Second Annual {ACM-SIAM} Symposium
  on Discrete Algorithms, {SODA} 2011, San Francisco, California, USA, January
  23-25, 2011}, pages 1355--1365, 2011.

\bibitem[BvdBG{\etalchar{+}}20]{BernsteinBGNSS20}
Aaron Bernstein, Jan van~den Brand, Maximilian~Probst Gutenberg, Danupon
  Nanongkai, Thatchaphol Saranurak, Aaron Sidford, and He~Sun.
\newblock Fully-dynamic graph sparsifiers against an adaptive adversary.
\newblock {\em CoRR}, abs/2004.08432, 2020.

\bibitem[CGL{\etalchar{+}}19]{ChuzhoyGLNPS19}
Julia Chuzhoy, Yu~Gao, Jason Li, Danupon Nanongkai, Richard Peng, and
  Thatchaphol Saranurak.
\newblock A deterministic algorithm for balanced cut with applications to
  dynamic connectivity, flows, and beyond.
\newblock {\em CoRR}, abs/1910.08025, 2019.

\bibitem[Che18]{chechik}
Shiri Chechik.
\newblock Near-optimal approximate decremental all pairs shortest paths.
\newblock In {\em Proc. of the IEEE 59th Annual Symposium on Foundations of
  Computer Science}, 2018.

\bibitem[CK19]{fast-vertex-sparsest}
Julia Chuzhoy and Sanjeev Khanna.
\newblock A new algorithm for decremental single-source shortest paths with
  applications to vertex-capacitated flow and cut problems.
\newblock In {\em STOC 2019, to appear}, 2019.

\bibitem[CMSV17]{CohenMSV17}
Michael~B. Cohen, Aleksander Madry, Piotr Sankowski, and Adrian Vladu.
\newblock Negative-weight shortest paths and unit capacity minimum cost flow in
  ${\tilde{o}} (m^{10/7} log w)$ time.
\newblock In {\em Proceedings of the Twenty-Eighth Annual {ACM-SIAM} Symposium
  on Discrete Algorithms, {SODA} 2017, Barcelona, Spain, Hotel Porta Fira,
  January 16-19}, pages 752--771, 2017.

\bibitem[CQ17]{ChekuriQ17}
Chandra Chekuri and Kent Quanrud.
\newblock Approximating the held-karp bound for metric {TSP} in nearly-linear
  time.
\newblock In {\em 58th {IEEE} Annual Symposium on Foundations of Computer
  Science, {FOCS} 2017, Berkeley, CA, USA, October 15-17, 2017}, pages
  789--800, 2017.

\bibitem[DHZ00]{DorHZ00}
Dorit Dor, Shay Halperin, and Uri Zwick.
\newblock All-pairs almost shortest paths.
\newblock {\em {SIAM} J. Comput.}, 29(5):1740--1759, 2000.

\bibitem[Din06]{Dinitz}
Yefim Dinitz.
\newblock Dinitz' algorithm: The original version and {Even's} version.
\newblock In {\em Theoretical computer science}, pages 218--240. Springer,
  2006.

\bibitem[ES81]{EvenS}
Shimon Even and Yossi Shiloach.
\newblock An on-line edge-deletion problem.
\newblock {\em Journal of the ACM (JACM)}, 28(1):1--4, 1981.

\bibitem[FG19]{ForsterG19}
Sebastian Forster and Gramoz Goranci.
\newblock Dynamic low-stretch trees via dynamic low-diameter decompositions.
\newblock In {\em Proceedings of the 51st Annual {ACM} {SIGACT} Symposium on
  Theory of Computing, {STOC} 2019, Phoenix, AZ, USA, June 23-26, 2019}, pages
  377--388, 2019.

\bibitem[FHL08]{FeigeHL05}
U.~Feige, M.T. Hajiaghayi, and J.R. Lee.
\newblock Improved approximation algorithms for minimum weight vertex
  separators.
\newblock {\em SIAM Journal on Computing}, 38:629--657, 2008.

\bibitem[FHN14a]{HenzingerKN14_focs}
Sebastian Forster, Monika Henzinger, and Danupon Nanongkai.
\newblock Decremental single-source shortest paths on undirected graphs in
  near-linear total update time.
\newblock In {\em 55th {IEEE} Annual Symposium on Foundations of Computer
  Science, {FOCS} 2014, Philadelphia, PA, USA, October 18-21, 2014}, pages
  146--155, 2014.

\bibitem[FHN14b]{HenzingerKN14_soda}
Sebastian Forster, Monika Henzinger, and Danupon Nanongkai.
\newblock A subquadratic-time algorithm for decremental single-source shortest
  paths.
\newblock In {\em Proceedings of the Twenty-Fifth Annual {ACM-SIAM} Symposium
  on Discrete Algorithms, {SODA} 2014, Portland, Oregon, USA, January 5-7,
  2014}, pages 1053--1072, 2014.

\bibitem[FHN16]{henzinger16}
Sebastian Forster, Monika Henzinger, and Danupon Nanongkai.
\newblock Dynamic approximate all-pairs shortest paths: Breaking the {O(mn)}
  barrier and derandomization.
\newblock {\em SIAM Journal on Computing}, 45(3):947--1006, 2016.
\newblock Announced at FOCS'13.

\bibitem[FHNS15]{HenzingerKNS15}
Sebastian Forster, Monika Henzinger, Danupon Nanongkai, and Thatchaphol
  Saranurak.
\newblock Unifying and strengthening hardness for dynamic problems via the
  online matrix-vector multiplication conjecture.
\newblock In {\em Proceedings of the Forty-Seventh Annual {ACM} on Symposium on
  Theory of Computing, {STOC} 2015, Portland, OR, USA, June 14-17, 2015}, pages
  21--30, 2015.

\bibitem[Fle00]{Fleischer00}
Lisa Fleischer.
\newblock Approximating fractional multicommodity flow independent of the
  number of commodities.
\newblock {\em {SIAM} J. Discrete Math.}, 13(4):505--520, 2000.

\bibitem[FLS{\etalchar{+}}18]{FominLSPW18}
Fedor~V. Fomin, Daniel Lokshtanov, Saket Saurabh, Michal Pilipczuk, and Marcin
  Wrochna.
\newblock Fully polynomial-time parameterized computations for graphs and
  matrices of low treewidth.
\newblock {\em {ACM} Trans. Algorithms}, 14(3):34:1--34:45, 2018.

\bibitem[GK96]{GK}
M.~Goemans and J.~Kleinberg.
\newblock An improved approximation ratio for the minimum latency problem.
\newblock {\em Proceedings of the ACM-SIAM Symposium on Discrete Algorithms},
  1996.

\bibitem[GK98]{GK98}
Naveen Garg and Jochen K{\"{o}}nemann.
\newblock Faster and simpler algorithms for multicommodity flow and other
  fractional packing problems.
\newblock In {\em 39th Annual Symposium on Foundations of Computer Science,
  {FOCS} '98, November 8-11, 1998, Palo Alto, California, {USA}}, pages
  300--309, 1998.

\bibitem[GWN20]{GutenbergW20}
Maximilian~Probst Gutenberg and Christian Wulff-Nilsen.
\newblock Deterministic algorithms for decremental approximate shortest paths:
  Faster and simpler.
\newblock In {\em Proceedings of the Fourteenth Annual ACM-SIAM Symposium on
  Discrete Algorithms}, pages 2522--2541. SIAM, 2020.

\bibitem[HdLT01]{dynamic-connectivity}
Jacob Holm, Kristian de~Lichtenberg, and Mikkel Thorup.
\newblock Poly-logarithmic deterministic fully-dynamic algorithms for
  connectivity, minimum spanning tree, 2-edge, and biconnectivity.
\newblock {\em J. ACM}, 48(4):723--760, July 2001.

\bibitem[HK95]{HenzingerKing}
Monika~Rauch Henzinger and Valerie King.
\newblock Fully dynamic biconnectivity and transitive closure.
\newblock In {\em Foundations of Computer Science, 1995. Proceedings., 36th
  Annual Symposium on}, pages 664--672. IEEE, 1995.

\bibitem[IOO18]{IwataOO18}
Yoichi Iwata, Tomoaki Ogasawara, and Naoto Ohsaka.
\newblock On the power of tree-depth for fully polynomial {FPT} algorithms.
\newblock In {\em 35th Symposium on Theoretical Aspects of Computer Science,
  {STACS} 2018, February 28 to March 3, 2018, Caen, France}, pages 41:1--41:14,
  2018.

\bibitem[KKOV07]{KhandekarKOV2007cut}
Rohit Khandekar, Subhash Khot, Lorenzo Orecchia, and Nisheeth~K Vishnoi.
\newblock On a cut-matching game for the sparsest cut problem.
\newblock {\em Univ. California, Berkeley, CA, USA, Tech. Rep.
  UCB/EECS-2007-177}, 6(7):12, 2007.

\bibitem[KMP12]{KelnerMP12}
Jonathan~A. Kelner, Gary~L. Miller, and Richard Peng.
\newblock Faster approximate multicommodity flow using quadratically coupled
  flows.
\newblock In {\em Proceedings of the 44th Symposium on Theory of Computing
  Conference, {STOC} 2012, New York, NY, USA, May 19 - 22, 2012}, pages 1--18,
  2012.

\bibitem[KRV09]{KRV}
Rohit Khandekar, Satish Rao, and Umesh Vazirani.
\newblock Graph partitioning using single commodity flows.
\newblock {\em Journal of the ACM (JACM)}, 56(4):19, 2009.

\bibitem[KT19]{KawarabayashiT19}
Ken{-}ichi Kawarabayashi and Mikkel Thorup.
\newblock Deterministic edge connectivity in near-linear time.
\newblock {\em J. {ACM}}, 66(1):4:1--4:50, 2019.

\bibitem[LS14]{LeeS14}
Yin~Tat Lee and Aaron Sidford.
\newblock Path finding methods for linear programming: Solving linear programs
  in {\~{o}}(vrank) iterations and faster algorithms for maximum flow.
\newblock In {\em 55th {IEEE} Annual Symposium on Foundations of Computer
  Science, {FOCS} 2014, Philadelphia, PA, USA, October 18-21, 2014}, pages
  424--433, 2014.

\bibitem[Mad10]{Madry10_stoc}
Aleksander Madry.
\newblock Faster approximation schemes for fractional multicommodity flow
  problems via dynamic graph algorithms.
\newblock In {\em Proceedings of the 42nd {ACM} Symposium on Theory of
  Computing, {STOC} 2010, Cambridge, Massachusetts, USA, 5-8 June 2010}, pages
  121--130, 2010.

\bibitem[NS17]{Saranurak}
Danupon Nanongkai and Thatchaphol Saranurak.
\newblock Dynamic spanning forest with worst-case update time: adaptive, las
  vegas, and {O}(n\({}^{\mbox{1/2 - {\(\epsilon\)}}}\))-time.
\newblock In {\em Proceedings of the 49th Annual {ACM} {SIGACT} Symposium on
  Theory of Computing, {STOC} 2017, Montreal, QC, Canada, June 19-23, 2017},
  pages 1122--1129, 2017.

\bibitem[NSW17]{NanongkaiSW17}
Danupon Nanongkai, Thatchaphol Saranurak, and Christian Wulff{-}Nilsen.
\newblock Dynamic minimum spanning forest with subpolynomial worst-case update
  time.
\newblock In {\em 58th {IEEE} Annual Symposium on Foundations of Computer
  Science, {FOCS} 2017, Berkeley, CA, USA, October 15-17, 2017}, pages
  950--961, 2017.

\bibitem[RS95]{flat-wall-RS}
Neil Robertson and Paul~D Seymour.
\newblock Graph minors. {XIII}. the disjoint paths problem.
\newblock {\em Journal of Combinatorial Theory, Series B}, 63(1):65--110, 1995.

\bibitem[RZ11]{RodittyZ11}
Liam Roditty and Uri Zwick.
\newblock On dynamic shortest paths problems.
\newblock {\em Algorithmica}, 61(2):389--401, 2011.

\bibitem[RZ12]{rodittyZ2}
Liam Roditty and Uri Zwick.
\newblock Dynamic approximate all-pairs shortest paths in undirected graphs.
\newblock {\em SIAM Journal on Computing}, 41(3):670--683, 2012.

\bibitem[San05]{APSPfully4}
Piotr Sankowski.
\newblock Subquadratic algorithm for dynamic shortest distances.
\newblock In {\em International Computing and Combinatorics Conference}, pages
  461--470. Springer, 2005.

\bibitem[She17]{Sherman17}
Jonah Sherman.
\newblock Area-convexity, l\({}_{\mbox{{\(\infty\)}}}\) regularization, and
  undirected multicommodity flow.
\newblock In {\em Proceedings of the 49th Annual {ACM} {SIGACT} Symposium on
  Theory of Computing, {STOC} 2017, Montreal, QC, Canada, June 19-23, 2017},
  pages 452--460, 2017.

\bibitem[ST83]{SleatorT83}
Daniel~Dominic Sleator and Robert~Endre Tarjan.
\newblock A data structure for dynamic trees.
\newblock {\em J. Comput. Syst. Sci.}, 26(3):362--391, 1983.

\bibitem[SW19]{expander-pruning}
Thatchaphol Saranurak and Di~Wang.
\newblock Expander decomposition and pruning: Faster, stronger, and simpler.
\newblock In {\em Proceedings of the Thirtieth Annual {ACM-SIAM} Symposium on
  Discrete Algorithms, {SODA} 2019, San Diego, California, USA, January 6-9,
  2019}, pages 2616--2635, 2019.

\bibitem[TZ01]{TZ}
M.~Thorup and U.~Zwick.
\newblock Approximate distance oracles.
\newblock {\em Annual ACM Symposium on Theory of Computing}, 2001.

\bibitem[vdBNS19]{BrandNS19}
Jan van~den Brand, Danupon Nanongkai, and Thatchaphol Saranurak.
\newblock Dynamic matrix inverse: Improved algorithms and matching conditional
  lower bounds.
\newblock In {\em 60th {IEEE} Annual Symposium on Foundations of Computer
  Science, {FOCS} 2019, Baltimore, Maryland, USA, November 9-12, 2019}, pages
  456--480, 2019.

\bibitem[Waj20]{Wajc20}
David Wajc.
\newblock Rounding dynamic matchings against an adaptive adversary.
\newblock In {\em Proccedings of the 52nd Annual {ACM} {SIGACT} Symposium on
  Theory of Computing, {STOC} 2020, Chicago, IL, USA, June 22-26, 2020}, pages
  194--207, 2020.

\bibitem[WN17]{dynamic-spanning-forest}
Christian Wulff-Nilsen.
\newblock Fully-dynamic minimum spanning forest with improved worst-case update
  time.
\newblock In {\em Proceedings of the 49th Annual ACM SIGACT Symposium on Theory
  of Computing}, pages 1130--1143. ACM, 2017.
\newblock Full version at arXiv:1611.02864.

\bibitem[Zwi98]{Zwick98}
Uri Zwick.
\newblock All pairs shortest paths in weighted directed graphs-exact and almost
  exact algorithms.
\newblock In {\em Proceedings 39th Annual Symposium on Foundations of Computer
  Science (Cat. No. 98CB36280)}, pages 310--319. IEEE, 1998.

\end{thebibliography}
